\PassOptionsToPackage{hyphens}{url}
\PassOptionsToPackage{dvipsnames,svgnames,x11names}{xcolor}
\documentclass[
  12pt]{article}
\usepackage{tikz}
\usepackage{amsthm}
\usepackage{geometry}
\usepackage{setspace}
\usepackage{multirow}
\usepackage{arxiv}
\usepackage{enumitem}

\usepackage{amsmath}
\usepackage{caption}

\newcommand{\calN}{\mathcal{N}}

\def\beqr{\begin{eqnarray}}
\def\eeqr{\end{eqnarray}}
\def\beqrs{\begin{eqnarray*}}
\def\eeqrs{\end{eqnarray*}}

\def\wh{\widehat}
\def\wt{\widetilde}

\newcommand{\todistribution}{\xrightarrow{\text{d}}}

\newcommand{\defn}{\stackrel{\mbox{{\tiny def}}}{=}}

\def\cov{\mathrm{cov}}

\usepackage{xcolor}

\usepackage{bm}

\newcommand{\wb}{\overline}
\newtheorem{assumption}{Assumption}
\newtheorem{Lemma}{Lemma}
\newtheorem{corollary}{Corollary}
\usepackage{comment}
\newtheorem{theorem}{Theorem}
\usepackage[export]{adjustbox}

\usepackage{epsfig}
\usepackage{threeparttable}
\graphicspath{{Images/}}
\usepackage{xr}

\numberwithin{equation}{section}
\usepackage{amsmath,amssymb}
\usepackage{iftex}
\ifPDFTeX
  \usepackage[T1]{fontenc}
  \usepackage[utf8]{inputenc}
  \usepackage{textcomp} % provide euro and other symbols
\else % if luatex or xetex
  \usepackage{unicode-math}
  \defaultfontfeatures{Scale=MatchLowercase}
  \defaultfontfeatures[\rmfamily]{Ligatures=TeX,Scale=1}
\fi
\usepackage{lmodern}
\ifPDFTeX\else  
\fi
\IfFileExists{upquote.sty}{\usepackage{upquote}}{}
\IfFileExists{microtype.sty}{% use microtype if available
  \usepackage[]{microtype}
  \UseMicrotypeSet[protrusion]{basicmath} % disable protrusion for tt fonts
}{}
\makeatletter
\@ifundefined{KOMAClassName}{% if non-KOMA class
  \IfFileExists{parskip.sty}{%
    \usepackage{parskip}
  }{% else
    \setlength{\parindent}{0pt}
    \setlength{\parskip}{6pt plus 2pt minus 1pt}}
}{% if KOMA class
  \KOMAoptions{parskip=half}}
\makeatother
\usepackage{xcolor}
\makeatletter
\ifx\paragraph\undefined\else
  \let\oldparagraph\paragraph
  \renewcommand{\paragraph}{
    \@ifstar
      \xxxParagraphStar
      \xxxParagraphNoStar
  }
  \newcommand{\xxxParagraphStar}[1]{\oldparagraph*{#1}\mbox{}}
  \newcommand{\xxxParagraphNoStar}[1]{\oldparagraph{#1}\mbox{}}
\fi
\ifx\subparagraph\undefined\else
  \let\oldsubparagraph\subparagraph
  \renewcommand{\subparagraph}{
    \@ifstar
      \xxxSubParagraphStar
      \xxxSubParagraphNoStar
  }
  \newcommand{\xxxSubParagraphStar}[1]{\oldsubparagraph*{#1}\mbox{}}
  \newcommand{\xxxSubParagraphNoStar}[1]{\oldsubparagraph{#1}\mbox{}}
\fi
\makeatother

\usepackage{longtable,booktabs,array}
\usepackage{calc} % for calculating minipage widths
\usepackage{etoolbox}
\makeatletter
\patchcmd\longtable{\par}{\if@noskipsec\mbox{}\fi\par}{}{}
\makeatother
\IfFileExists{footnotehyper.sty}{\usepackage{footnotehyper}}{\usepackage{footnote}}
\makesavenoteenv{longtable}
\usepackage{graphicx}
\makeatletter
\def\maxwidth{\ifdim\Gin@nat@width>\linewidth\linewidth\else\Gin@nat@width\fi}
\def\maxheight{\ifdim\Gin@nat@height>\textheight\textheight\else\Gin@nat@height\fi}
\makeatother
\setkeys{Gin}{width=\maxwidth,height=\maxheight,keepaspectratio}
\makeatletter
\def\fps@figure{htbp}
\makeatother

\makeatletter
\@ifpackageloaded{caption}{}{\usepackage{caption}}
\AtBeginDocument{%
\ifdefined\contentsname
  \renewcommand*\contentsname{Table of contents}
\else
  \newcommand\contentsname{Table of contents}
\fi
\ifdefined\listfigurename
  \renewcommand*\listfigurename{List of Figures}
\else
  \newcommand\listfigurename{List of Figures}
\fi
\ifdefined\listtablename
  \renewcommand*\listtablename{List of Tables}
\else
  \newcommand\listtablename{List of Tables}
\fi
\ifdefined\figurename
  \renewcommand*\figurename{Figure}
\else
  \newcommand\figurename{Figure}
\fi
\ifdefined\tablename
  \renewcommand*\tablename{Table}
\else
  \newcommand\tablename{Table}
\fi
}
\@ifpackageloaded{float}{}{\usepackage{float}}
\floatstyle{ruled}
\@ifundefined{c@chapter}{\newfloat{codelisting}{h}{lop}}{\newfloat{codelisting}{h}{lop}[chapter]}
\floatname{codelisting}{Listing}

\makeatother
\makeatletter
\@ifpackageloaded{caption}{}{\usepackage{caption}}
\@ifpackageloaded{subcaption}{}{\usepackage{subcaption}}
\makeatother
\ifLuaTeX
  \usepackage{selnolig}  % disable illegal ligatures
\fi
\usepackage[]{natbib}
\usepackage{bookmark}

\IfFileExists{xurl.sty}{\usepackage{xurl}}{} % add URL line breaks if available
\hypersetup{
  pdftitle={Title},
  pdfauthor={Author 1; Author 2},
  pdfkeywords={3 to 6 keywords, that do not appear in the title},
  colorlinks=true,
  linkcolor={blue},
  filecolor={Maroon},
  citecolor={Blue},
  urlcolor={Blue},
  pdfcreator={LaTeX via pandoc}}
\title{\bf Identification, Estimation, and Inference for Sequential Causally Ordered Mediation Pathways} 
  \author{Ritoban Kundu$^\dagger$, Canyi Chen$^\dagger$, and Peter X.K. Song\textsuperscript{*}\\\textit{Department of Biostatistics, University of Michigan, Ann Arbor, U.S.A.}\\
  $^\dagger$ Co-first Author, \textsuperscript{*} Corresponding Author}

\newcommand\independent{\protect\mathpalette{\protect\independenT}{\perp}}
\def\independenT#1#2{\mathrel{\rlap{$#1#2$}\mkern2mu{#1#2}}}

\newcommand \pxs[1] {{\color{red}{#1}}}

\usepackage{cleveref}
\crefname{theo}{Theorem}{Theorems}
\crefname{lemm}{Lemma}{Lemmas}
\crefname{coro}{Corollary}{Corollaries}
\crefname{prop}{Proposition}{Propositions}
\counterwithout{equation}{section}
\begin{document}

\maketitle
\bigskip
\begin{abstract}
\noindent
Mediation analysis plays an essential role in uncovering the mechanisms by which an exposure influences an outcome through intermediate pathways. While methodological advances for single-mediator settings are well established, rigorous tools for handling multiple, sequentially ordered mediators remain underdeveloped. Such settings are common in applications like longitudinal cohort studies, where exposures operate through complex chains of mediators over time. In this paper, we establish a general framework for sequentially ordered mediators that enables the identification and formal decomposition of the total effect into component path-specific effects. We also develop estimation procedures for mediation estimands with both continuous and categorical outcomes. Furthermore, we introduce a new testing strategy to conduct inference using a studentized statistic combined with data-splitting. This approach achieves valid Type I error control under the composite null across diverse data-generating mechanisms. Through extensive simulations and applications to two large-scale empirical studies, we demonstrate that the proposed methodology provides reliable estimation, valid inference, and improved power for discovering novel mediation pathways.
\end{abstract}

\noindent%
\vfill

\doublespacing

\section{Introduction}\label{sec-intro}
Mediation analysis involving sequential causally causally mediators (SCOMs) is a critical tool for providing systematic insights into the complex pathways underlying scientific phenomena. One example of this kind arising from environmental health studies pertains to the origin of disease, under the so-called \emph{Developmental Origins of Health and Disease (DOHaD)} hypothesis: early-life exposures (e.g. toxicant exposure during pregnancy) may influence individual health trajectories across the lifespan; see, e.g., Figure \ref{fig:dohad}, which serves as a running example. These trajectories are defined by interdependent SCOM pathways that reflect specific developmental windows \citep{heckman2006case,halfon2014lifecourse}. Many biomedical studies examine these phases to elucidate the mechanisms underlying the transition from health to chronic disease, which typically involve a progression of physiological alterations rather than a single isolated change. Figure \ref{fig:dohad} (top) depicts this journey from the embryonic stage through birth and childhood to adulthood. Similar structural dependencies are observed in other domains, such as biomarker pathways in systems biology and sequential policy interventions in accident prevention. Within these fields, many longitudinal studies collect time-course data where a mediator is measured repeatedly; such temporal sequences represent a primary and highly relevant instance of the SCOM framework.

The growth and developmental trajectory in Figure~\ref{fig:dohad} (top) is often translated and formulated into a directed acyclic graph (DAG) shown in Figure \ref{fig:dohad} (bottom). For ease of exposition, we consider a simple case of two SCOMs, \emph{say}, birth weight ($M_1$) at age 0 and a serum lipidomic biomarker recorded sequentially over adolescent period ($M_2$) from infancy to youth. Prenatal exposure ($S$) to an environmental toxicant, such as lead measured from maternal blood samples, has been associated with physiological alterations in growth and development, which can be quantified by body mass index (BMI) ($Y$) \citep{wang2019AssociationMaternalExposure}. In this scientific investigation, multiple pathways arise naturally to explain the influence of exposure ($S$) (i.e. the at-origin cause disease) on outcome of interest ($Y$), including $S\rightarrow M_1 \rightarrow Y$, $S\rightarrow M_2 \rightarrow Y$ and $S\rightarrow M_1 \rightarrow M_2 \rightarrow Y$ shown in Figure \ref{fig:dohad} (bottom).

In the current literature of mediation analysis, a DAG with a single mediator has been extensively investigated because it is conceptually simple, analytically manageable, and computationally feasible. Recently, extending such simplest mediation analysis setting to include additional mediators, has made several important advances in identifying path-specific effects in the scenarios of SCOMs \citep{daniel2015causal, lin2017interventional, zhou2022semiparametric}. However, existing work has been largely focused on estimation, whereas rigorous methods of
hypothesis testing for the presence of mediation pathways in SCOMs remain almost absent. As a matter of fact, even in the single-mediator case (i.e. $M_2$ being absent in Figure \ref{fig:dohad} (bottom)), methods for hypothesis testing on mediation effects pose significant technical challenges. Standard tests such as Sobel’s test and the MaxP test are known to overly conservative in terms of Type-I error control due to the co-existence of two different distributions of the test statistics under the null hypothesis \citep{mackinnon2002comparison, fritz2007RequiredSampleSize, barfield2017testing}. Developing a valid and powerful inferential framework for the broader SCOM setting is a primary objective of this paper.

The contributions of this paper are threefold. First, we establish a rigorous identification framework for SCOMs that enables the valid decomposition of the total effect into sequentially ordered mediation components. This unified framework accommodates continuous, binary, and count outcomes. Second, we develop maximum likelihood estimation procedures and derive the joint asymptotic properties of the resulting estimators across these diverse outcome types. Third, we propose the Sequentially Ordered Mediation Effect Test (SOMET). Inspired by the classical Sobel test \citep{sobel1982AsymptoticConfidenceIntervals} and utilizing a data-splitting approach, SOMET is both computationally tractable and theoretically grounded. We demonstrate, both analytically and through extensive simulations, that SOMET maintains valid type-I control without prior knowledge of the underlying null configuration. The remainder of this paper is organized as follows. Section \ref{section:identify} defines the general problem and proposes a Generalized Structural Equation Model (G-SEM) under the $\text{SCOM}(K)$ setup. Section \ref{sec:cont} introduces the non-parametric identification assumptions and the general identification theorem, followed by the derivation of path-specific mediation effects under linear structural models. Section \ref{sec:cat} extends this methodology to categorical outcomes via Generalized Linear Model (GLM) specifications. In Section \ref{sec:est}, we develop the maximum likelihood estimation framework and establish the asymptotic distribution of the path-specific estimators. Section \ref{section:test} introduces the SOMET procedure and provides its theoretical justification. Section \ref{section:simulation} presents comprehensive simulation studies evaluating the finite-sample performance of our estimators and the SOMET procedure across diverse outcome types and null configurations. Finally, Section \ref{section:real_data} illustrates the practical utility of our method using the NIH All of Us EHR data to study life-course risk factors for diabetes onset and the ELEMENT cohort study to assess prenatal influences on adolescent obesity. 

\section{Formulation}\label{section:identify}
This section introduces the background, structural equation modeling (SEM) approach, and the potential outcome paradigm.

\subsection{Setup}
Let \( S \) denote an exposure and \( Y \) an outcome of interest, where both variables may be continuous or categorical. For any \( K \geq 1 \), consider an ordered sequence of mediators denoted by \( \overline{\boldsymbol{M}}_K = (M_1, M_2, \dots, M_K) \), where the ascending order of subscripts reflects the assumed causal ordering, namely \( M_1 \rightarrow M_2 \rightarrow \cdots \rightarrow M_K \). Such ordering may naturally arise from study designs in certain scientific contexts, such as human growth and development, in Figure \ref{fig:dohad}. We refer to the sequence of mediators as \( \text{SCOM}(K) \), which corresponds to the setting of \textit{Sequential Causally Ordered Mediators of order $K$}. For instance, for $K = 1$, \( \overline{\boldsymbol{M}}_1 = (M_1) \), and for $K = 2$, \( \overline{\boldsymbol{M}}_2 = (M_1, M_2) \), and so on. Let \( \boldsymbol{W} \) denote the set of baseline covariates, that may confound the relationships in the DAG between the exposure, SCOM($K$) and outcome. The central goal is to assess whether, and to what extent, the influence of the exposure $S$ on the outcome $Y$ may be mediated through the SCOM($K$) \( \overline{\boldsymbol{M}}_K \).

Given SCOM($K$) with $\overline{\boldsymbol M}_K$, we introduce a general structural equation modeling (G-SEM) approach. The G-SEM is a hierarchical system of models, consisting of conditional models $[S|\boldsymbol W]$, $ [M_1|S,\boldsymbol W],\cdots,[M_K|M_{K-1},\cdots, M_1,S,\boldsymbol W]$, and $[Y|M_{K},\cdots, M_1,S,\boldsymbol W]$ specified as follows:
\begin{align}
\left\{
\begin{aligned}
    S &= h_S(\boldsymbol{W}, \epsilon_S), \\
    M_1 &= h_{M_1}(S, \boldsymbol{W}, \epsilon_{M_1}),M_2= h_{M_2}(S,\overline{M}_1, \boldsymbol{W}, \epsilon_{M_2}) \dots, M_K = h_{M_K}(S, \overline{\boldsymbol{M}}_{K-1}, \boldsymbol{W}, \epsilon_{M_K}), \\
    Y &= h_Y(S, \overline{\boldsymbol{M}}_{K}, \boldsymbol{W}, \epsilon_Y),\label{eq:npsem}
\end{aligned}
\right.
\end{align}
%till here
where in the spirit of a Markov chain, each mediator $M_k$ is modeled as a function of the preceding mediators  $\overline{\boldsymbol{M}}_{k-1}=(M_1,M_2,\cdots,M_{k-1})$ to capture their sequential evolutions in the DAG, and the error terms $\epsilon_S,\{\epsilon_{M_k}\}_{k=1}^K$ and $\epsilon_Y$ are mutually independent. The link functions $h_S$, $\{h_{M_k}\}_{k=1}^K$, and $h_Y$ regulate the respective underlying data-generating processes and will be specified according to certain model types with respect to particular data scenarios. The above G-SEM formulation \eqref{eq:npsem} can accommodate both continuous and categorical variables for $S$ and $Y$, while all mediators $M_1,M_2,\cdots,M_K$ are confined as continuous variables in this paper.

\subsection{Potential Outcomes Paradigm}
In this section, we utilize the potential outcomes paradigm to formally define the key causal estimands within the SCOM$(K)$ setting. For any $s \in \text{supp}(S)$ and $\overline{\boldsymbol{m}}_K = (m_1, \dots, m_K) \in \text{supp}(\overline{\boldsymbol{M}}_K)$, let $Y(s, \overline{\boldsymbol{m}}_K)$ denote the potential outcome that would be observed if the exposure $S$ were set to $s$ and the mediators were jointly intervened upon such that $\overline{\boldsymbol{M}}_K = \overline{\boldsymbol{m}}_K$. Here, $\text{supp}(\overline{\boldsymbol{M}}_K) := \prod_{k=1}^K \text{supp}(M_k)$, where $\text{supp}(X)$ denotes the support of the random variable $X$. The potential values of the SCOMs $\overline{\boldsymbol{M}}_K$ satisfy the prespecified sequential causal ordering defined in equation \eqref{eq:npsem}. In the example of two mediators, the potential value of the first mediator is \( M_1(s) \), and the second is \( M_2(s, m_1) \). Generally, the potential value of the \( k \)th mediator, $M_k$ is denoted by the counterfactual, \( M_k(s, \overline{\boldsymbol{m}}_{k-1}) \), where \( k = 1, \ldots, K \), reflecting it's dependence on the values of both exposure $s$ and all preceding mediators $(m_1,\cdots,m_{k-1})$.

\par
In the context of SCOM($K$), nested counterfactuals arise naturally and are defined as follows. Let $\boldsymbol{s} = (s_1, \ldots, s_K) \in \text{supp}(S)^{K}$ be a $K$-element vector of exposure values. Based on the G-SEM in \eqref{eq:npsem}, we define the recursive potential outcomes for the mediators, following their pre-specified ordering,
\begin{align}
\left\{
\begin{aligned}
    & \overline{M}_1(\boldsymbol{s}) = M_1(s_1), \hspace{0.2cm} \overline{M}_2(\boldsymbol{s}) = \left(\overline{M}_1(\boldsymbol{s}), M_2(s_2, \overline{M}_1(\boldsymbol{s}))\right), \\
    & \hspace{0.2in} \vdots \\
    & \overline{M}_K(\boldsymbol{s}) = \left(\overline{M}_{K-1}(\boldsymbol{s}), M_K(s_K, \overline{M}_{K-1}(\boldsymbol{s}))\right) \cdot \label{eq:eqM}
\end{aligned}
\right.
\end{align}
\noindent
\textbf{Remark.} In definition \eqref{eq:eqM} for the nested $\text{SCOM}(K)$ potential outcomes, each exposure component $s_k$ is strictly tied to its corresponding mediator $M_k$ across all recursive levels. For example, in the vector $\overline{M}_2(\boldsymbol{s}) = \left(M_1(s_1), M_2(s_2, M_1(s_1))\right)$, the value $s_1$ is consistently assigned to all occurences of $M_1$, regardless of whether $M_1$ appears as a primary term or as an argument within $M_2$. Similarly, in $\overline{M}_3(\boldsymbol{s}) = \left(\overline{M}_2(\boldsymbol{s}), M_3(s_3, \overline{M}_2(\boldsymbol{s}))\right)$, the value $s_2$ is fixed for all occurrences of $M_2$. This indexing ensures that the $k$-th exposure value $s_k$ dictates the potential outcome of the $k$-th mediator $M_k$ throughout the entire chain. The rationale for this restriction is discussed in detail in Section~\ref{sec:twomed}.

For any $s_0\in \text{supp}(S)$ and $\boldsymbol s=(s_1,s_2,\cdots,s_K)\in \text{supp}(S)^K$, the counterfactual outcome is defined by \( Y(s_0,\boldsymbol s) := Y(s_0, \overline{M}_K(\boldsymbol{s})) \). For notational convenience for any $s \in \text{supp}(S)$ we denote \( Y(s) := Y(s,\underbrace{s,\cdots,s}_{K \hspace{0,2cm}\text{times}}) \) at a single exposure value $s$. Similarly, we denote \( M_k(s) := M_k(s, \overline{M}_{k-1}(s)) \) at a single exposure value $s$, for \( k = 2, \ldots, K \).

\section{Methodology for Continuous Outcomes} \label{sec:cont}
We now formally defines the causal estimands for continuous outcomes within the SCOM(K) framework, and establish the key identifiability conditions and the associated formulas.

\subsection{Total Effect}
For any \( s\in \text{supp}(S)\) and \(s'\in \text{supp}(S)  \), when exposure $S$ varies from \( s \) to \( s' \), the \textit{total effect} on continuous outcome \( Y \) is given by
\begin{align}
    \text{TE}(s',s) = \mathbb{E}\{Y(s')-Y(s)\}=\mathbb{E}\{Y(s')\} - \mathbb{E}\{Y(s)\},\label{eq:toteff}
\end{align}
where the expectation is operated under the joint distribution of the counterfactual continuous outcomes $\{Y(s'),Y(s)\}$.
This quantity captures the overall difference in the expected potential outcomes when exposure $S$ is set at \( s' \) versus \( s \), which embraces all possible direct and indirect pathways through which \( S \) influences \( Y \). To grasp the potential combinatorial complexity in possible pathways, Section \ref{sec:twomed} presents a simple but illustrative case with two SCOMs shown in Figure \ref{fig:dohad}.

\subsection{Causal pathway decomposition under SCOM(2)}\label{sec:twomed}
Utilizing the SCOM(2) setting with $M_1 \rightarrow M_2$, we demonstrate how the total effect can be decomposed into distinct causal pathways. This decomposition characterizes the distributed influence of exposure $S$ on outcome $Y$, distinguishing between direct and indirect effects. In this case, under the G-SEM in \eqref{eq:npsem},
the total effect in equation \eqref{eq:toteff} is decomposed into the following four distinct path-specific effects.
\begin{itemize}
    \item \textit{Path-1}: Pathway $(S\to Y)$ represents the \textit{direct effect} denoted by 
$\tau_1(S\to Y)$, reflecting the influence of $S$ on $Y$ without 
passing through any mediators
$$\tau_1(S\to Y)_{(s,s')}: \mathbb{E}\{Y(s^\prime, M_1(s^\prime), M_2(s^\prime, M_1(s^\prime))\} - \mathbb{E}\{Y(s, M_1(s^\prime), M_2(s^\prime, M_1(s^\prime))\} \cdot$$
 \item \textit{Path-2}: Pathway $(S\to M_1\to Y)$ pertains to the \textit{single-mediator effect} denoted by $\tau_2(S\to M_1\to Y)$ via only the first mediator \( M_1 \), where \( S \) first affects \( M_1 \) and then influences \( Y \):
$$\tau_2(S\to M_1\to Y)_{(s,s')}: \mathbb{E}\{Y(s, M_1(s^\prime), M_2(s^\prime, M_1(s^\prime))\} - \mathbb{E}\{Y(s, M_1(s), M_2(s^\prime, M_1(s^\prime))\} \cdot$$
 \item \textit{Path-3}: Pathway $(S\to M_2\to Y)$ concerns the \textit{single-mediator effect} denoted by $\tau_3(S\to M_2\to Y)$ via only the second mediator \( M_2 \), in that \( S \) first impacts \( M_2 \), and then affects \( Y \) :
$$\tau_3(S\to M_2\to Y)_{(s,s')}:\mathbb{E}\{Y(s, M_1(s), M_2(s^\prime, M_1(s^\prime))\} - \mathbb{E}\{Y(s, M_1(s), M_2(s, M_1(s^\prime))\}\cdot$$

\item \textit{Path-4}: Pathway $(S\to M_1\to M_2\to Y)$ corresponds to the \textit{sequential mediation effect} denoted by $\tau_4(S\to M_1\to M_2\to Y)$, where \( S \) first influences \( M_1 \), which then affects \( M_2 \), ultimately impacts \( Y \):
$$\tau_4(S\to M_1\to M_2\to Y)_{(s,s')}:\mathbb{E}\{Y(s, M_1(s), M_2(s, M_1(s^\prime))\} - \mathbb{E}\{Y(s, M_1(s), M_2(s, M_1(s))\}\cdot$$
\end{itemize}
To evaluate these four mediation effects $\tau_1$ to $\tau_4$, it suffices to estimate a generic expectation,
\begin{equation}
    \Gamma(s_0,s_1,s_2,s_3) = \mathbb{E}\left\{Y\left(s_0, M_1(s_1), M_2(s_2, M_1(s_3))\right)\right\},\hspace{0.1cm} \text{where}\hspace{0.1cm} s_0,s_1,s_2,s_3 \in \text{Supp}(S) \label{eq:identify}
\end{equation}
To proceed, we impose the following identifiability conditions in the setting of $\text{SCOM}(2)$.
\begin{assumption}[Consistency]\label{assumption:consistency_two_mediators}
The consistency assumption states that the observed values are equal to their corresponding potential outcomes: $
Y = Y(s, m_1, m_2),M_2 = M_2(s, m_1)$ and $M_1 = M_1(s),$
with \( s \), \( m_1 \), and \( m_2 \) being the observed realizations of variables \( S \), \( M_1 \), and \( M_2 \), respectively.
\end{assumption}

\begin{assumption}[Sequential Ignorability]\label{assumption:identification_two_mediators}
For any $s,s'\in \text{supp}(S)$, $m_1,m_1'\in \text{supp}(M_1)$ and $m_2\in \text{supp}(M_2)$, the following conditions hold:
  \begin{enumerate}
      \item[(i)] $Y(s,m_1,m_2)\perp\!\!\!\perp S|\boldsymbol W$, $Y(s,m_1,m_2)\perp\!\!\!\perp M_1|S,\boldsymbol W$, and $Y(s,m_1,m_2)\perp\!\!\!\perp M_2|M_1,S,\boldsymbol W$;
      \item[(ii)] $Y(s,m_1,m_2)\perp\!\!\!\perp M_1(s')|\boldsymbol W$, $Y(s,m_1,m_2)\perp\!\!\!\perp M_2(s',m_1')|\boldsymbol W$ and $M_2(s,m_1)\perp\!\!\!\perp M_1(s')|\boldsymbol W$,
  \end{enumerate}
\end{assumption}
where $\boldsymbol W$ is the vector of confounding factors. For any three vectors \( U,V \) and $Z$, the notation \( U \perp\!\!\!\perp V |Z\) denotes that \( U \) is independent of \( V \) conditional on $Z$. Assumption \ref{assumption:identification_two_mediators}(i) ensures the absence of unmeasured confounding for the exposure-outcome and mediator-outcome relationships. Assumption \ref{assumption:identification_two_mediators}(ii) further stipulates that there are no exposure-induced confounders of the mediator-outcome relationship; that is, no factors affected by the exposure $S$ confound the association between the mediators $(M_1, M_2)$ and the outcome $Y$.
\par
It is important to note that under Assumptions \ref{assumption:consistency_two_mediators}--\ref{assumption:identification_two_mediators}, the counterfactual expectation $\Gamma(s_0, s_1, s_2, s_3)$ in \eqref{eq:identify} is identifiable if and only if $s_1 = s_3$ \citep{avin2005identifiability}. Scenarios where $s_1 \neq s_3$ necessitate a restrictive ``cross-world" independence assumption, specifically, $M_1(s_1) \perp\!\!\!\perp M_1(s_3) \mid \boldsymbol{W}$ which rarely holds in practical settings. Consequently, without this additional constraint, the specific mediation pathways $\tau_2(S \to M_1 \to Y)_{(s,s')}$ and $\tau_4(S \to M_1 \to M_2 \to Y)_{(s,s')}$ are not individually identifiable under the above assumptions \ref{assumption:consistency_two_mediators}--\ref{assumption:identification_two_mediators}.
\par
To circumvent such an identifiability issue, we adopt the means of a composite pathway \( S \to M_1 \leadsto Y \)\citep{zhou2022semiparametric}, which collapses the two individual pathways \( S \to M_1 \to Y \) and \( S \to M_1 \to M_2 \to Y \) into one set. This union enables the capture of a set of relevant causal mediation routes from \( S \) to \( Y \) that are all mediated by \( M_1 \), regardless of contributions by \( M_2 \). This results in a new type of causal estimand arising from a composite pathway given as follows: 
\begin{align*}
    \tau_1^{\dag}(S \to M_1 \leadsto Y)_{(s,s')}&:=\tau_2(S\to M_1\to Y)_{(s,s')}+\tau_4(S\to M_1\to M_2\to Y)_{(s,s')}\\
    &= \mathbb{E}\{Y(s, M_1(s^\prime), M_2(s^\prime, M_1(s^\prime)))\} - \mathbb{E}\{Y(s, M_1(s), M_2(s^\prime, M_1(s^\prime)))\}\cdot
\end{align*}
Under Assumptions \ref{assumption:consistency_two_mediators}--\ref{assumption:identification_two_mediators} above, these three effects \( \tau_1(S \to Y) \), \( \tau_1^{\dag}(S \to M_1 \leadsto Y) \), and \( \tau_3(S \to M_2 \to Y) \) are identifiable, which are the targeted unknowns for estimation and inference in this paper. 

\subsection{Causal pathway decomposition for $\text{SCOM}(K)$}\label{sec:contdef}
We begin with the following notation. Generalizing the basic case of DAG under SCOM(2), we now formalize the definitions of causal estimands, establish identifiability conditions, and derive identification formulae for the general class of DAGs under $\text{SCOM}(K)$. For any $s_0 \in \text{supp}(S)$ and $\boldsymbol{s} = (s_1, \dots, s_K) \in \text{supp}(S)^K$, let $\Gamma(s_0, \boldsymbol{s}) = \mathbb{E}[Y(s_0, \overline{M}_K(\boldsymbol{s}))]$ denote the expected potential outcome when the exposure $S$ is set to $s_0$ and the series of $K$ mediators, $\overline{M}_K$ are set to their nested potential values under the series of exposure vector given by $\boldsymbol{s}= (s_1, \dots, s_K)$. Next, we decompose the total effect within the SCOM($K$) framework into $K+1$ identifiable path-specific effects (PSEs). This decomposition is not unique, because a directed path from treatment to outcome can pass through various combinations of the $K$ mediators, and PSEs can be defined in a multitude of ways depending on which other paths are active. Given the causal ordering, the resulting partition is fundamentally order-dependent. The following decomposition represents one such valid partition.
\begin{itemize}
    \item First, the direct effect $\tau_1(S\to Y)_{(s,s')}=\Gamma(s',\boldsymbol{s}')-\Gamma(s,\boldsymbol{s}')$, where $\boldsymbol {s}' \in \text{supp}(S)^{K} $ with all elements equal to $s'$, namely $\boldsymbol {s}'=(\underbrace{s',\cdots,s'}_{K \hspace{0,2cm}\text{times}})$.
    \item Second, for any \( k \geq 1 \), the composite mediator effect $\tau_k^{\dag}$ through \( M_k \), denoted as 
  $\tau_k^{\dag}(S \to M_k \leadsto Y)_{(s,s')}=\Gamma(s, \boldsymbol{s}_{1k}) - \Gamma(s, \boldsymbol{s}_{2k}),$ where \( \boldsymbol{s}_{1k} \) is a $k-$ element vector whose first \( (k-1) \) elements are all equal to \( s \) and its remaining \( (K-k+1) \) elements are all equal to \( s' \), namely $\boldsymbol s_{1k}=(\underbrace{s,\cdots,s}_{(k-1) \hspace{0.1cm}\text{times}},\underbrace{s',\cdots,s'}_{(K-k+1) \hspace{0.1cm}\text{times}})$. Similarly, the $k-$ element vector \( \boldsymbol{s}_{2k} \) has its first \( k \) elements as \( s \) and the remaining \( (K-k) \) elements as \( s' \), namely $s_{2k}=(\underbrace{s,\cdots,s}_{k \hspace{0.1cm}\text{times}},\underbrace{s',\cdots,s'}_{(K-k) \hspace{0.1cm}\text{times}})$.
\end{itemize}

\subsection{Identification of Causal Mediation Estimands}
By applying the law of iterated expectations, the generic expectation $\Gamma(s_0, \boldsymbol{s}) = \mathbb{E}\{Y(s_0, \overline{M}_K(\boldsymbol{s}))\}$ is given by $\mathbb{E}_{\boldsymbol W}[\mathbb{E}\{Y(s_0, \overline{M}_K(\boldsymbol{s})) \mid \boldsymbol W\}]$, where $s_0 \in \text{supp}(S)$ and $\boldsymbol{s} = (s_1, \dots, s_K) \in \text{supp}(S)^K$. Consequently, identifying the conditional expectation $\Gamma(s_0, \boldsymbol{s} \mid \boldsymbol W) := \mathbb{E}\{Y(s_0, \overline{M}_K(\boldsymbol{s})) \mid \boldsymbol W\}$ is sufficient for identifying the full estimand. To proceed, we extend the identification assumptions \ref{assumption:consistency_two_mediators}--\ref{assumption:identification_two_mediators} to the general $\text{SCOM}(K)$ setting as follows:
\begin{assumption}[Consistency]\label{assumption:consistency}
The consistency assumption states that the observed values are equal to their corresponding potential outcomes: $M_1=M_1(s)$, $M_2=M_2(s,m_1)$, $\cdots$, $M_K=M_K(s,\overline{\boldsymbol m}_{K-1})$ and $Y=Y(s,\overline{\boldsymbol m}_{K})$, where $s,\overline{\boldsymbol m}_{K}=(m_1,m_2,\cdots,m_K)$ denote the observed realizations of the exposure $S$ and $K$ sequentially ordered mediators $\overline{\boldsymbol M}_{K}=(M_1,M_2,\cdots,M_K)$.
\end{assumption}
\begin{assumption}\label{assumption:identification}
For any \(s_0 \in \text{supp}(S)\) and \(\boldsymbol{s} = (s_1, s_2, \ldots, s_K) \in \text{supp}(S)^K\), $\overline{\boldsymbol m}_{K} \in  \text{supp}(\overline{\boldsymbol M}_{K})$ and $y\in  \text{supp}(Y)$:
\begin{enumerate}
    \item[(i)] $Y(s,\overline{\boldsymbol m}_{K})\perp\!\!\!\perp S|\boldsymbol W$,
    \item[(ii)] $Y(s,\overline{\boldsymbol m}_{K})\perp\!\!\!\perp M_1|S,\boldsymbol W$ and for any $k>1$, $Y(s,\overline{\boldsymbol m}_{K})\perp\!\!\!\perp M_k|S,\overline{\boldsymbol M}_{k-1},\boldsymbol W$,
    \item[(iii)] For any $k>1$, $M_k(s_k,\overline{\boldsymbol m}_{k-1}) \perp\!\!\!\perp M_1|S,\boldsymbol W$, and  for any $k>1$ and $l>1$, with $k>l$, $M_k(s_k,\overline{\boldsymbol m}_{k-1}) \perp\!\!\!\perp M_l|S,\overline{\boldsymbol M}_{l-1},\boldsymbol W$.
    \item[(iv)] For $k>1$, $M_k(s_k,\overline{\boldsymbol m}_{k-1}) \perp\!\!\!\perp M_1(s_1)|\boldsymbol W$, and for any $k>1$ and $l>1$, with $k>l$, $M_k(s_k,\overline{\boldsymbol m}_{k-1}) \perp\!\!\!\perp M_l(s_l,\overline{\boldsymbol m}_{l-1})|\boldsymbol W$.
    \item[(v)] $Y(s,\overline{\boldsymbol m}_{K})\perp\!\!\!\perp M_1(s_1)|\boldsymbol W$, and for any $k>1$, $Y(s,\overline{\boldsymbol m}_{K})\perp\!\!\!\perp M_k(s_k,\overline{\boldsymbol m}_{k-1})|\boldsymbol W$.
\end{enumerate}
\end{assumption}
Theorem \ref{theorem:identification} establishes the identification formulae for a DAG under the setup of SCOM$(K)$.

\begin{theorem}\label{theorem:identification}
    For any \(s_0 \in \text{supp}(S)\) and \(\boldsymbol{s} = (s_1, s_2, \ldots, s_K) \in \text{supp}(S)^K\), $K$ sequentially ordered mediators $M_1,\cdots,M_K$, generated from the G-SEM system \eqref{eq:npsem},
    under Assumptions \ref{assumption:consistency}-\ref{assumption:identification}, the conditional expectation $\Gamma(s_0, \boldsymbol{s}|\boldsymbol W)$ is identified as,
\begin{equation*}
\int_{\text{supp}(\overline{\boldsymbol M}_{K})}\mathbb{E}(Y \mid S=s_0, \overline{\boldsymbol M}_{K}=\overline{\boldsymbol m}_{K}, \boldsymbol W) \, 
\prod_{k=1}^K
dF_{M_k}(m_k\mid S=s_k, \overline{\boldsymbol M}_{k-1}=\overline{\boldsymbol m}_{k-1}, \boldsymbol W)
\end{equation*}
\end{theorem}
where for any two random variables $U$ and $V$, \(F_{U}(u|V=v)\) denotes the conditional distribution function of \(U\) given \(V\) and for any $k\in \{1,2,\cdots,K\}$, $\overline{\boldsymbol m}_{k}=(m_1,m_2,\cdots,m_k)$ and $\overline{\boldsymbol m}_{0}=\phi$. The proof of Theorem \ref{theorem:identification} is given in Supplementary Section A.1.
Theorem \ref{theorem:identification} expresses $\Gamma(s_0, \boldsymbol{s}|\boldsymbol W)$ as a nested sequence of integrals over the conditional distributions of SCOM($K$) and the conditional expectation of $Y$ under G-SEM \eqref{eq:npsem}. 
\subsection{Linear hierarchical models for continuous outcomes}\label{sec:linout}
A linear model for continuous outcome $Y$ takes the form:
\begin{align}
     Y &= \gamma_S S + \sum_{k=1}^K \beta_{k} M_k + \boldsymbol\delta_{Y}^T \boldsymbol W + \epsilon_Y, \label{eq:linear}
\end{align}
where $\epsilon_Y$ is a random variable with mean 0 and finite variance. Here, $\gamma_S$ denotes the effect of $S$ on $Y$ and $\beta_k$'s capture the effects of individual $M_k$ on $Y$. For the sequentially ordered continuous mediators,
the following system of hierarchical linear models specifies the G-SEM \eqref{eq:npsem}:
\begin{align}
\left\{
\begin{aligned}
    M_K &= \alpha_K S + \sum_{k=1}^{K-1} \xi_{K|k} M_k + \boldsymbol\delta_{M_K}^T \boldsymbol W + \epsilon_{M_K}, \\
    M_k &= \alpha_k S + \sum_{j=1}^{k-1} \xi_{k|j} M_j + \boldsymbol\delta_{M_k}^T \boldsymbol W + \epsilon_{M_k}, \quad k=K-1,\ldots, 2 \\
    M_1 &= \alpha_1 S + \boldsymbol\delta_{M_1}^T \boldsymbol W + \epsilon_{M_1},
\end{aligned}
\right.
\label{eq:medsem}
\end{align}
where error terms $\epsilon_{M_1}, \ldots, \epsilon_{M_K}$ are assumed to be mutually independent with mean 0 and finite variances. Here, $\alpha_k$'s quantify the effects of $S$ on individual $M_k$ and parameters $\xi_{j|k}$ (for $j > k$) quantify the effect of a preceding $M_k$ on a latter $M_j$, namely $M_k\rightarrow M_j$. In the following corollary \ref{corollary:linear}, we do not impose any distributional assumptions on the error terms $\epsilon_{M_1}, \ldots, \epsilon_{M_K}$ and $\epsilon_Y$ except for some moment conditions.
\begin{corollary}\label{corollary:linear}
   Suppose Assumptions \ref{assumption:consistency}-\ref{assumption:identification} hold under models \eqref{eq:linear} and \eqref{eq:medsem}. Then the conditional expectation $\Gamma(s_0, \boldsymbol{s}|\boldsymbol W)$ is given by
   \begin{align*}
    \Gamma(s_0, \boldsymbol{s}|\boldsymbol W)=\gamma_S s_0+ \sum_{k=1}^K\alpha_k\eta_{M_k\leadsto Y}s_k + \left(\boldsymbol\delta_{Y}^T+\sum_{k=1}^K \eta_{M_k\leadsto Y}\boldsymbol\delta_{M_k}^T\right)\boldsymbol W,
\end{align*}
where $\eta_{M_K \leadsto Y}=\beta_{K}$, and for $k \in \{1,2,\cdots,K-1\}$, $\eta_{M_k \leadsto Y}=\beta_{k} + \sum_{j=k+1}^K\xi_{j|k}\eta_{M_j\leadsto Y}$. Consequently, we obtain $\tau_1(S\rightarrow Y)_{(s,s')} = \gamma_S (s' - s)$ and $\tau_k^{\dag}(S\rightarrow M_K\leadsto Y)_{(s,s')}=\alpha_k\eta_{M_k\leadsto Y}(s' - s)$, $k=1,2,\cdots,K$.
\end{corollary}
Parameter $\eta_{M_k \leadsto Y}$ represents a composite effect that accounts for the set of all possible pathways from \( M_k \) to $Y$ regardless of the other mediators. The proof of Corollary \ref{corollary:linear} is given in Supplementary Section A.2. Hence, examining the existence of mediation effect $\tau_k^{\dag}(S\rightarrow M_K\leadsto Y)_{(s,s')}$ equal to 0 or not under \eqref{eq:linear} and \eqref{eq:medsem} is equivalent to testing $H_0\colon \alpha_k \eta_{M_k \leadsto Y} = 0$ versus $H_{1k}\colon \alpha_k \eta_{M_k \leadsto Y} \neq 0$ for each $k=1,2,\cdots,K$. See Section \ref{section:test} for more details of the hypothesis testing method developed for these composite nulls to achieve proper Type-I error control.
\section{Methodology for Discrete Outcomes}\label{sec:cat}
In this section, we extend the proposed framework to discrete outcomes, specifically focusing on binary and count outcomes.

\subsection{Definitions of Path-Specific Effects for Binary Outcomes}
We now focus on mediation estimands with binary outcomes under DAGs governed by SCOM($K$). In such a setting, causal estimands are typically expressed through odds ratios \citep{vanderweele2010odds}. For any $s,s'\in \text{Supp}(S)$, the total effect attributable to a change of exposure $S$ from $s$ to $s'$, takes the form of odds ratio conditional on the vector of confounders $\boldsymbol W$:
\begin{align}
    \text{TE}(s,s') = \frac{\text{P}(Y(s')=1 \mid \boldsymbol W) / \{1 - \text{P}(Y(s')=1 \mid \boldsymbol W)\}}{\text{P}(Y(s)=1 \mid \boldsymbol W) / \{1 - \text{P}(Y(s)=1 \mid \boldsymbol W)\}},\label{eq:totbin}
\end{align}
where $\text{P}(Y(s')=1 \mid \boldsymbol W)$ and $\text{P}(Y(s)=1 \mid \boldsymbol W)$ denote the probabilities of the potential binary outcomes $Y(s')$ and $Y(s)$ taking value 1, respectively, conditional on $\boldsymbol W$. 
\par
Denote a generic quantity $\Gamma(s_0,\boldsymbol s|\boldsymbol W)=
\text{P}[Y(s_0,\overline{M}_K(\boldsymbol{s}))=1|\boldsymbol W]$ for any $s_0 \in \text{supp}(S)$ and $\boldsymbol s \in \text{supp}(S)^K$. First, the direct effect of $\tau_1(S\rightarrow Y)_{(s,s')}$ takes the following form of odds ratio:
\begin{equation}
   \tau_1(S\rightarrow Y)_{(s,s')}=\frac{\Gamma(s',\boldsymbol{s}'|\boldsymbol W)/\{1-\Gamma(s',\boldsymbol{s}'|\boldsymbol W)\}}{\Gamma(s,\boldsymbol{s}'|\boldsymbol W)/\{1-\Gamma(s,\boldsymbol{s}'|\boldsymbol W)\}},\label{eq:dibin}
\end{equation}
where $\boldsymbol {s}'=(s',s',\cdots,s')$ a $k-$ dimensional vector. For $k \in \{1,2,\cdots,K\}$, the composite mediation effect $\tau_k^{\dag}(S\rightarrow M_k\leadsto Y)_{(s,s')}$ is given by
\begin{align}
   \tau_k^{\dag}(S\rightarrow M_k\leadsto Y)_{(s,s')}=\frac{\Gamma(s, \boldsymbol{s}_{1k}|\boldsymbol W)/\{1-\Gamma(s, \boldsymbol{s}_{1k}|\boldsymbol W)\}}{\Gamma(s, \boldsymbol{s}_{2k}|\boldsymbol W)/\{1-\Gamma(s, \boldsymbol{s}_{2k}|\boldsymbol W)\}},\label{eq:mbin}
\end{align}
where the vectors \( \boldsymbol{s}_{1k} \) and $\boldsymbol{s}_{2k}$ are defined in Section \ref{sec:contdef}. 
\par
To identify the path-specific effects in \eqref{eq:dibin} and \eqref{eq:mbin} under binary outcome setting, we first identify the generic quantity $\Gamma(s_0, \boldsymbol{s} \mid \boldsymbol{W}) = P[Y(s_0, \overline{M}_K(\boldsymbol{s})) = 1 \mid \boldsymbol{W}] = \mathbb{E}[Y(s_0, \overline{M}_K(\boldsymbol{s})) \mid \boldsymbol{W}]$. Since the expectation of a binary indicator corresponds to its success probability, Theorem \ref{theorem:identification} remains applicable under Assumptions \ref{assumption:consistency}--\ref{assumption:identification} to identify this quantity in the binary outcome setting.
\subsection{Probit model for Binary Outcomes}
The probit model for a binary outcome takes the form,
\begin{align}
    P(Y=1|S,M_1,\cdots,M_K,\boldsymbol W)=
    \Phi\left(\gamma_S S+\sum_{k=1}^K\beta_{k}M_k +\boldsymbol\delta_{Y}^T\boldsymbol W\right), \label{eq:probit}
\end{align}
where the link function $\Phi(\cdot)$ is the CDF of the standard normal distribution. The hierarchical models for the sequentially ordered mediators $M_1,M_2,\cdots,M_K$ remain the same as those given in \eqref{eq:medsem}, in which we assume that the error terms \( \epsilon_{M_k} \sim \mathcal{N}(0, \sigma^2_{M_k}) \), $k=1,2,\cdots,K$ marginally and are mutually independent. 

\begin{corollary}\label{corollary:probit}
Under Assumptions \ref{assumption:consistency}-\ref{assumption:identification}, the probit model \eqref{eq:probit} and the linear normal hierarchical models \eqref{eq:medsem}, the following closed form expression for the generic quantity $\Gamma(s_0, \boldsymbol{s}|\boldsymbol W)$:
\begin{align*}
   \Gamma(s_0, \boldsymbol{s}|\boldsymbol W)=\Phi\left(\frac{\gamma_S s_0+ \sum_{k=1}^K\alpha_k\eta_{M_k\leadsto Y}s_k + \left(\boldsymbol\delta_{Y}^T+\sum_{k=1}^K \eta_{M_k\leadsto Y}\boldsymbol\delta_{M_k}^T\right)\boldsymbol W }{\sqrt{1+\sum_{k=1}^K\eta_{M_k\leadsto Y}^2\sigma^2_{M_{k}}}}\right),
\end{align*}
\end{corollary}
where $\eta_{M_K \leadsto Y}=\beta_{K}$, and for $k \in \{1,2,\cdots,K-1\}$, $\eta_{M_k \leadsto Y}=\beta_{k} + \sum_{j=k+1}^K\xi_{j|k}\eta_{M_j\leadsto Y}$. 
The proof of Corollary \ref{corollary:probit} is presented in Supplementary Section A.3. Plugging the above expression into \eqref{eq:dibin} and \eqref{eq:mbin}, we obtain the closed-form expressions of both direct and composite mediation effects in the probit outcome model. To establish the presence of a composite mediation effect $\tau_k^{\dag}(S\rightarrow M_k\leadsto Y)_{(s,s')}$ equal to 0 or not under \eqref{eq:probit} and \eqref{eq:medsem}, we test the following equivalent null hypothesis,
$H_0\colon \Gamma(s_0, \boldsymbol{s}_{1k}|\boldsymbol W) = \Gamma(s_0, \boldsymbol{s}_{2k}|\boldsymbol W)$ versus $H_{1k}\colon \Gamma(s_0, \boldsymbol{s}_{1k}|\boldsymbol W) \neq \Gamma(s_0, \boldsymbol{s}_{2k}|\boldsymbol W)$, for each $k=1,2,\cdots,K$. Using Corollary \ref{corollary:probit}, this is equivalent to test $H_0\colon \alpha_k \eta_{M_k \leadsto Y} = 0$ versus $H_{1k}\colon \alpha_k \eta_{M_k \leadsto Y} \neq 0$, due to the fact that $\Phi(\cdot)$ is a strictly positive and monotonically increasing function. 

\subsection{Logistic Model for Binary Outcomes}\label{sec:logitrare}
Consider a logistic outcome model for binary outcome $Y$ of the form,
\begin{align}
    \text{logit}\{P(Y=1|S,M_1,\cdots,M_K,\boldsymbol W)\}=\gamma_S S+\sum_{k=1}^K\beta_{k}M_k +\boldsymbol\delta_{Y}^T\boldsymbol W \cdot \label{eq:logit}
\end{align}
The linear normal SEM \eqref{eq:medsem} is used to model the sequentially ordered mediators $M_1,M_2,\cdots,M_K$ with independent normal error terms \( \epsilon_{M_k} \sim \mathcal{N}(0, \sigma^2_{M_k}) \), $k=1,2,\cdots,K$. In this case, unfortunately the integral 
in Theorem \ref{theorem:identification} does not permit a closed-form expression, even in the scenario of a single mediator $(K=1)$. To address this numerical challenge, \citet{vanderweele2010odds} introduced a special setting under the rare outcome assumption, while \citet{gaynor2019mediation} proposed to use the probit link function to approximate the logit link function. Since the latter introduces substantial complexity and technical difficulties even in the case of a single mediator, below we focus only on investigating the rare outcome scenario under the following assumption of approximation,
\begin{equation}
    \text{logit}\{P(Y=1|S,\overline{\boldsymbol M}_{K},\boldsymbol W)\}\approx \text{log}\{P(Y=1|S,\overline{\boldsymbol M}_{K},\boldsymbol W)\}\cdot\label{eq:logitrare}
\end{equation}
Under the approximation \ref{eq:logitrare}, the odds ratio reduces approximately to the relative risk. 
\begin{corollary}\label{corollary:logitrare}
    Under Assumptions \ref{assumption:consistency}-\ref{assumption:identification}, the logistic outcome model \eqref{eq:logit}, the hierarchical mediator models \eqref{eq:medsem}, and the rare outcome approximation \ref{eq:logitrare}, the generic expectation $\Gamma(s_0, \boldsymbol{s}|\boldsymbol W)$ may be expressed as follows,
    \begin{align*}
     \Gamma(s_0, \boldsymbol{s})={\exp}\left\{\gamma_S s_0+ \sum_{k=1}^K\alpha_k\eta_{M_k\leadsto Y}s_k + \left(\boldsymbol\delta_{Y}^T+\sum_{k=1}^K \eta_{M_k\leadsto Y}\boldsymbol\delta_{M_k}^T\right)\boldsymbol W +\frac{1}{2}\sum_{k=1}^K\eta_{M_k\leadsto Y}^2\sigma^2_{M_{k}}\right\}.
\end{align*}
Consequently, the direct effect $\tau_1(S\rightarrow Y)_{(s,s')} =\exp\{\gamma_S (s' - s)\}$, and the composite mediation effect $\tau_k^{\dag}(S\rightarrow M_K\leadsto Y)_{(s,s')}=\exp\{\alpha_k\eta_{M_k\leadsto Y}(s' - s)\}$.
\end{corollary}
Here $\eta_{M_K \leadsto Y}=\beta_{K}$, and for $k \in \{1,2,\cdots,K-1\}$, $\eta_{M_k \leadsto Y}=\beta_{k} + \sum_{j=k+1}^K\xi_{j|k}\eta_{M_j\leadsto Y}$. The proof of Corollary \ref{corollary:logitrare} is included in Supplementary Section A.4. According to Corollary \ref{corollary:logitrare}, to test if the composite mediation effect \( \tau_k^{\dag}(S\rightarrow M_K\leadsto Y)_{(s,s')} \) equal to 0 or not under \eqref{eq:logit}, \eqref{eq:medsem}, and \eqref{eq:logitrare}, equivalently we check whether \( \exp(\alpha_k \eta_{M_k \leadsto Y}) = 1 \) or not. This is equivalent to test $H_0\colon \alpha_k \eta_{M_k \leadsto Y} = 0$ versus $H_{1k}\colon \alpha_k \eta_{M_k \leadsto Y} \neq 0$ for each $k=1,2,\cdots,K$.

\subsection{Definitions of Path-Specific Effects for Count Outcomes}
For a count outcome \( Y \), typically causal effects are expressed multiplicatively via a ratio between the expected potential outcomes. Then we write the total effect conditioned on the vector of confounders $\boldsymbol W$ as $\text{TE}(s, s') = {\mathbb{E}\{Y(s')|\boldsymbol W\}}\big/{\mathbb{E}\{Y(s)|\boldsymbol W\}}$
where $s, s' \in \text{Supp}(S)$. Denote a generic expectation $\Gamma(s_0,\boldsymbol s|\boldsymbol W)=\mathbb{E}\{Y(s_0,\overline{M}_K(\boldsymbol{s}))|\boldsymbol W\}$ for any $s_0 \in \text{supp}(S)$ and $\boldsymbol s \in \text{supp}(S)^K$. 

\par
First, the direct effect of \( S \rightarrow Y \) takes the form $\tau_1(S\rightarrow Y)_{(s,s')}  = {\Gamma(s', \boldsymbol{s}'|\boldsymbol W)}\big/{\Gamma(s, \boldsymbol{s}'|\boldsymbol W)}$, where \( \boldsymbol{s}'=(s',s',\cdots,s') \in \text{supp}(S)^K \). Second, for any \( k \in \{1,2,\cdots,K\} \), the composite mediation effect is given by $\tau_k^{\dag}(S\rightarrow M_K\leadsto Y)_{(s,s')} = {\Gamma(s, \boldsymbol{s}_{1k}|\boldsymbol W)}\big/{\Gamma(s, \boldsymbol{s}_{2k}|\boldsymbol W)}$, where the vectors \( \boldsymbol{s}_{1k} \) and \( \boldsymbol{s}_{2k} \) are defined in Section \ref{sec:contdef}.
To identify the above path-specific effects in such setting, we focus on the generic quantity, $\Gamma(s_0, \boldsymbol{s} \mid \boldsymbol{W}) = \mathbb{E}[Y(s_0, \overline{M}_K(\boldsymbol{s})) \mid \boldsymbol{W}]$. Because this quantity is a conditional expectation, the identification result established in Theorem \ref{theorem:identification} remains valid under Assumptions \ref{assumption:consistency}--\ref{assumption:identification}. 
\subsection{Log-Linear Model for Count Outcomes}
Consider a log-linear outcome model for count outcomes specified as follows.
\begin{equation}
    \log\left\{\mathbb{E}(Y \mid S, M_1, \dots, M_K, \boldsymbol{W})\right\} = \gamma_S S + \sum_{k=1}^K \beta_{k} M_k + \boldsymbol{\delta}_Y^T \boldsymbol{W}. \label{eq:loglinear}
\end{equation}
The linear normal SEM \eqref{eq:medsem} is used to model the sequentially ordered mediators $M_1,\cdots,M_K$ with independent normal error terms \( \epsilon_{M_k} \sim \mathcal{N}(0, \sigma^2_{M_k}) \), $k=1,2,\cdots,K$. It is easy to show that in such setting of models \eqref{eq:loglinear} and \eqref{eq:medsem}, the direct effect of \( S \) on \( Y \) takes the form $\tau_1(S\rightarrow Y)_{(s,s')}=\exp\{\gamma_S (s' - s)\}$ and moreover, the composite mediation effects $\tau_k^{\dag}(S\rightarrow M_K\leadsto Y)_{(s,s')} = \exp\{\alpha_k \eta_{M_k \leadsto Y} (s' - s)\}$, $k\in\{1,2,\cdots,K\}$. Thus to test whether $\tau_k^{\dag}(S\rightarrow M_K\leadsto Y)_{(s,s')}$ equals to 0 or not is equivalent to testing $H_0\colon \alpha_k \eta_{M_k \leadsto Y} = 0$ versus $H_{1k}\colon \alpha_k \eta_{M_k \leadsto Y} \neq 0$ for each $k=1,2,\cdots,K$.
\section{Estimation} \label{sec:est}
In Sections \ref{sec:cont}--\ref{sec:cat}, we examined three widely used scenarios where the direct effect $\tau_1(S \rightarrow Y)_{(s,s')}$ and the composite mediation effects $\{\tau_k^{\dag}(S \rightarrow M_k \leadsto Y)_{(s,s')}\}_{k=1}^K$ are expressed as explicit functions of the model parameters $\gamma_S$ and $(\alpha_k, \xi_{j \mid k}, \beta_k, \sigma_{M_k})_{k=1,\, j \geq k}^K$. We employ maximum likelihood estimation (MLE) to estimate these parameters under their respective model specifications. The choice of MLE as our primary estimation framework facilitates the derivation of several key analytical properties.

\par
In the case of continuous outcomes, we assume that the error term in the outcome linear model \eqref{eq:linear} follows a normal distribution, 
\(\epsilon_Y \sim \mathcal{N}(0,\sigma^2_Y)\). In the case of probit or logistic model for binary outcomes, we assume outcome $Y$ follows a Bernoulli distribution for \(Y\), with the success probability specified by \eqref{eq:probit} for probit model and \eqref{eq:logit} for logistic model, respectively. In the case of count outcome, we assume outcome $Y$ follows a conditional Poisson distribution or a Negative Binomial distribution, with mean specified by the log-linear model in \eqref{eq:loglinear}. 
In all cases, the linear normal SEM \eqref{eq:medsem} is used to model the sequentially ordered mediators where \(\epsilon_{M_k} \sim \mathcal{N}(0, \sigma^2_{M_k})\) and are mutually independent. As usual, we estimate the direct effect \( \tau_1(S \rightarrow Y)_{(s,s')} \) and the composite mediation effects \( \tau_k^{\dag}(S \rightarrow M_k \leadsto Y)_{(s,s')} \) by plugging in the MLEs of the corresponding parameters into the suitable functional expressions. 

\par
Under some mild regularity conditions \citep{newey1994large}, the MLEs \( (\widehat{\alpha}_k, \widehat{\xi}_{j\mid k}, \widehat{\beta}_k)_{k = 1,\, j \geq k}^K \) are consistent estimators of the model parameters \( (\alpha_k, \xi_{j\mid k}, \beta_k)_{k = 1,\, j \geq k}^K \), and asymptotically normally distributed as the sample size \( n \to \infty \) namely,
\begin{equation}
    n^{1/2}\left\{(\widehat{\alpha}_k, \widehat{\xi}_{j\mid k}, \widehat{\beta}_k)_{k = 1,\, j \geq k}^K - (\alpha_k, \xi_{j\mid k}, \beta_k)_{k = 1,\, j \geq k}^K\right\}
    \xrightarrow{d} (Z_{\alpha_k}, Z_{\xi_{j\mid k}}, Z_{\beta_k})_{k = 1,\, j \geq k}^K, \label{eq:asymmle}
\end{equation}
where the asymptotic multivariate normal distribution has a certain positive definite covariance matrix. Given $\eta_{M_K \leadsto Y} = \beta_{K}$ and $\eta_{M_k \leadsto Y} = \beta_{k} + \sum_{j=k+1}^K \xi_{j \mid k} \eta_{M_j \leadsto Y}$ for $k \in \{1, \dots, K-1\}$, the corresponding MLEs are obtained via the recursive substitution $\widehat{\eta}_{M_K \leadsto Y} = \widehat{\beta}_{K}$ and $\widehat{\eta}_{M_k \leadsto Y} = \widehat{\beta}_{k} + \sum_{j=k+1}^K \widehat{\xi}_{j \mid k} \widehat{\eta}_{M_j \leadsto Y}$ for $k \in \{1, \dots, K-1\}$.

\begin{theorem}\label{thm:dist}
Let $(\sigma_{\alpha_k}^2,\sigma^2_{\xi_{j\mid k}},\sigma_{\beta_{k}}^2)_{k = 1,\, j \geq k}^K$ be the variances of $(Z_{\alpha_k}, Z_{\xi_{j\mid k}}, Z_{\beta_k})_{k = 1,\, j \geq k}^K$ respectively. 
Under the standard regularity conditions, for each $k=1,2,...K$, as $n\rightarrow \infty$
\begin{equation*}
    n^{1/2}
\begin{pmatrix}
    \wh \alpha_k\\
    \wh \eta_{M_k\leadsto Y}
\end{pmatrix}
\todistribution
\calN
\left\{
\begin{pmatrix}
    \alpha_k\\
    \eta_{M_k\leadsto Y}
\end{pmatrix},
\begin{pmatrix}
    \sigma_{\alpha_k}^2  & 0\\
    0 & \sigma_{\eta_{M_k\leadsto Y}}^2
\end{pmatrix}
\right\},
\end{equation*}
where $\sigma_{\eta_{M_K\leadsto Y}}^2 = \sigma_{\beta_{K}}^2$ and for any $k<K$, $$\sigma_{\eta_{M_{k}\leadsto Y}}^2 = \sigma_{\beta_{k}}^2 + \sum_{j = k + 1}^K\{(\xi_{j\mid k})^2 \sigma_{\eta_{M_j\leadsto Y}}^2 + (\eta_{M_j\leadsto Y})^2\sigma_{\xi_{j\mid k}}^2 + 2\xi_{j\mid k}\cov(Z_{\beta_{k}}, Z_{\beta_{j}})\}\cdot$$
\end{theorem}
Consequently, $\widehat{\alpha}_k$ and $\widehat{\eta}_{M_k\leadsto Y}$ are asymptotically independent for each $k=1, \dots, K$, as the covariance terms in the joint bivariate normal distribution of $(\widehat{\alpha}_k, \widehat{\eta}_{M_k\leadsto Y})$ are zero according to Theorem~\ref{thm:dist}. The proof of Theorem~\ref{thm:dist} can be found in Supplementary Section A.5.
Theorem~\ref{thm:dist} allows us to calculate the asymptotic variances of the direct effect \( \tau_1(S \rightarrow Y)_{(s,s')} \) and the composite mediation effects \( \tau_k^{\dag}(S \rightarrow M_k \leadsto Y)_{(s,s')} \) are obtained by employing the first-order delta method. Moreover, Theorem~\ref{thm:dist} serves as the basis to establish inference in Section~\ref{section:test}.

\section{Inference}\label{section:test}
\subsection{Hypothesis Testing Problem}
Testing the existence of the direct effect $\tau_1(S \rightarrow Y)_{(s,s')}$ is equivalent to testing the null hypothesis $H_0: \gamma_S = 0$ or not in all the specified GLM outcome models. This can be conducted using standard asymptotic tests, such as the Wald, Rao’s score, or likelihood ratio test. The main challenge here pertains to testing the existence of the composite mediation pathways \( \tau_k^{\dag}(S \rightarrow M_k \leadsto Y)_{(s,s')} \) while ensuring a proper Type-I error control. As discussed in Sections \ref{sec:cont}--\ref{sec:cat}, across the class of GLM outcome models testing whether \( \tau_k^{\dag}(S \rightarrow M_k \leadsto Y)_{(s,s')} \) equals 0 or not is equivalent to testing:
\begin{equation}\label{null_hypothesis}
H_0\colon \alpha_k \eta_{M_k \leadsto Y} = 0 \quad \text{versus} \quad H_{1k}\colon \alpha_k \eta_{M_k \leadsto Y} \neq 0,
\end{equation}
\noindent
for each $k=1,2,\cdots,K$. Note that the null hypothesis \( H_0 \) in \eqref{null_hypothesis} is composite as the zero product can arise from three distinct cases: (i)  \( \alpha_k = 0 \) and \( \eta_{M_k \leadsto Y} \neq 0 \), (ii) \( \alpha_k \neq 0 \) and \( \eta_{M_k \leadsto Y} = 0 \), and (iii) \( \alpha_k = 0 \) and \( \eta_{M_k \leadsto Y} = 0 \).

Even in a single mediator setting, standard tests such as Sobel’s test and the MaxP test are known to be challenged by co-existence of two different limiting distributions between null scenarios (i),(ii) versus (iii), but beforehand the underlying truth is unknown. Owing to this issue of non-uniqueness, existing classical tests are known to be overly conservative in the control of Type-I error. To address such conservativeness under the composite null, there have been recent advances proposed by \citet{he2024adaptive,roy2024subsampling} and \citet{chen2024quantile}. However, all these new approaches remain restricted to the single-mediator testing problems. In parallel, recent efforts have targeted hypothesis testing in multi-mediator settings. Notably, \citet{dai2022multiple}, \citet{liu2022large}, and \citet{liu2025simple} developed testing procedures that explicitly use the proportions of the three null types to enhance power. Arguably, their approaches rely on accurate estimation of the null composition, and more crucially, they are not equipped to handle sequential or temporal causal ordering among mediators considered in this paper. To address this technical gap, we propose a novel procedure inspired by the method in \citet{roy2024subsampling}. While their work focused on a single mediator, we establish a powerful testing procedure for the composite mediation pathway $\tau_k^{\dag}(S \rightarrow M_k \leadsto Y)_{(s,s')}$ under the null hypothesis $H_0$ in \eqref{null_hypothesis} for each $k=1, 2, \dots, K$.

\subsection{Classical Sobel's Test}
Since our test is built upon an extension of the classical Sobel test \citep{sobel1982AsymptoticConfidenceIntervals}, we begin with an introduction to this test in this section. Let $\wh\sigma^2_{\alpha_k}$ and $\wh\sigma^2_{\eta_{M_k\leadsto Y}}$ be consistent estimators of the asymptotic variances $\sigma_{\alpha_k}^2$ and $\sigma_{\eta_{M_k\leadsto Y}}^2$, $k=1,2,\cdots,K$, respectively. Under the asymptotic normality, we may construct the corresponding Student's $t$ statistics, $T_{\alpha_k} = n^{1/2}\wh \alpha_k/\wh\sigma_{\alpha_k}$ for $\alpha_k$ and $T_{\eta_{M_k\leadsto Y}} = n^{1/2} \wh \eta_{M_k\leadsto Y}/\wh\sigma_{\eta_{M_k\leadsto Y}}$ for $\eta_{M_k\leadsto Y}$. This leads to the so-called classical Sobel's statistic by a direct product:
\begin{equation}
    T_{S\to M_k\leadsto Y} = \frac{T_{\alpha_k}T_{\eta_{M_k\leadsto Y}}}{(T_{\alpha_k}^2 + T_{\eta_{M_k\leadsto Y}}^2)^{1/2}}, \quad k=1,2,\cdots,K
    \cdot \label{eq:sobel}
\end{equation}
Lemma \ref{lemma:sobel_limit} below establishes the asymptotic distribution of the Sobel's test $T_{S\to M_k\leadsto Y}$ which varies by the parameter configuration under the null hypothesis $H_0$ in \eqref{null_hypothesis} for each $k=1,2,\cdots,K$.
\begin{Lemma}\label{lemma:sobel_limit}
Under the null hypothesis $H_0$ in \eqref{null_hypothesis} for each $k=1,2,\cdots,K$, the limiting distribution of Sobel's test statistic \( T_{S \to M_k \leadsto Y} \) in \eqref{eq:sobel} is either \( \mathcal{N}(0, 1) \) when \( (\alpha_k, \eta_{M_k \leadsto Y}) \neq (0, 0) \), or \( \mathcal{N}(0, 1/4) \) when \( (\alpha_k, \eta_{M_k \leadsto Y}) = (0, 0) \).
\end{Lemma}
The proof of Lemma \ref{lemma:sobel_limit}  is detailed in Supplementary Section A.6. The above Sobel's test in \eqref{eq:sobel} would suffer a conservation in Type-I error control when the cutoff is determined by the standard normal $\mathcal{N}(0, 1)$ not by the other $\mathcal{N}(0, 1/4)$, as known in the current literature \citep{mackinnon2002comparison}. Conversely, if $\mathcal{N}(0, 1/4)$ was used to determine the cutoff, then an inflated Type-I error occurs.
% till here
\subsection{Sequentially Ordered Mediation Effect  Test (SOMET)}
To overcome the shortcoming of the Sobel's test in the Type-I error control for SCOM($K$), we propose a \( Q \)-fold data-splitting approach, inspired by the method proposed by \citet{roy2024subsampling} with a single mediator ($K=1$). The key technical advantage of the proposed data-split strategy is to utilize a Student’s \( t \)-transformation that allows to disconnect the reliance of non-unique asymptotic variances varied over different null hypothesis scenarios as shown in Lemma \ref{lemma:sobel_limit}. Consequently, a unique limiting distribution is established to determine the cutoff, thereby elegantly overcoming the inherent conservativeness of the test.

First, partition the full data $\mathcal{D}$ randomly into \( Q \) disjoint sub-datasets, denoted by \( \{\mathcal{D}_q\}_{q = 1}^Q \) such that $\mathcal{D}=\cup_{q=1}^Q \mathcal{D}_q$. For each sub-dataset \( \mathcal{D}_q \) with the size $n_q=|\mathcal{D}_q|$, we construct the Sobel's test statistic according to \eqref{eq:sobel}, denoted by \( T_{S \to M_k \leadsto Y}(\mathcal{D}_q) \). Second, take an average of these $Q$ Sobel's test statistics given as $\wb T_{S \to M_k \leadsto Y} = \frac{1}{Q} \sum_{q = 1}^Q T_{S \to M_k \leadsto Y}(\mathcal{D}_q)$. Third, construct a Student’s \( t \)-transformation of the form:
\begin{equation}
    \wt T_{S\to M_k\leadsto Y}= \frac{Q^{1/2} \, \wb T_{S\to M_k\leadsto Y}}{\left[\sum_{q = 1}^Q \left\{T_{S\to M_k\leadsto Y}(\mathcal{D}_q) - \wb T_{S\to M_k\leadsto Y}\right\}^2 / (Q - 1) \right]^{1/2}}\cdot \label{eq:somet}
\end{equation}
\begin{theorem}\label{theorem:t_distribution}
    Under the composite null hypothesis $H_0$ in \eqref{null_hypothesis} for each $k=1,2,\cdots,K$, for a given positive integer $Q$, the student's t-statistic \( \wt T_{S \to M_k \leadsto Y} \todistribution t_{Q - 1} \) as $\text{min}\{n_1,n_2,\cdots,n_Q\} \rightarrow \infty$, where \( t_{Q - 1} \) denotes the Student’s \( t \)-distribution with \( (Q - 1) \) degrees of freedom. Moreover, in the ease of equal data split with $n_q=\frac{n}{Q}$, the asymptotic t-distribution holds as $n\rightarrow \infty$.
\end{theorem}
The validity of Theorem \ref{theorem:t_distribution} is in fact immediately implied by Lemma \ref{lemma:sobel_limit} according to the continuous mapping theorem. Thus, the proof is omitted. It is worth emphasizing that Theorem~\ref{theorem:t_distribution} states that the null distribution of our proposed test statistic \(\widetilde{T}_{S \to M_k \leadsto Y}\) in \eqref{eq:somet} is 
\(t_{Q-1}\), which does not depend on any specific asymptotic variances with respect to any specific underlying true null parameter scenarios. Consequently, this ensures a valid inference without requiring \emph{a priori} knowledge of the true scenario (e.g., proportions of the null scenarios) under \(H_0\) for each $k=1,2,\cdots,K$.
In practice, the choice of \(Q\) suggested by \citet{roy2024subsampling}, namely \( Q = \lfloor 0.5\, n^{1/2} \rfloor \), together with an equally sized division with \( n_q = n_Q \approx \tfrac{n}{Q} \), has been found to work reasonably well based on our extensive numerical experiences in implementing the above \(Q\)-fold data-splitting strategy.

Although theoretically, one round of random split can ensure a proper Type I error control, our numerical experiences suggest that statistical power can be low due to the small sample size of the sub-dataset and the large variability of statistic $\wt T_{S\to M_k\leadsto Y}$ in \eqref{eq:somet}. To improve power and robustness by reducing the variability of the test statistic, we propose to repeat the \( Q \)-fold splitting procedure by \textit{say}, \( B \) rounds, and then aggregate the yielded \( p \)-values over $B$ replicates using a method similar to Cauchy Combination Test \citep{liu2022large}. Of note, the Cauchy test method is originally developed for different hypotheses with a common dataset. Let $\{p_b\}_{b = 1}^B$ be the set of $p$-values obtained from $B$ rounds of $Q-$folds data-split, and let $\{w_b\}_{b = 1}^B$ be certain non-negative weights satisfying $\sum_{b = 1}^Bw_b = 1$. Given a nominal level $\psi\in(0,1)$, the Cauchy Combination test enables to combine these $p$-values to form a combination-type statistic, $T_B = \sum_{b = 1}^B w_b\text{tan}\{(0.5 - p_b)\pi\}$, where $\text{tan}\{(0.5 - p_b)\pi\}$ follows a standard Cauchy distribution. Then we reject the composite null $H_0$ in \eqref{null_hypothesis} if the test statistic $T_B$ exceeds the $(1-\psi)^{\text{th}}$ quantile of the standard Cauchy distribution. A clear technical advantage of the Cauchy combination test approach is that it remains valid under arbitrary dependence among the \( p \)-values, making it particularly well-suited for the repeated-splitting framework.

\section{Simulation Experiments}\label{section:simulation}
We conduct simulation studies to evaluate the performance of both MLE estimation and SOMET-based inference. The evaluation encompasses three popular types of outcome model specifications, including a linear model for continuous outcomes, and logistic and probit models for binary outcomes.

\subsection{Linear SCOM(3) for continuous outcomes}\label{subsection:linear_SEM}
Consider a data-generation mechanism governed by the following linear SEM with three sequential causally ordered mediators,
SCOM(3), given as follows:
\begin{align}\label{model:simulation_linear_SEM}
\left\{
\begin{aligned}
   M_1 &= \alpha_{1} S + \delta_{M_1} W + \epsilon_{M_1}, \\
   M_2 &= \alpha_{2} S + \xi_{2|1} M_1 + \delta_{M_2} W + \epsilon_{M_2}, \\
   M_3 &= \alpha_{3} S + \sum_{k=1}^{2} \xi_{3|k} M_k + \delta_{M_3} W + \epsilon_{M_3}, \\
   Y   &= \gamma_S S + \sum_{k=1}^{3} \beta_{k} M_k + \delta_{Y} W + \epsilon_Y.
\end{aligned}
\right.
\end{align}
with a binary exposure $S\sim \text{Ber}(0.5)$, a continuous confounder $W\sim \text{Unif}(0,1)$, and four-dimensional vector of mutually independent normal errors $(\epsilon_{M_1},\epsilon_{M_2},\epsilon_{M_3},\epsilon_{Y})\sim\calN(\mathbf{0}_4,4\mathbf{I}_{4\times4})$.
%We set the true parameters $\xi_{2|1} = \xi_{3|1} = \xi_{3|2} = 1$, 
Throughout the simulation experiments, fix the true parameters, $\xi_{2|1} = \xi_{3|1} = \xi_{3|2} = 1$, 
$\gamma_S = 1$, and $\delta_{M_k} = \delta_Y = 1$ for $k=1,2,3$, used above.

\textbf{Results of Estimation.} We evaluate the performance of the MLE under the four parameter configurations, in addition to the fixed parameter above,  
\begin{align}
H_{01}:& \; (\alpha_1,\alpha_2,\alpha_3,\beta_1,\beta_2,\beta_3) = (1,1,1,0,0,0)\Rightarrow \alpha_k\neq 0, \eta_{M_k \leadsto Y} = 0 \nonumber\\
H_{02}:& \; (\alpha_1,\alpha_2,\alpha_3,\beta_1,\beta_2,\beta_3) = (0,0,0,1,1,1) \Rightarrow \alpha_k= 0, \eta_{M_k \leadsto Y} \neq 0 \nonumber\\
H_{03}:& \; (\alpha_1,\alpha_2,\alpha_3,\beta_1,\beta_2,\beta_3) = (0,0,0,0,0,0)\Rightarrow \alpha_k= 0, \eta_{M_k \leadsto Y} = 0 \nonumber\\
H_{1}:& \; (\alpha_1,\alpha_2,\alpha_3,\beta_1,\beta_2,\beta_3) = (1,1,1,1,1,1) \Rightarrow \alpha_k\neq 0, \eta_{M_k \leadsto Y} \neq 0. \label{eq:nullsimu}
\end{align}
Note that under $H_{01}$--$H_{03}$ settings, the composite mediation effects $\tau_1^{\dag}$, $\tau_2^{\dag}$ and $\tau_3^{\dag}$ equals to zero. 
The performance of MLE for the model parameters is assessed using mean squared error (MSE), estimation bias, and empirical variance over $2000$ rounds of replications. The sample size $n$ varies over $\{400, 600, 800, 1000\}$. Supplementary Table B.1 and Supplementary Figure B.1 summarize the results for the MLE estimation under these settings. 
Across all scenarios, it is clear that  bias, MSE, and empirical variance decrease with increased sample sizes. The MLE attains a linear convergence rate under the settings of \(H_{01}\), \(H_{02}\), and \(H_{1}\) according to Supplementary Figure B.1, while, under the \(H_{03}\) setting, the convergence rate accelerates due to singularity. This phenomenon has been reported by \cite{he2024adaptive}; that is, the null hypothesis corresponds to a composite structure that mixes different convergence rates as order of $n$.

\textbf{Results of Hypothesis Testing.} In the simulation experiments, we also compare our proposed test, SOMET, against four benchmark testing approaches of great popular in the current literature:  
(i) Sobel's test \citep{sobel1982AsymptoticConfidenceIntervals};  
(ii) Sobel's test with bootstrap (Sobel-Boot);  
(iii) MaxP test \citep{fritz2007RequiredSampleSize};  and 
(iv) MaxP test with bootstrap (MaxP-Boot).
To make a fair comparison in terms of computing, complexity and runtime, both bootstrap sample size and number of data splits in SOMET are equal to $500$. In all cases, the sample size $n$ is fixed at 2000.  

\textit{Type I Error Control.} For \(k = 1,2,3\), we evaluate Type~I error rates under the three null cases above \(H_{01}\), \(H_{02}\), and \(H_{03}\) in \eqref{eq:nullsimu}, which corresponds to, respectively, the composite null, \(\alpha_k \eta_{M_k \leadsto Y} = 0\). 
Figure~\ref{fig:qqplot_three_mediators} shows quantile-quantile plots in the $-\log_{10}$ scale for all these five tests over $2000$ replicates, where each row panel corresponds to a composite mediation pathway $S \to M_k \leadsto Y$, $k=1,2,3$, while each column represents different null hypothesis $H_{0k}$, $k = 1, 2, 3$. 

Figure~\ref{fig:qqplot_three_mediators} 
unveils the following highlights. 
(i) The four existing benchmark tests can control type I error under $H_{01}$ and $H_{02}$ in that the classical MaxP test appears overly conservative under $H_{01}$ for $k = 1$. 
(ii) Under $H_{03}$, all benchmark tests appear conservative, and only our SOMET shows an adequate type I error control for all $k = 1, 2, 3$. Since SOMET is built upon the Sobel's test, it clearly outperforms the classical Sobel test and the bootstrap-based Sobel test. Thus, the data-splitting strategy is deemed more desirable in the considered hypothesis setting.

\textit{Statistical Power.}  
To examine the statistical power of the proposed SOMET test under the alternative hypothesis $H_1$, we generate data for simplicity with $\xi_{2|1} = \xi_{3|1} = \xi_{3|2} = 0$ and the remaining parameters set by two scenarios: (i) varying $\alpha_k = \beta_k$ from $0$ to $0.4$;  
and (ii) fixing the composite mediation effect $\alpha_k\eta_{M_k\leadsto Y} =  \alpha_k \beta_k = 0.05$ while varying the ratio $\alpha_k / \beta_k$ from $0.5$ to $4$. Figure~\ref{fig:power_three_mediators} displays the empirical power ratio of each competing test versus the proposed test over $2000$ replicates. The left-to-right columns correspond to the pathway $S\to M_k\leadsto Y$, $k=1,2,3$. In almost all settings, the proposed SOMET test exhibits the highest statistical power. Of note, when $\alpha_k = 0$ in the first scenario, only SOMET provides an adequate type I error control. 

The results for probit and logistic outcomes are presented in Supplementary Section B.1.

\section{Applications}\label{section:real_data}
% \newpage
\subsection{Logistic SEM-SCOM in Assessing Risk of Diabetes Onset}
\label{sec:AllofUs}
This application utilizes the Fitbit device data from the NIH {\it All of Us} Research Program \citep{all2019all} to investigate if an exposure $(S)$ to family diabetes history (FDH) affects the risk of diabetes onset, DO (\(Y\)) and if the effect is mediated by physical activity measured via Fitbit wearable devices. Here exposure variable FDH is binary, with 1 indicating that the participant reports at least one diagnosis of diabetes in his or her family members (parent, grandparent, or sibling), and 0 otherwise. The outcome variable DO, is also binary, with 1 for the participant being diagnosed with diabetes. We are interested in evaluating if this FDH-DO relationship may be mediated by two temporally ordered mediators: baseline Body Mass Index (BMI), denoted as \(M_1\), measured as the average BMI over the two-year period preceding the physical activity window, and Physical Activity (PA), \(M_2\), recorded as the average daily step counts from wearable devices prior to the date of diabetes onset for cases, and over all available records of step counts for controls. This suggests a causal DAG depicted in Supplementary Figure B.8.
\par
There is a vast literature reporting that a family history of diabetes (FDH) significantly increases the risk of developing diabetes in offspring \citep{meigs2000parental, van2010role}. In particular, \citet{valdez2007family} reported that FDH is a strong and independent risk factor for type 2 diabetes, increasing risk by a factor of 2–6 fold depending on the number and closeness of affected relatives. In this analysis, we offer an in-depth risk analysis by understanding the causal driving factors between FDH exposure and diabetic onset. In particular, we plan to examine if and to what extent BMI and physical activity may potentially mediate this risk of $FDH\rightarrow DO$. The vector of confounding factors used in this analysis includes sex, race, age, smoking status, and alcohol consumption. We conducted a complete case analysis by removing participants with missing values in any of the variables used in the study. Additionally, we excluded individuals with outliers in either PA ($M_2$) or BMI ($M_1$). After the standard quality control procedures, the remaining data consists of 20,233 participants to be used in our analysis. Summary statistics of the variables are available in Supplementary Table C.4.

% \subsubsection{Results}
Our findings indicate that the dataset involves a rare outcome, with approximately 2\% of participants experiencing diabetic incidents among those with physical activity data. To address this low prevalence, we adopt a logistic outcome model. Specifically, a logistic SEM-SCOM(2) framework is deemed appropriate for this analysis. To ensure a balanced ratio of cases to controls in data splitting for the use of our SOMET method, we employ a stratified splitting strategy. We obtain the estimated direct effect, $\widehat{\tau}_1(FDH\rightarrow DO)=0.83$ in terms of log odds ratio or estimated odds ratio of 2.3. This suggests that individuals who come from a family with FDH to a family with no FDH have more than twice the odds of developing diabetes. This effect of FDH  on the risk of diabetes is highly statistically significant, with a $p$-value of \(10^{-13}\) by the Wald test.

To examine the composite mediation effects, we calculated the $p$-values of testing the presence of effects $\tau^{\dag}_1(S \rightarrow M_1 \leadsto Y)$ (via BMI) and $\tau^{\dag}_2(S \rightarrow M_2 \leadsto Y)$ (via PA). Table~\ref{tab:AOU_res} reports the $p$-values using our SOMET and other existing methods. We obtained estimated log-odds ratios, $\wh\tau^{\dag}_1(S \rightarrow M_1 \leadsto Y)=0.003$ ($p$-value=0.0006) and  $\wh\tau^{\dag}_2(S \rightarrow M_2 \leadsto Y)=0.013$ ($p$-value=0.0396) respectively. They suggest significant mediation effects of baseline BMI and PA prior to the disease onset. In other words, both baseline BMI and PA prior to the diabetes onset can mediate how FDH influences the risk of DO. In this analysis, all five methods reported positive discoveries.

\begin{table}[H]
\centering
\caption{P-values for the composite mediation effects of FDH on DO via BMI ($M_1$) and PA ($M_2$) across different methods for the null hypothesis of no composite mediation effect. X.Boot stands for method X is implemented by the bootstrap sampling method with 500 bootstrap replicates. SOMET was conducted using 500 rounds of data splitting.}
\begin{tabular}{lccccc}
\toprule
Mediator & Sobel & MaxP & Sobel.Boot & MaxP.Boot & SOMET \\
\midrule
BMI ($M_1$) & 0.0217 & 0.0024 & 0.0760 & 0.0175 & 0.0006   \\
PA ($M_2$)  & 0.0131 & 0.0054 & 0.005 & 0.0055 & 0.0396  \\
\bottomrule
\end{tabular}
\label{tab:AOU_res}
\end{table}

\subsection{Linear SEM-SCOM for Adolescent Obesity}
\label{sec:ELEMENT}
The second application illustrates a mediation analysis to examine potential causal mediation pathways through which prenatal exposure to blood lead (BP) may influence childhood body mass index (BMI). We analyze a dataset from the ELEMENT cohort \citep{perng2019EarlyLifeExposure}, comprising 288 adolescents (ages 10--18) residing in Mexico City. Prior studies e.g. \citep{wang2019AssociationMaternalExposure} have reported that elevated maternal lead exposure during pregnancy is associated with an increased risk of childhood overweight and obesity. The biological mechanisms underpinning this association remain largely unknown.

Maternal lead exposure may trigger certain intermediate biological processes impacting negatively on birth weight \citep{wild2023EffectPregnancyLipid} and fatty acid-based lipid metabolites \citep{chen2020AssociationMaternalBlood, koemel2022MaternalDietaryFatty}. Elevated birth weight has been found to be linked to a higher likelihood of abnormal childhood BMI \citep{kapral2018AssociationsBirthweightOverweight} leading to an increased risk of obesity. All these findings suggest a causal diagram (see Figure \ref{fig:dohad}), including exposure $S$ of maternal blood lead at delivery, one mediator $M_1$ of birth weight, another mediator $M_2$ of lipid-based fatty acid metabolites $M_2$, and outcome $Y$ of adolescent BMI. Clearly, these two mediators $M_1$ and $M_2$ are temporally or sequentially ordered, with $M_1$ preceding $M_2$. Here, sixty-five metabolites are considered as candidates of $M_2$ in the investigation. The vector of confounders $\boldsymbol W$ consists of eight baseline covariates; they are maternal age at delivery, parity, gestational age, educational attainment, marital status, child’s age and sex, and household socioeconomic status. Lipid metabolite concentrations are standardized into Z-scores prior to the analysis. Of note, missing values in variables $S$, $M_1$, and $M_2$ are handled by the hot-deck imputation method with ten-nearest-neighbors imputation using the R package \texttt{impute}.

We focus on examining three primary causal pathways: the direct effect $\tau_1(S \to Y)$, and two composite mediation effects, one via birth weight, $\tau_1^{\dag}(S \to M_1 \leadsto Y)$, and another via lipid metabolites, $\tau_2^{\dag}(S \to M_2 \to Y)$. In this analysis, we evaluate one lipid metabolite at one time in the analysis, and the multiplicity is adjusted by the Benjamini–Hochberg procedure \citep{benjamini1995ControllingFalseDiscovery} at the false discovery rate of 0.1. P-values are calculated using five methods: SOMET, MaxP, MaxP-Boot, Sobel, and Sobel-Boot where 1000 bootstrap replicates for Sobel-Boot and MaxP-Boot are used and 1000 rounds of data splitting for SOMET.

Our main findings are summarized as follows. Lipid FA 18:0 OH stands out as a significant $M_2$, whose analysis results are included in Table~\ref{table:real_data}. The lipid FA 18:0 OH has been reported in the literature to be associated with a higher risk of multiple islet autoimmunity in childhood \citep{niinisto2021ChildrensErythrocyteFatty}. Moreover, some epidemiological studies have found that higher dietary intake of stearic acid correlates with increased BMI, suggesting that exogenous sources of this fatty acid may contribute to body weight regulation \citep{raatz2017RelationshipReportedIntakes}. Our analysis aligns with these existing findings through the application of causal mediation analysis. The results of the non-significant metabolites are included in Supplementary Table C.4. Interestingly, our SOMET is the only test to discover the statistical significance of the mediation pathway through Lipid FA 18:0 OH.

\begin{comment}
    \pxs{Higher maternal intake or altered metabolism of stearic acid (18:0) may increase maternal BMI and modify fetal lipid exposure, potentially influencing birth weight. The same metabolic pathways that elevate hydroxystearic acid (18:0 OH) could also affect immune system maturation, linking prenatal lipid environment to both growth and later autoimmune susceptibility.} 
\end{comment}

\begin{comment}
    \begin{table}[H]
    \centering
    \caption{One significant lipid found in the mediation analysis of $\text{maternal blood lead}\to\text{birth weight}\to \text{lipid}\to \text{childhood BMI}$ at controlled FDR rate $0.1$. } 
    \label{table:real_data}
    % \begin{adjustbox}{max width=4\textwidth}
    \resizebox{\columnwidth}{!}{%
    \begin{tabular}{c|ccccc|ccccc|cccc}
        \toprule
        & \multicolumn{5}{c|}{$S\to M_1\leadsto Y$ (p-value)} & \multicolumn{5}{c|}{$S\to M_2\to Y$ (p-value)} & \multicolumn{4}{c}{Effects}\\
        Lipid & SOMET & JS & JS-B & PoC & PoC-B & SOMET & JS & JS-B & PoC & PoC-B & NIE ($M_1$) & NIE ($M_2$) & NDE & NTE \\ 
        \midrule
        FA 18:0 OH   & 0.8516 & 0.37 & 0.27 & 0.4456 & 0.285 & 0.0007 & 0.3884 & 0.342 & 0.4139 & 0.221 & 0.0061 & -0.0021 & 0.2614 & 0.2654 \\ 
        \bottomrule
    \end{tabular}
    }
    % \end{adjustbox}
\end{table}
\end{comment}
\begin{table}[H]
    \centering
    \caption{P-values for the composite mediation effects of maternal blood lead on childhood BMI via birth weight ($M_1$) and lipid ($M_2$) across different methods, at controlled FDR rate 0.1. X.Boot stands for method X is implemented by the bootstrap sampling method with 1000 bootstrap replicates. SOMET was conducted using 1000 rounds of data splitting.} 
    \label{table:real_data}
    \begin{tabular}{lccccc}
        \toprule
        Mediator & SOMET & MaxP & MaxP.Boot & Sobel & Sobel.Boot \\
        \midrule
        Birth weight ($M_1$) & 0.8516 & 0.3700 & 0.2700 & 0.4456 & 0.2850 \\
        Lipid ($M_2$)        & 0.0007 & 0.3884 & 0.3420 & 0.4139 & 0.2210 \\
        \bottomrule
    \end{tabular}
\end{table}

\section{Concluding Remarks}\label{sec-conc}
This paper develops a general mediation analysis framework for identifying, estimating, and testing pathway effects involving multiple sequentially ordered mediators. Such settings arise naturally in a variety of scientific domains, including epidemiology, systems biology, and longitudinal studies where exposure effects propagate through temporally structured intermediate variables on changes in outcomes. We establish nonparametric identification results and derive closed-form expressions for mediation pathway effects under the GLM models of both continuous and categorical outcomes.

A central methodological contribution given in this paper is a new strategy to deal with the challenge that the null space corresponding to the absence of pathway effects is composite, consisting of multiple non-nested parameter configurations. This composite nature of the null parameter space may restrain conventional inference procedures from yielding valid or optimal inference
in such an overly conservative type-I error control. To address this technical gap, we propose a new studentized testing procedure termed as \textit{Sequentially Ordered Mediation Effect Test (SOMET)}, which utilizes the data splitting technique to achieve robust variance estimation. As shown in this paper, our SOMET test can control Type I error under the composite null and outperform existing benchmark tests in terms of statistical power. Moreover we leverage an idea similar to Cauchy combination method to aggregate multiple independent $p$-values, each produced in one round of data split. The power can be further improved by repeated data splits. In the meantime, SOMET may increase computational cost, particularly when the number of mediators grows. 

Our proposed methodology is broadly applicable, straightforward to implement, and demonstrates strong performance across a variety of simulation scenarios. Nonetheless, several important directions remain for future research. One key avenue involves eliminating the step of data splitting, either through joint significance testing procedures or by employing some advanced adaptive bootstrap methods with no need of tuning parameter selection. Additionally, developing scalable alternatives to the Cauchy combination test could alleviate the computational burden associated with high-dimensional mediation settings. Extending the framework to incorporate time-varying exposures and dynamic feedback between mediators and outcomes represents another promising area for methodological advancement. In sum, this work provides a principled foundation for conducting rigorous, scalable, and interpretable mediation analysis inference in complex systems characterized by structured causal pathways in the presence of SCOM($K$).

\section*{Data and Code Availability Statement}\label{data-availability-statement}
\textbf{All of Us:} Patient confidentiality prevents the sharing of data publicly. The data underlying the results in Section~\ref{sec:AllofUs} are available from the All of Us Research Program at \url{https://www.researchallofus.org/register/} for researchers who meet the criteria for confidential data access.\\

\noindent
\textbf{ELEMENT:} Patient confidentiality prevents the sharing of data publicly. The data underlying the results in Section~\ref{sec:ELEMENT} are available from  \url{https://www.umich-element.org/} for researchers who meet the criteria for confidential data access via appropriate data use agreements.\\

\noindent
\textbf{Computing Code:} The code is available in the Github repo \url{https://github.com/Ritoban1/SOMET}.

\section*{Acknowledgment}
The research presented in this paper is funded by NIH R01ES033565.
\bibliographystyle{abbrvnat} 
\bibliography{references.bib}

@article{niinisto2021ChildrensErythrocyteFatty,
  title = {Children's Erythrocyte Fatty Acids Are Associated with the Risk of Islet Autoimmunity},
  author = {Niinist{\"o}, Sari and Erlund, Iris and Lee, Hye-Seung and Uusitalo, Ulla and Salminen, Irma and Aronsson, Carin Andr{\'e}n and Parikh, Hemang M. and Liu, Xiang and Hummel, Sandra and Toppari, Jorma and She, Jin-Xiong and Lernmark, {\AA}ke and Ziegler, Annette G. and Rewers, Marian and Akolkar, Beena and Krischer, Jeffrey P. and Galas, David and Das, Siba and Sakhanenko, Nikita and Rich, Stephen S. and Hagopian, William and Norris, Jill M. and Virtanen, Suvi M.},
  year = {2021},
  month = feb,
  journal = {Scientific Reports},
  volume = {11},
  number = {1},
  pages = {3627},
  publisher = {Nature Publishing Group},
  urldate = {2025-04-23},
  copyright = {2021 The Author(s)},
  langid = {english}
}

@article{raatz2017RelationshipReportedIntakes,
  title = {Relationship of the Reported Intakes of Fat and Fatty Acids to Body Weight in US Adults},
  author = {Raatz, Susan K and Conrad, Zach and Johnson, LuAnn K and Picklo, Matthew J and Jahns, Lisa},
  year = {2017},
  month = apr,
  journal = {Nutrients},
  volume = {9},
  number = {5},
  pages = {438},
  urldate = {2025-04-23},
  pmcid = {PMC5452168},
  pmid = {28452961}
}

@article{kapral2018AssociationsBirthweightOverweight,
  title = {Associations between Birthweight and Overweight and Obesity in School-Age Children},
  author = {Kapral, Nicole and Miller, Sarah E. and Scharf, Rebecca J. and Gurka, Matthew J. and DeBoer, Mark D.},
  year = {2018},
  month = jun,
  journal = {Pediatric obesity},
  volume = {13},
  number = {6},
  pages = {333--341},
  urldate = {2025-04-23},
  pmcid = {PMC5756526},
  pmid = {28685963}
}

@article{wang2019AssociationMaternalExposure,
  title = {Association Between Maternal Exposure to Lead, Maternal Folate Status, and Intergenerational Risk of Childhood Overweight and Obesity},
  author = {Wang, Guoying and DiBari, Jessica and Bind, Eric and Steffens, Andrew M. and Mukherjee, Jhindan and Azuine, Romuladus E. and Singh, Gopal K. and Hong, Xiumei and Ji, Yuelong and Ji, Hongkai and Pearson, Colleen and Zuckerman, Barry S. and Cheng, Tina L. and Wang, Xiaobin},
  year = {2019},
  month = oct,
  journal = {JAMA Network Open},
  volume = {2},
  number = {10},
  pages = {e1912343},
  urldate = {2025-04-23},
  pmcid = {PMC6777254},
  pmid = {31577354}
}

@article{benjamini1995ControllingFalseDiscovery,
  title = {Controlling the False Discovery Rate: A Practical and Powerful Approach to Multiple Testing},
  shorttitle = {Controlling the {{False Discovery Rate}}},
  author = {Benjamini, Yoav and Hochberg, Yosef},
  year = {1995},
  month = jan,
  journal = {Journal of the Royal Statistical Society: Series B (Methodological)},
  volume = {57},
  number = {1},
  pages = {289--300},
  urldate = {2022-10-25},
  langid = {english}
}

@article{koemel2022MaternalDietaryFatty,
  title = {Maternal Dietary Fatty Acid Composition and Newborn Epigenetic Aging---a Geometric Framework Approach},
  author = {Koemel, Nicholas A and Senior, Alistair M and Dissanayake, Hasthi U and Ross, Jason and McMullan, Rowena L and Kong, Yang and Phang, Melinda and Hyett, Jon and Raubenheimer, David and Gordon, Adrienne and Simpson, Stephen J and Skilton, Michael R},
  year = {2022},
  month = jan,
  journal = {The American Journal of Clinical Nutrition},
  volume = {115},
  number = {1},
  pages = {118--127},
  urldate = {2025-04-22}
}

@article{chen2020AssociationMaternalBlood,
  title = {Association between Maternal Blood Lipids Levels during Pregnancy and Risk of Small-for-Gestational-Age Infants},
  author = {Chen, Qinqing and Chen, Huiqi and Xi, Fangfang and Sagnelli, Matthew and Zhao, Baihui and Chen, Yuan and Yang, Mengmeng and Xu, Dong and Jiang, Ying and Chen, Guangdi and Luo, Qiong},
  year = {2020},
  month = nov,
  journal = {Scientific Reports},
  volume = {10},
  pages = {19865},
  urldate = {2025-04-22},
  pmcid = {PMC7669834},
  pmid = {33199750}
}

@incollection{wild2023EffectPregnancyLipid,
  title = {Effect of Pregnancy on Lipid Metabolism and Lipoprotein Levels},
  booktitle = {Endotext [Internet]},
  author = {Wild, Robert and Feingold, Kenneth R.},
  year = {2023},
  month = mar,
  publisher = {MDText.com, Inc.},
  urldate = {2025-04-22},
  langid = {english},
  pmid = {29714937}
}

@article{perng2019EarlyLifeExposure,
  title = {Early Life Exposure in Mexico to ENvironmental Toxicants ({{ELEMENT}}) Project},
  author = {Perng, Wei and {Tamayo-Ortiz}, Marcela and Tang, Lu and S{\'a}nchez, Brisa N. and Cantoral, Alejandra and Meeker, John D. and others},
  year = {2019},
  month = aug,
  journal = {BMJ Open},
  volume = {9},
  number = {8},
  pages = {e030427},
  publisher = {{British Medical Journal Publishing Group}},
  urldate = {2022-11-27},
  chapter = {Public health},
  copyright = {\textcopyright{} Author(s) (or their employer(s)) 2019. Re-use permitted under CC BY-NC. No commercial re-use. See rights and permissions. Published by BMJ.. This is an open access article distributed in accordance with the Creative Commons Attribution Non Commercial (CC BY-NC 4.0) license, which permits others to distribute, remix, adapt, build upon this work non-commercially, and license their derivative works on different terms, provided the original work is properly cited, appropriate credit is given, any changes made indicated, and the use is non-commercial. See:~http://creativecommons.org/licenses/by-nc/4.0/.},
  langid = {english},
  pmid = {31455712}
}

@article{avin2005identifiability,
  title={Identifiability of path-specific effects},
  author={Avin, Chen and Shpitser, Ilya and Pearl, Judea},
  year={2005}
}

@article{daniel2015causal,
  title={Causal mediation analysis with multiple mediators},
  author={Daniel, Rhian M and De Stavola, Bianca L and Cousens, Simon N and Vansteelandt, Stijn},
  journal={Biometrics},
  volume={71},
  number={1},
  pages={1--14},
  year={2015},
  publisher={Wiley Online Library}
}

@article{lin2017interventional,
  title={Interventional approach for path-specific effects},
  author={Lin, Sheng-Hsuan and VanderWeele, Tyler},
  journal={Journal of Causal Inference},
  volume={5},
  number={1},
  pages={20150027},
  year={2017},
  publisher={De Gruyter}
}

@article{zhou2022semiparametric,
  title={Semiparametric estimation for causal mediation analysis with multiple causally ordered mediators},
  author={Zhou, Xiang},
  journal={Journal of the Royal Statistical Society Series B: Statistical Methodology},
  volume={84},
  number={3},
  pages={794--821},
  year={2022},
  publisher={Oxford University Press}
}

@article{he2024adaptive,
  title={Adaptive bootstrap tests for composite null hypotheses in the mediation pathway analysis},
  author={He, Yinqiu and Song, Peter XK and Xu, Gongjun},
  journal={Journal of the Royal Statistical Society Series B: Statistical Methodology},
  volume={86},
  number={2},
  pages={411--434},
  year={2024},
  publisher={Oxford University Press US}
}

@article{chen2024quantile,
  title={Quantile Mediation Analytics},
  author={Chen, Canyi and He, Yinqiu and Wang, Huixia J and Xu, Gongjun and Song, Peter X-K},
  journal={arXiv preprint arXiv:2412.15401},
  year={2024}
}

@article{vanderweele2010odds,
  title={Odds ratios for mediation analysis for a dichotomous outcome},
  author={VanderWeele, Tyler J and Vansteelandt, Stijn},
  journal={American journal of epidemiology},
  volume={172},
  number={12},
  pages={1339--1348},
  year={2010},
  publisher={Oxford University Press}
}

@article{gaynor2019mediation,
  title={Mediation analysis for common binary outcomes},
  author={Gaynor, Sheila M and Schwartz, Joel and Lin, Xihong},
  journal={Statistics in medicine},
  volume={38},
  number={4},
  pages={512--529},
  year={2019},
  publisher={Wiley Online Library}
}

@article{sobel1982AsymptoticConfidenceIntervals,
  title = {Asymptotic Confidence Intervals for Indirect Effects in Structural Equation Models},
  author = {Sobel, Michael E.},
  year = {1982},
  journal = {Sociological Methodology},
  volume = {13},
  eprint = {270723},
  eprinttype = {jstor},
  pages = {290},
  urldate = {2023-07-19}
}

@article{fritz2007RequiredSampleSize,
  ids = {fritz2007RequiredSampleSizea},
  title = {Required {{Sample Size}} to {{Detect}} the {{Mediated Effect}}},
  author = {Fritz, Matthew S. and MacKinnon, David P.},
  year = {2007},
  month = mar,
  journal = {Psychological science},
  volume = {18},
  number = {3},
  pages = {233--239},
  urldate = {2023-09-10},
  langid = {english},
  pmcid = {PMC2843527},
  pmid = {17444920}
}

@article{all2019all,
  title = {The "All of Us" Research Program},
  author = {{All of Us Research Program Investigators}},
  journal = {New England Journal of Medicine},
  volume = {381},
  number = {7},
  pages = {668--676},
  year = {2019},
  publisher = {Massachusetts Medical Society}
}

@article{van2010role,
  title={Role of adiposity and lifestyle in the relationship between family history of diabetes and 20-year incidence of type 2 diabetes in US women},
  author={van't Riet, Esther and Dekker, Jacqueline M and Sun, Qi and Nijpels, Giel and Hu, Frank B and van Dam, Rob M},
  journal={Diabetes care},
  volume={33},
  number={4},
  pages={763--767},
  year={2010},
  publisher={American Diabetes Association}
}

@article{meigs2000parental,
  title={Parental transmission of type 2 diabetes: the Framingham Offspring Study.},
  author={Meigs, James B and Cupples, L Adrienne and Wilson, PW},
  journal={Diabetes},
  volume={49},
  number={12},
  pages={2201--2207},
  year={2000}
}

@article{valdez2007family,
  title={Family history and prevalence of diabetes in the US population: the 6-year results from the National Health and Nutrition Examination Survey (1999--2004)},
  author={Valdez, Rodolfo and Yoon, Paula W and Liu, Tiebin and Khoury, Muin J},
  journal={Diabetes care},
  volume={30},
  number={10},
  pages={2517--2522},
  year={2007},
  publisher={American Diabetes Association}
}

@article{mackinnon2002comparison,
  title={A comparison of methods to test mediation and other intervening variable effects.},
  author={MacKinnon, David P and Lockwood, Chondra M and Hoffman, Jeanne M and West, Stephen G and Sheets, Virgil},
  journal={Psychological methods},
  volume={7},
  number={1},
  pages={83},
  year={2002},
  publisher={American Psychological Association}
}

@article{barfield2017testing,
  title={Testing for the indirect effect under the null for genome-wide mediation analyses},
  author={Barfield, Richard and Shen, Jincheng and Just, Allan C and Vokonas, Pantel S and Schwartz, Joel and Baccarelli, Andrea A and VanderWeele, Tyler J and Lin, Xihong},
  journal={Genetic epidemiology},
  volume={41},
  number={8},
  pages={824--833},
  year={2017},
  publisher={Wiley Online Library}
}

@article{dai2022multiple,
  title={A multiple-testing procedure for high-dimensional mediation hypotheses},
  author={Dai, James Y and Stanford, Janet L and LeBlanc, Michael},
  journal={Journal of the American Statistical Association},
  volume={117},
  number={537},
  pages={198--213},
  year={2022},
  publisher={Taylor \& Francis}
}

@article{liu2022large,
  title={Large-scale hypothesis testing for causal mediation effects with applications in genome-wide epigenetic studies},
  author={Liu, Zhonghua and Shen, Jincheng and Barfield, Richard and Schwartz, Joel and Baccarelli, Andrea A and Lin, Xihong},
  journal={Journal of the American Statistical Association},
  volume={117},
  number={537},
  pages={67--81},
  year={2022},
  publisher={Taylor \& Francis}
}

@article{liu2025simple,
  title={A simple and powerful method for large-scale composite null hypothesis testing with applications in mediation analysis},
  author={Liu, Yaowu},
  journal={Biometrics},
  volume={81},
  number={1},
  pages={ujaf011},
  year={2025},
  publisher={Oxford University Press}
}

@article{roy2024subsampling,
  title={Subsampling-based Tests in Mediation Analysis},
  author={Roy, Asmita and Zhou, Huijuan and Zhao, Ni and Zhang, Xianyang},
  journal={arXiv preprint arXiv:2411.10648},
  year={2024}
}

@article{halfon2014lifecourse,
  title={Lifecourse health development: past, present and future},
  author={Halfon, Neal and Larson, Kandyce and Lu, Michael and Tullis, Ericka and Russ, Shirley},
  journal={Maternal and child health journal},
  volume={18},
  pages={344--365},
  year={2014},
  publisher={Springer}
}

@article{heckman2006case,
  title={The case for investing in early childhood},
  author={Heckman, James and Tremblay, Richard},
  journal={A snapshot of research by University of Chicago, USA \& University of Montreal, Canada. Sydney: The Smith Family Research and Development},
  year={2006}
}

@article{newey1994large,
  title={Large sample estimation and hypothesis testing},
  author={Newey, Whitney K and McFadden, Daniel},
  journal={Handbook of econometrics},
  volume={4},
  pages={2111--2245},
  year={1994},
  publisher={Elsevier}
}
\newpage
\begin{figure}[H]
    \centering
    \includegraphics[width=0.7\linewidth]{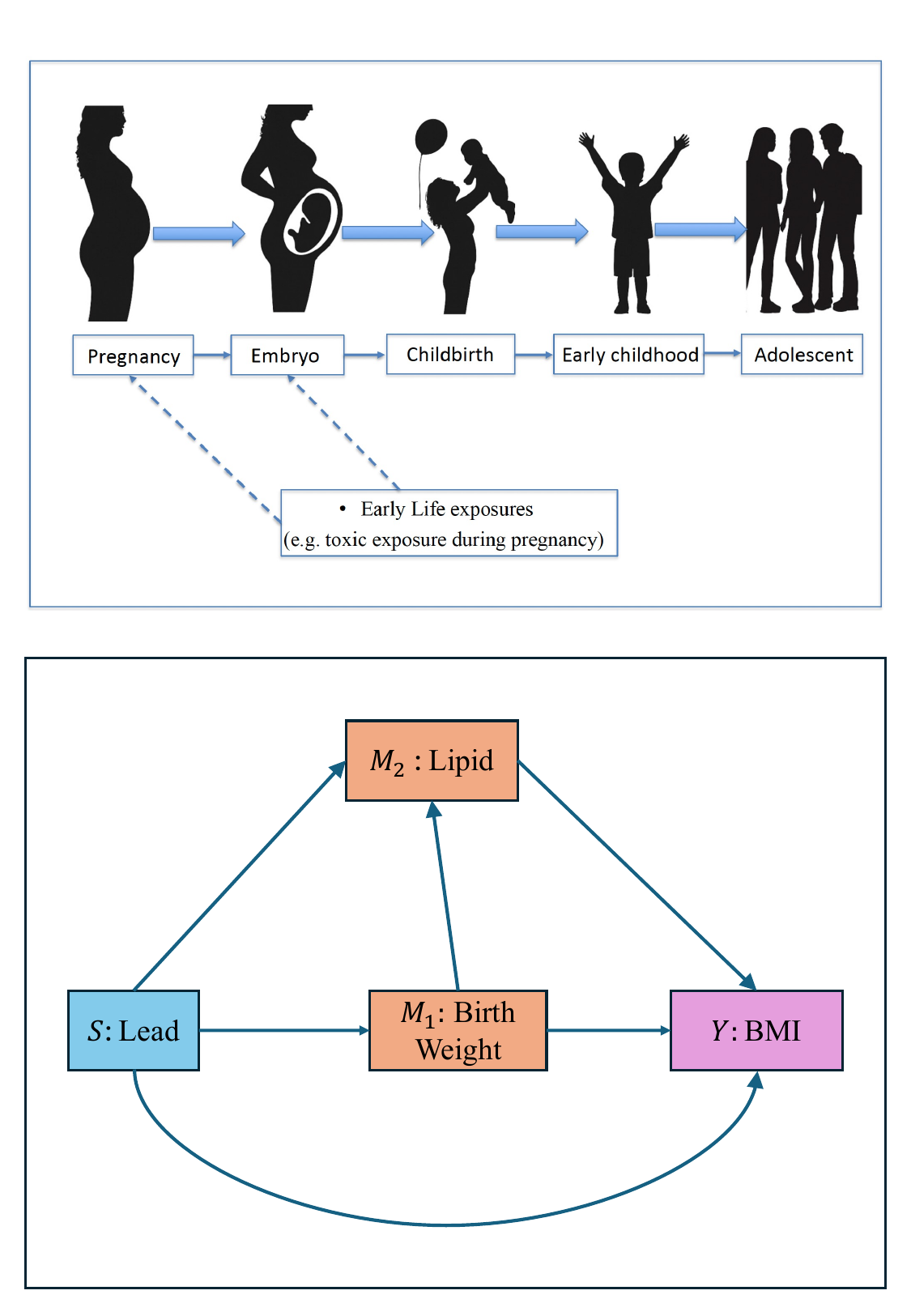}
    \caption{The top panel depicts the Developmental Origins of Health and Disease (DOHaD) hypothesis that involves some key milestones of human growth and development: pregnancy, birth, early childhood, adolescence, and adulthood. During such journey early-life exposures, such as those occurring during pregnancy, potentially shape health trajectories across the lifespan. The bottom panel shows a directed acyclic graph with two sequentially ordered mediators, $M_1$ (birth weight) and $M_2$ (lipids during adolescence) through which prenatal exposure $S$ (lead) influences outcome of BMI $Y$.}
    \label{fig:dohad}
\end{figure}

\begin{figure}[H]
	\centerline{%\renewcommand{\arraystretch}{0.8} %<- modify value to suit your needs
		\begin{tabular}{c}			
       \includegraphics[width=7in]{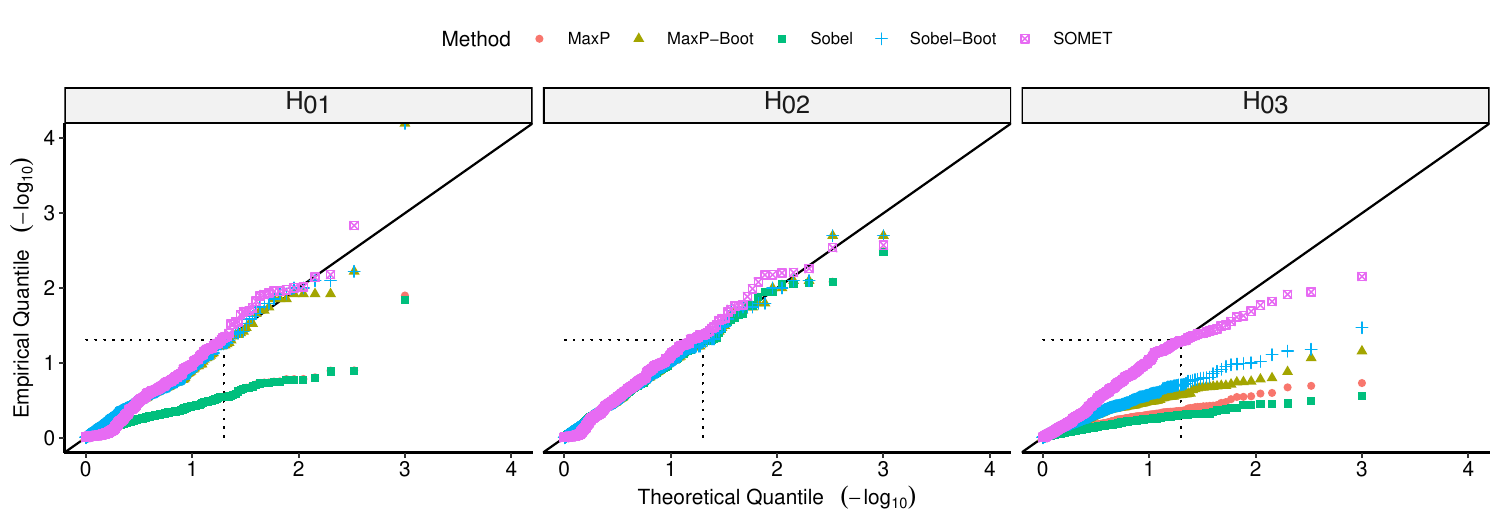}\\
        (a) $S\to M_1\leadsto Y$\\
        \includegraphics[width=7in]{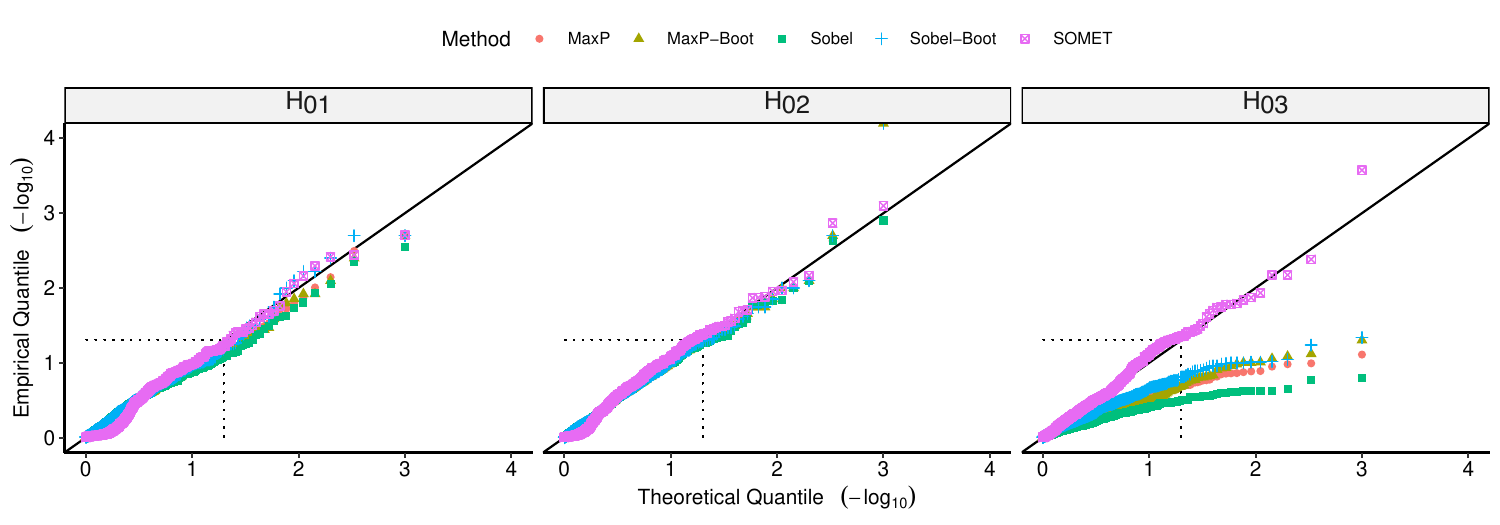}\\
        (b) $S\to M_2\leadsto Y$\\
        \includegraphics[width=7in]{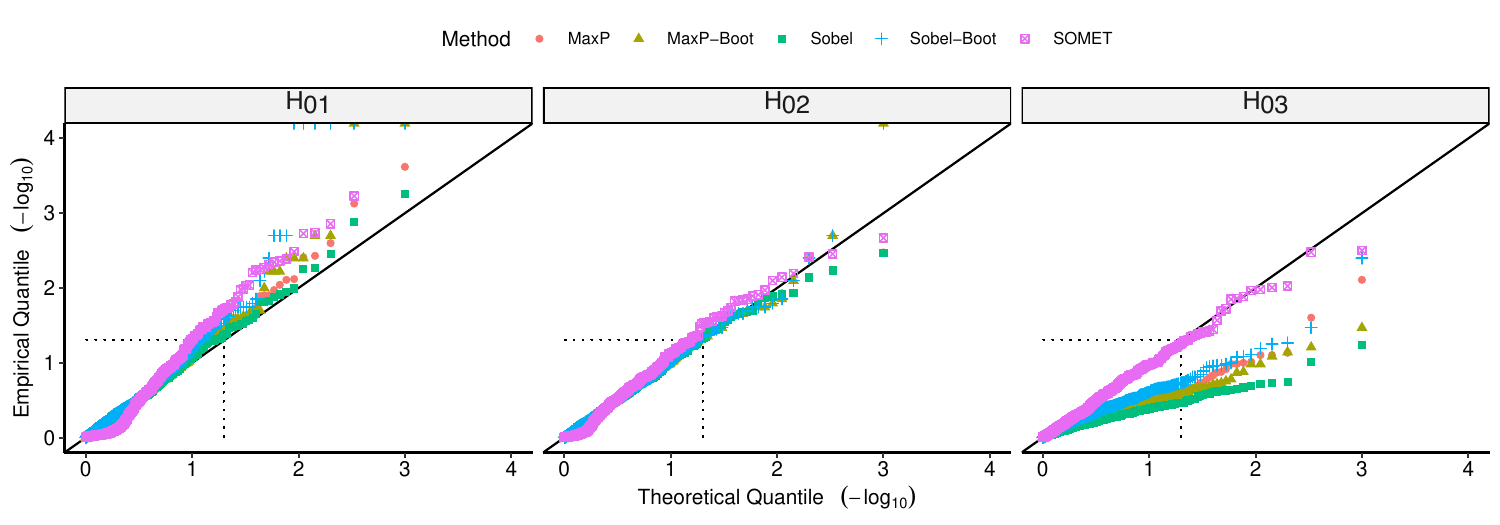}\\
        (c) $S \to M_3\leadsto Y$
		\end{tabular}
	}
	\captionsetup{font = footnotesize}
	\caption{Quantile-quantile plots of $p$-values obtained by five tests over 500 replicates under the three nulls, where data are generated by the linear SCOM(3) model and sample size $n = 2000$. The dotted line marks the $0.05$ nominal level.}
	\label{fig:qqplot_three_mediators}
\end{figure}
\begin{figure}[H]
	\centerline{%\renewcommand{\arraystretch}{0.8} %<- modify value to suit your needs
		\begin{tabular}{c}			
                \includegraphics[width=7in]{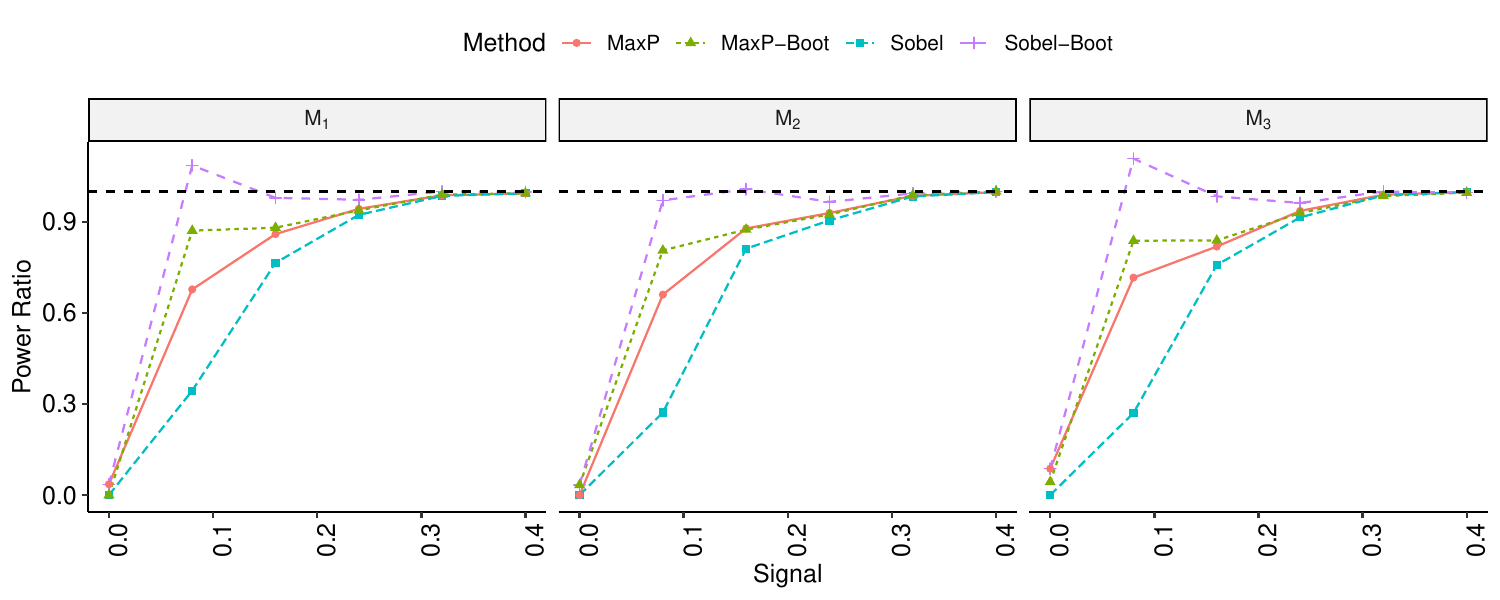}\\
        (a)  Signal strength  $\alpha_k = \beta_k$ \\
                \includegraphics[width=7in]{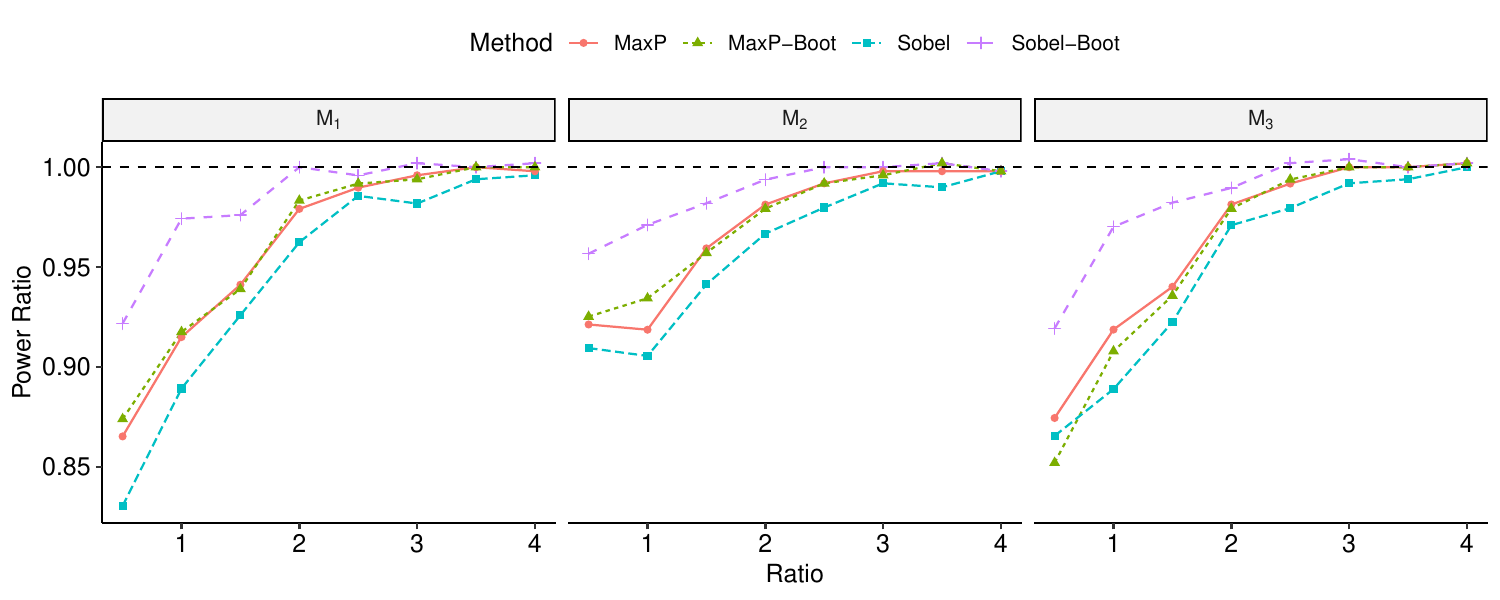}\\
        (b) Ratio $\alpha_k/\beta_k$
		\end{tabular}
	}
	\captionsetup{font = footnotesize}
	\caption{Empirical power ratio of the competing test versus the proposed test SOMET over the signal strength  $\alpha_k = \beta_k$ in panel (a) and over the ratio $\alpha_k/\beta_k$ in panel (b). The composite mediation pathway $S\rightarrow M_k\leadsto Y$, $k=1,2,3$ is tested from the left ($k=1$) to right ($k=3$).}
	\label{fig:power_three_mediators}
\end{figure}

\newpage
\section*{Supplementary Materials}
\appendix
\renewcommand{\thesection}{\Alph{section}}
\renewcommand{\thetheorem}{\Alph{section}.\arabic{theorem}}
\renewcommand{\theequation}{\Alph{section}.\arabic{equation}}
\renewcommand{\thefigure}{\Alph{section}.\arabic{figure}}
\renewcommand{\thetable}{\Alph{section}.\arabic{table}}

\section{Proof of Main Results}
\subsection{Proof of Theorem 1}\label{sec:prthm1}
\begin{proof}
$\Gamma(s_0, \boldsymbol{s}|\boldsymbol W)$ can be written as
   \begin{align*}
       &\mathbb{E}[Y(s_0, \overline{M}_K(\boldsymbol{s}))|\boldsymbol W]=\int_{\text{supp}(\overline{\boldsymbol M}_{K})}\mathbb{E}[Y(s_0, \overline{\boldsymbol m}_{K})| \overline{M}_K(\boldsymbol{s})=\overline{\boldsymbol m}_{K},\boldsymbol W]dF_{\overline{M}_K(\boldsymbol{s})}(\overline{\boldsymbol m}_{K}|\boldsymbol W)
   \end{align*}
Using Assumption 4(v) the above expression simplifies to,
\begin{align*}
    &\int_{\text{supp}(\overline{\boldsymbol M}_{K})}\mathbb{E}[Y(s_0, \overline{\boldsymbol m}_{K})|\boldsymbol W]dF_{\overline{M}_K(\boldsymbol{s})|\boldsymbol W}(m_1,\cdots,m_K)=\\
    & \int_{\text{supp}(M_1)}\left[\cdots\left[\int_{\text{supp}(M_K)}\mathbb{E}[Y(s_0, \overline{\boldsymbol m}_{K})|\boldsymbol W]dF_{M_K(s_K,\overline{\boldsymbol m}_{K-1})}(m_K|M_{K-1}(s_{K-1},\overline{\boldsymbol m}_{K-2}),\cdots,M_1(s_1),\boldsymbol W)\right]\right.\\
    &\hspace{4cm}\left.\cdots\right]dF_{M_1(s_1)}(m_1|\boldsymbol W)  
\end{align*}
Using Assumptions 4(i) and 4(ii) to the first part of the integral which consists of potential outcome distribution and Assumptions 4(iv), then 4(iii) to mediator distributions, the above expressions simplifies to
\begin{align*}
    &\int_{\text{supp}(M_1)}\left[\cdots\left[\int_{\text{supp}(M_K)}\mathbb{E}[Y(s_0, \overline{\boldsymbol m}_{K})|\overline{\boldsymbol M}_{K}=\overline{\boldsymbol m}_{K},S=s_0,\boldsymbol W]\right.\right.\\
    &\quad\left.\left.dF_{M_K(s_K,\overline{\boldsymbol m}_{K-1})}(m_K|\overline{\boldsymbol M}_{K-1}=\overline{\boldsymbol m}_{K-1},S=s_K,\boldsymbol W)\right]\cdots\right]dF_{M_1(s_1)}(m_1|S=s_1,\boldsymbol W)  
\end{align*}
Using the Assumption 3, we obtain the final expression as,
\begin{align*}
    &\int_{\text{supp}(M_1)}\left[\cdots\left[\int_{\text{supp}(M_K)}\mathbb{E}[Y|\overline{\boldsymbol M}_{K}=\overline{\boldsymbol m}_{K},S=s_0,\boldsymbol W]dF_{M_K}(m_K|\overline{\boldsymbol M}_{K-1}=\overline{\boldsymbol m}_{K-1},S=s_K,\boldsymbol W)\right]\cdots\right]\\
    & \hspace{4cm}dF_{M_1}(m_1|S=s_1,\boldsymbol W)\\
    &=\int_{\text{supp}(\overline{\boldsymbol M}_{K})}\mathbb{E}(Y \mid S=s_0, \overline{\boldsymbol M}_{K}=\overline{\boldsymbol m}_{K}, \boldsymbol W)  
\prod_{k=1}^K
dF_{M_k}(m_k\mid S=s_k, \overline{\boldsymbol M}_{k-1}=\overline{\boldsymbol m}_{k-1}, \boldsymbol W)
\end{align*}
This proves Theorem 1.
\end{proof}
\subsection{Proof of Corollary 1}\label{sec:prco1}
\begin{proof}
    First we start with the innermost integral in Theorem 1 which is given by:
\begin{align*}
    &\int_{\text{supp}(M_K)}\mathbb{E}[Y|S=s_0,\overline{\boldsymbol M}_{K}=\overline{\boldsymbol m}_{K},\boldsymbol W]dF_{M_K}(m_K|S=s_K,\overline{\boldsymbol M}_{K-1}=\overline{\boldsymbol m}_{K-1},\boldsymbol W)\\
    &=\int_{\text{supp}(M_k)}\left(\gamma_S s_0+\sum_{k=1}^K\beta_{k}m_k +\boldsymbol\delta_{Y}^T\boldsymbol W\right)dF_{M_K}(m_K|S=s_K,\overline{\boldsymbol M}_{K-1}=\overline{\boldsymbol m}_{K-1},\boldsymbol W)\\
    &=\mathbb{E}_{M_k}\left(\gamma_S s_0+\sum_{k=1}^K\beta_{k}m_k +\boldsymbol\delta_{Y}^T\boldsymbol W|S=s_K,\overline{\boldsymbol M}_{K-1}=\overline{\boldsymbol m}_{K-1},\boldsymbol W\right)\\
    &=\gamma_S s_0+\sum_{k=1}^{K-1}\beta_{k}m_k +\boldsymbol\delta_{Y}^T\boldsymbol W + \beta_{K}(\alpha_K s_K + \sum_{k=1}^{K-1}\xi_{K|k}m_k+\boldsymbol\delta_{M_K}^T\boldsymbol W)\\
    &=(\gamma_S s_0 + \beta_{K}\alpha_K s_K)+\sum_{k=1}^{K-1}\left(\beta_{k} + \beta_{K}\xi_{K|k}\right)m_k+ (\boldsymbol\delta_{Y}^T+\beta_{K}\boldsymbol\delta_{M_K}^T)\boldsymbol W
\end{align*}
To derive the subsequent integrals, we first define \( \eta_{M_k \leadsto Y} \) as the collapsed effect from \( M_k \) to \( Y \), encompassing all possible pathways through which \( M_k \) influences \( Y \) for \( k = 1,2, \dots, K \). Hence $\eta_{M_K \leadsto Y}=\beta_{K}$. Next we derive the expression of $\eta_{M_{K-1} \leadsto Y}$
\begin{align*}
    \eta_{M_{K-1} \leadsto Y}=\beta_{K}+\xi_{K|K-1}\eta_{M_K\leadsto Y}
\end{align*}
Similarly for any general $k<K$, one can write the expression for  $\eta_{M_k \leadsto Y}$ as
\begin{align*}
    \eta_{M_k \leadsto Y}=\beta_{k} + \sum_{j=k+1}^K\xi_{j|k}\eta_{M_j\leadsto Y}
\end{align*}
Hence using this notation and the above technique for subsequent integrals, we finally obtain the expression for $\Gamma(s_0, \boldsymbol{s}|\boldsymbol W)$ as
\begin{align*}
  \Gamma(s_0, \boldsymbol{s}|\boldsymbol W)=\gamma_S s_0+ \sum_{k=1}^K\alpha_k\eta_{M_k\leadsto Y}s_k + \left(\boldsymbol\delta_{Y}^T+\sum_{k=1}^K \eta_{M_k\leadsto Y}\boldsymbol\delta_{M_k}^T\right)\boldsymbol W \cdot
\end{align*}
\end{proof}
\subsection{Proof of Corollary 2}\label{sec:prcr2}
\begin{proof}
 We start with the expression of $\Gamma(s_0, \boldsymbol{s}|\boldsymbol W)$. The innermost integral in $\Gamma(s_0, \boldsymbol{s}|\boldsymbol W)$ is simplified by,
\begin{align*}
    &\int_{\text{supp}(M_K)}\text{P}[Y=1|S=s_0,M_1=m_1,\cdots,M_K=m_K,\boldsymbol W]dF_{M_K}(m_K|S=s_K,\overline{\boldsymbol M}_{K-1}=\overline{\boldsymbol m}_{K-1},\boldsymbol W)\\
    &=\int_{\text{supp}(M_K)}\Phi\left(\gamma_S s_0+\sum_{k=1}^K\beta_{k}m_k +\boldsymbol\delta_{Y}^T\boldsymbol w\right)f_{M_K}(m_K|S=s_K,\overline{\boldsymbol M}_{K-1}=\overline{\boldsymbol m}_{K-1},\boldsymbol W)dm_K
\end{align*}
Let $\mu_{M_k}=\alpha_k s_k + \sum_{j=1}^{k-1}\xi_{k|j}m_j+\boldsymbol\delta_{M_k}^T\boldsymbol W$, for every $k>1$ and $\mu_{1}=\alpha_1 s_1+\boldsymbol\delta_{M_1}^T\boldsymbol W$. Let $\phi(.)$ denote the pdf of a standard normal distribution. Using these notations, the above integral can be written as,
\begin{align*}
    \frac{1}{\sigma_{M_k}}\int_{\text{supp}(M_K)}\Phi\left(\gamma_S s_0+\sum_{k=1}^K\beta_{k}m_k +\boldsymbol\delta_{Y}^T\boldsymbol W\right)\phi\left(\frac{m_K-\mu_{M_K}}{\sigma_{M_K}}\right)dm_K
\end{align*}
Let $z=\frac{m_K-\mu_{M_K}}{\sigma_{M_K}}$. Hence $m_K=\mu_{M_K}+\sigma_{M_K}z$. Using this substitution, one can write the above integral as,
\begin{align*}
\int_{z}\Phi\left(\gamma_S s_0+\sum_{k=1}^{K-1}\beta_{k}m_k + \beta_{K}\mu_{M_K}+\boldsymbol\delta_{Y}^T\boldsymbol W+\beta_{K}\sigma_{M_K}z\right)\phi(z)dz
\end{align*}
By a well known result, we obtain that $\int_{-\infty}^{\infty}\Phi(a+bz)\phi(z)dz=\Phi\left(\frac{a}{\sqrt{1+b^2}}\right)$. Using this result, the above integral equals,
\begin{align*}
    &\Phi\left(\frac{\gamma_S s_0 + \sum_{k=1}^{K-1} \beta_{k} m_k +\beta_{K}\mu_{M_K}+ \boldsymbol{\delta}_Y^T \boldsymbol{W}}{\sqrt{1 + \beta_{K}^2 \sigma_{M_K}^2}}\right)\\
    &=\Phi\left(\frac{\gamma_S s_0 + \sum_{k=1}^{K-1} \beta_{k} m_k +\beta_{K}(\alpha_K s_K + \sum_{k=1}^{K-1}\xi_{K|k}m_k+\boldsymbol\delta_{M_K}^T\boldsymbol W)+ \boldsymbol{\delta}_Y^T \boldsymbol{W}}{\sqrt{1 + \beta_{K}^2 \sigma_{M_K}^2}}\right)\\
    &=\Phi\left(\frac{\gamma_S s_0 +\beta_{K}\alpha_K s_K + \sum_{k=1}^{K-1}(\beta_{k} +\beta_{K}\xi_{K|k})m_k+(\boldsymbol{\delta}_Y^T+\beta_{K}\delta_{M_K}^T)\boldsymbol W}{\sqrt{1 + \beta_{K}^2 \sigma_{M_K}^2}}\right)
\end{align*}
Hence using the same technique for subsequent integrals that we employed for the previous two cases, the final expression for  $\Gamma(s_0, \boldsymbol{s})$ is obtained as,
\begin{align*}
    \Gamma(s_0, \boldsymbol{s})=\Phi\left(\frac{\gamma_S s_0+ \sum_{k=1}^K\alpha_k\eta_{M_k\leadsto Y}s_k + \left(\boldsymbol\delta_{Y}^T+\sum_{k=1}^K \eta_{M_k\leadsto Y}\boldsymbol\delta_{M_k}^T\right)\boldsymbol W }{\sqrt{1+\sum_{k=1}^K\eta_{M_k\leadsto Y}^2\sigma^2_{M_{k}}}}\right)
\end{align*}   
\end{proof}
\subsection{Proof of Corollary 3}\label{sec:prcr3}
\begin{proof}
    To derive the expression of $\Gamma(s_0, \boldsymbol{s}|\boldsymbol W)$, we start simplifying with innermost integral given by,
\begin{align*}
   &\int_{\text{supp}(M_k)}\text{P}[Y=1|S=s_0,\overline{\boldsymbol M}_{K}=\overline{\boldsymbol m}_{K},\boldsymbol W]dF_{M_K}(m_K|S=s_K,\overline{\boldsymbol M}_{K-1}=\overline{\boldsymbol m}_{K-1},\boldsymbol W)\\
   &=\int_{\text{supp}(M_k)}\text{exp}(\gamma_S s_0+\sum_{k=1}^K\beta_{k}m_k +\boldsymbol\delta_{Y}^T\boldsymbol W)f_{M_K}(m_K|S=s_K,\overline{\boldsymbol M}_{K-1}=\overline{\boldsymbol m}_{K-1},\boldsymbol W)dm_K\\
   &=\text{exp}(\gamma_S s_0+\sum_{k=1}^{K-1}\beta_{M_{k}}m_{k} +\boldsymbol\delta_{Y}^T\boldsymbol w)\int_{\text{supp}(M_k)}\text{exp}(\beta_{M_{K}}m_{K})f_{M_K}(m_K|S=s_K,\overline{\boldsymbol M}_{K-1}=\overline{\boldsymbol m}_{K-1},\boldsymbol W)dm_K\\
   &=\text{exp}\left[\gamma_S s_0+\sum_{k=1}^{K-1}\beta_{M_{k}}m_{k} +\boldsymbol\delta_{Y}^T\boldsymbol W+\beta_{K}\alpha_K s_K + \sum_{k=1}^{K-1}\beta_{M_{K}}\xi_{K|k}m_k+\beta_{M_{K}}\boldsymbol\delta_{M_K}^T\boldsymbol W+\frac{1}{2}\beta^2_{M_{K}}\sigma^2_{M_{K}}\right]\\
   &=\text{exp}\left[(\gamma_S s_0 + \beta_{K}\alpha_K s_K) +\sum_{k=1}^{K-1}\left(\beta_{k} + \beta_{K}\xi_{K|k}\right)m_k+ (\boldsymbol\delta_{Y}^T+\beta_{K}\boldsymbol\delta_{M_K}^T)\boldsymbol W+\frac{1}{2}\beta^2_{M_{K}}\sigma^2_{M_{K}}\right]
\end{align*}
Hence using the same technique that we employed for the linear case, the final expression for  $\Gamma(s_0, \boldsymbol{s})$ as
\begin{align*}
     \Gamma(s_0, \boldsymbol{s})=\text{exp}\left[\gamma_S s_0+ \sum_{k=1}^K\alpha_k\eta_{M_k\leadsto Y}s_k + \left(\boldsymbol\delta_{Y}^T+\sum_{k=1}^K \eta_{M_k\leadsto Y}\boldsymbol\delta_{M_k}^T\right)\boldsymbol W +\frac{1}{2}\sum_{k=1}^K\eta_{M_k\leadsto Y}^2\sigma^2_{M_{k}}\right]\cdot
\end{align*}
\end{proof}
\subsection{Proof of Theorem 2}\label{sec:prthm2}
\begin{proof}
    It follows immediately from Equation (14) of the main text that for each \( k = 1, 2, \ldots, K \),
\begin{equation}
    n^{1/2}(\wh \alpha_k - \alpha^*_k)\todistribution\calN(0,\sigma_{\alpha_k}^2) \quad \text{as} \quad n\rightarrow \infty \cdot\label{eq:alpha}
\end{equation}
where $\sigma_{\alpha_k}^2=\text{var}(Z_{\alpha_k})$. Since $\eta^*_{M_K\leadsto Y}=\beta^*_{K}$ and $\widehat{\eta}_{M_K \leadsto Y} = \widehat{\beta}_{K}$, we write, $n^{1/2}(\wh \eta_{M_K\leadsto Y} - \eta^*_{M_K\leadsto Y})  = n^{1/2}(\wh \beta_{K} - \beta^*_{K})\todistribution Z_{\eta_{M_K\leadsto Y}}\sim \calN(0, \sigma_{\beta_{K}}^2)$, according to from Equation (14) of the main text. Since $\eta^*_{M_{K-1} \leadsto Y}=\beta^*_{K-1}+\wh\xi_{K|K-1}\wh\eta_{M_K\leadsto Y}$ and $\eta^*_{M_{K-1} \leadsto Y}=\beta^*_{K-1}+\xi^*_{K|K-1}\eta^*_{M_K\leadsto Y}$, the use of first order delta method leads to
\beqrs
n^{1/2}(\wh \eta_{M_{K-1}\leadsto Y} - \eta^*_{M_{K - 1}\leadsto Y}) \todistribution Z_{\beta_{K - 1}} + \xi^*_{K\mid K - 1} Z_{\eta_{M_K\leadsto Y}} + \eta_{M_K\leadsto Y} Z_{\xi_{K\mid K - 1}}\defn Z_{\eta_{M_{K - 1}\leadsto Y}}\sim \mathcal{N}(0,\sigma_{\eta_{M_{K - 1}\leadsto Y}}^2),
\eeqrs
Repeatedly using recursion for any $k<K$, we can obtain
\beqrs
n^{1/2}(\wh \eta_{M_{k}\leadsto Y} - \eta^*_{M_{k}\leadsto Y}) \todistribution Z_{\beta_k} + \sum_{j = k + 1}^K(\xi^*_{j\mid k} Z_{\eta_{M_j\leadsto Y}} + \eta^*_{M_j\leadsto Y} Z_{\xi_{j\mid k}})\defn Z_{\eta_{M_{k}\leadsto Y}}\sim \mathcal{N}(0,\sigma_{\eta_{M_{k}\leadsto Y}}^2)\cdot
\eeqrs
Next we derive the asymptotic $\sigma_{\eta_{M_k\leadsto Y}}^2$, $k=\{1,2,\cdots,K\}$. First the likelihood orthogonality implies that (i) $Z_{\alpha_k}\independent (Z_{\beta_{M_j}},  Z_{\xi_{j\mid k}})$ for any $j\geq k$, (ii) $Z_{\xi_{j\mid k}}\independent (Z_{\beta_{k}}, \ldots, Z_{\beta_{K}})$ for any $j\geq k$, and (iii) $Z_{\xi_{j\mid k}}\independent Z_{\xi_{j^\prime\mid k}}$ for any $j,j^\prime\geq k$ and $j\neq j^\prime$. For each \( k \in \{1, 2, \ldots, K\} \) and any $j\geq k$, let $ \sigma^2_{\beta_{k}}$ and $\sigma^2_{\xi_{j\mid k}}$ be the non-null variances of $Z_{\beta_{k}}$ and $Z_{\xi_{j\mid k}}$ respectively.
First we obtain that $\sigma_{\eta_{M_K\leadsto Y}}^2 = \sigma_{\beta_{K}}^2$. Next, for $k = K - 1$, 
\beqrs
\sigma_{\eta_{M_{K - 1}\leadsto Y}}^2 = \sigma_{\beta_{K - 1}}^2 + (\xi^*_{K\mid K - 1})^2 \sigma_{\eta_{M_K\leadsto Y}}^2 + (\eta^*_{M_K\leadsto Y})^2\sigma_{\xi_{K\mid K - 1}}^2 + 2\xi^*_{K\mid K - 1}\cov(Z_{\beta_{K-1}}, Z_{\beta_{K}}).
\eeqrs
Likewise we have for any $k<K$,
\beqrs
\sigma_{\eta_{M_{k}\leadsto Y}}^2 = \sigma_{\beta_{k}}^2 + \sum_{j = k + 1}^K\{(\xi^*_{j\mid k})^2 \sigma_{\eta_{M_j\leadsto Y}}^2 + (\eta^*_{M_j\leadsto Y})^2\sigma_{\xi_{j\mid k}}^2 + 2\xi_{j\mid k}\cov(Z_{\beta_{k}}, Z_{\beta_{j}})\}.
\eeqrs
All the limiting variances $\sigma_{\eta_{M_{k}\leadsto Y}}^2$, $k=1,2,\cdots,K$ are positive because the covariance matrix of $(Z_{\alpha_k}, Z_{\xi_{j\mid k}}, Z_{\beta_k})_{k = 1, j\geq k}^K$ is positive definite. Hence we obtain
\begin{equation}
    n^{1/2}
\begin{pmatrix}
    \wh \alpha_k\\
    \wh \eta_{M_k\leadsto Y}
\end{pmatrix}
\todistribution
\calN
\left\{
\begin{pmatrix}
    \alpha^*_k\\
    \eta^*_{M_k\leadsto Y}
\end{pmatrix},
\begin{pmatrix}
    \sigma_{\alpha_k}^2  & 0\\
    0 & \sigma_{\eta_{M_k\leadsto Y}}^2
\end{pmatrix}
\right\}\cdot \label{eq:finasymp}
\end{equation}
\end{proof}

\subsection{Proof of Lemma 1}\label{sec:prlm1}

\begin{proof}[Proof of Lemma 1]
The result in Theorem 2 and the first-order delta method together yields  
\beqrs
n^{1/2} \wh \alpha_k \wh \eta_{M_k\leadsto Y} = n^{1/2} (\wh \alpha_k \wh \eta_{M_k\leadsto Y} -  \alpha_k \eta_{M_k\leadsto Y}) \todistribution \mathcal{N} \left(0, \alpha_k^2 \sigma^2_{\eta_{M_k\leadsto Y}} + \eta_{M_k\leadsto Y}^2 \sigma^2_{\alpha_k} \right).
\eeqrs  
Applying Slutsky’s theorem, it follows that $T_{S\to M_k\leadsto Y} \todistribution \mathcal{N}(0,1)$. \\

\noindent
However, when $(\alpha_k, \eta_{M_k\leadsto Y}) = (0,0)$, the first-order delta method is no longer applicable. Instead, noting that $T_{\alpha_k} \todistribution \mathcal{N}(0,1)$ and $T_{\eta_{M_k\leadsto Y}} \todistribution \mathcal{N}(0,1)$ are asymptotically independent, we apply the continuous mapping theorem to obtain  
\beqrs
T_{S\to M_k\leadsto Y} \todistribution \frac{Z_1 Z_2}{(Z_1^2 + Z_2^2)^{1/2}},
\eeqrs  
where $Z_1$ and $Z_2$ are independent standard normal variables. By the properties of independent normal distributions, the right-hand side follows a $\mathcal{N}(0,1/4)$ distribution.
\end{proof}

% \newpage
\section{Additional Simulations, Tables and Figures}
\subsection{Probit and Logistic SEM-SCOM(3) for Binary Outcomes}\label{sec:binsimu}
The data is simulated from model (18) of the main text, where the binary outcome $Y$ is generated from a probit model. Specifically, we assume $(\epsilon_{M_1},\epsilon_{M_2},\epsilon_{M_3}) \sim \text{MVN}(\mathbf{0}_3,\,0.25^2 \mathcal{I}_{3\times3})$. All parameters are kept the same as in the baseline setup, except that we set $\gamma_S = -2$. In addition, we consider an alternative scenario where the binary outcome $Y$ is generated from a logistic model.\\

\noindent
\textit{Estimation.} Tables \ref{table:estimation:logistic} and \ref{table:estimation:probit} and Figures \ref{fig:estimation_accuracy_logit} and \ref{fig:estimation_accuracy_probit} summarize the estimation results for logistic and probit outcome models respectively. Under both scenarios we observe that bias, MSE, and variance decrease with increasing sample size. The estimators attain a linear convergence rate under \(H_{01}\), \(H_{02}\), and \(H_{1}\) across all three SEMs. However, under \(H_{03}\), the convergence rate deteriorates due to singularity.\\

\noindent
\textit{Type I Error.} Figures~\ref{fig:qqplot_three_mediators_logit} and ~\ref{fig:qqplot_three_mediators_probit} show that  
quantile-quantile plots of $p$-values in $-\log_{10}$ scale under the three null hypotheses for logistic and probit outcome models respectively. These figures demonstrate that SOMET is the only method that can effectively control type I error under all three nulls, $H_{01},H_{02}$ and $H_{03}$.\\

\noindent
\textit{Statistical Power.}  
Figures~\ref{fig:power_three_mediators_logit} and ~\ref{fig:power_three_mediators_probit} display the empirical rejection rates over signal strength $\alpha_k = \beta_k$ and over ratio $\alpha_k / \beta_k$. These empirical power findings confirm that SOMET is superior over the other four methods in terms of statistical power.
\begin{table}[H]
\centering
\caption{The mean squared errors (MSE), bias, and variance of the mediation effect estimates under the linear structural equation model. The MSE, bias, and variance are computed as the relative to those under $n = 200$. The results are averaged over 2000 simulations.}
\label{table:estimation}
\begin{tabular}{cc ccc ccc ccc}
  \hline
    \hline
 % & null\_case & n & Relative\_MSE\_M.1 & Relative\_MSE\_M.2 & Relative\_MSE\_M.3 & Relative\_Bias\_M.1 & Relative\_Bias\_M.2 & Relative\_Bias\_M.3 & Relative\_Var\_M.1 & Relative\_Var\_M.2 & Relative\_Var\_M.3 \\ 
 & $n$ & \multicolumn{3}{c}{MSE} & \multicolumn{3}{c}{Bias} & \multicolumn{3}{c}{Variance} \\
\cline{3-5} \cline{6-8} \cline{9-11}
 & & $M_1$ & $M_2$ & $M_3$ & $M_1$ & $M_2$ & $M_3$ & $M_1$ & $M_2$ & $M_3$ \\
  \hline
 \multirow{4}{*}{$H_{01}$} & 400 & 0.47 & 0.43 & 0.48 & -0.69 & -1.00 & 0.25 & 0.47 & 0.43 & 0.48 \\ 
   & 600 & 0.30 & 0.29 & 0.32 & 0.10 & -0.43 & 1.48 & 0.30 & 0.29 & 0.32 \\ 
   & 800 & 0.23 & 0.21 & 0.24 & -0.16 & 0.12 & -0.78 & 0.23 & 0.21 & 0.24 \\ 
   & 1000 & 0.18 & 0.18 & 0.19 & 0.09 & -0.77 & 0.95 & 0.18 & 0.18 & 0.19 \\ 
   \hline
 \multirow{4}{*}{$H_{02}$} & 400 & 0.51 & 0.47 & 0.52 & 0.30 & 3.63 & 0.45 & 0.51 & 0.47 & 0.52 \\ 
   & 600 & 0.33 & 0.32 & 0.33 & 1.55 & 0.42 & 0.13 & 0.33 & 0.32 & 0.34 \\ 
   & 800 & 0.25 & 0.24 & 0.24 & 0.94 & 1.38 & 0.59 & 0.25 & 0.24 & 0.24 \\ 
   & 1000 & 0.20 & 0.19 & 0.20 & 0.27 & 0.40 & 0.07 & 0.20 & 0.19 & 0.20 \\
      \hline 
  \multirow{4}{*}{$H_{03}$}  & 400 & 0.24 & 0.25 & 0.27 & 0.12 & -1.01 & 1.06 & 0.24 & 0.25 & 0.27 \\ 
   & 600 & 0.11 & 0.10 & 0.12 & 0.06 & 1.23 & 0.18 & 0.11 & 0.10 & 0.12 \\ 
   & 800 & 0.06 & 0.06 & 0.07 & 0.02 & -0.53 & 0.12 & 0.06 & 0.06 & 0.07 \\ 
   & 1000 & 0.04 & 0.04 & 0.05 & 0.06 & -0.25 & 0.06 & 0.04 & 0.04 & 0.05 \\ 
      \hline
   \multirow{4}{*}{$H_{1}$} & 400 & 0.51 & 0.47 & 0.48 & 0.75 & 2.50 & 0.32 & 0.51 & 0.47 & 0.48 \\ 
   & 600 & 0.32 & 0.33 & 0.33 & 2.09 & 0.48 & 0.23 & 0.32 & 0.33 & 0.33 \\ 
   & 800 & 0.25 & 0.24 & 0.24 & 1.36 & 0.90 & 0.52 & 0.25 & 0.24 & 0.24 \\ 
   & 1000 & 0.20 & 0.18 & 0.20 & 0.26 & 0.34 & 0.38 & 0.20 & 0.18 & 0.20 \\ 
   \hline
         \hline
\end{tabular}
\end{table}

\begin{table}[H]
\centering
\caption{The mean squared errors (MSE), bias, and variance of the mediation effect estimates under the logistic structural equation model. The MSE, bias, and variance are computed as the relative to those under $n = 1000$. The results are averaged over 2000 simulations.}
\label{table:estimation:logistic}
\begin{tabular}{cc ccc ccc ccc}
  \hline
    \hline
 % & null\_case & n & Relative\_MSE\_M.1 & Relative\_MSE\_M.2 & Relative\_MSE\_M.3 & Relative\_Bias\_M.1 & Relative\_Bias\_M.2 & Relative\_Bias\_M.3 & Relative\_Var\_M.1 & Relative\_Var\_M.2 & Relative\_Var\_M.3 \\ 
 & $n$ & \multicolumn{3}{c}{MSE} & \multicolumn{3}{c}{Bias} & \multicolumn{3}{c}{Variance} \\
\cline{3-5} \cline{6-8} \cline{9-11}
 & & $M_1$ & $M_2$ & $M_3$ & $M_1$ & $M_2$ & $M_3$ & $M_1$ & $M_2$ & $M_3$ \\
  \hline
   \multirow{4}{*}{$H_{01}$} & 1500 & 0.58 & 0.67 & 0.63 & -4.21 & -0.58 & -0.65 & 0.58 & 0.67 & 0.63 \\ 
   & 2000 & 0.46 & 0.48 & 0.46 & -2.27 & -0.03 & -2.00 & 0.46 & 0.48 & 0.46 \\ 
   & 2500 & 0.38 & 0.40 & 0.37 & -1.53 & 0.16 & 0.35 & 0.38 & 0.40 & 0.37 \\ 
   & 3000 & 0.31 & 0.32 & 0.32 & 1.09 & -0.09 & -0.66 & 0.31 & 0.32 & 0.32 \\ 
         \hline
   \multirow{4}{*}{$H_{02}$} & 1500 & 0.64 & 0.66 & 0.59 & -1.37 & -0.10 & 0.33 & 0.64 & 0.66 & 0.59 \\ 
   & 2000 & 0.48 & 0.48 & 0.45 & 1.62 & -1.25 & 2.87 & 0.48 & 0.48 & 0.45 \\ 
   & 2500 & 0.37 & 0.39 & 0.37 & 1.25 & -0.21 & -1.74 & 0.37 & 0.39 & 0.37 \\ 
   & 3000 & 0.33 & 0.31 & 0.28 & -1.60 & 0.73 & 0.48 & 0.33 & 0.31 & 0.28 \\ 
         \hline
   \multirow{4}{*}{$H_{03}$} & 1500 & 0.40 & 0.43 & 0.39 & 0.39 & -54.36 & 0.04 & 0.40 & 0.43 & 0.39 \\ 
   & 2000 & 0.26 & 0.26 & 0.22 & 0.36 & 26.81 & -0.06 & 0.26 & 0.26 & 0.22 \\ 
   & 2500 & 0.15 & 0.18 & 0.15 & 0.01 & -1.34 & -0.17 & 0.15 & 0.18 & 0.15 \\ 
   & 3000 & 0.11 & 0.11 & 0.09 & -0.12 & -6.86 & 0.20 & 0.11 & 0.11 & 0.09 \\ 
         \hline
   \multirow{4}{*}{$H_{1}$} & 1500 & 0.57 & 0.61 & 0.61 & 0.62 & 0.59 & 0.59 & 0.58 & 0.63 & 0.62 \\ 
   & 2000 & 0.43 & 0.47 & 0.44 & 0.44 & 0.58 & 0.41 & 0.44 & 0.48 & 0.46 \\ 
   & 2500 & 0.34 & 0.38 & 0.38 & 0.34 & 0.40 & 0.42 & 0.35 & 0.39 & 0.39 \\ 
   & 3000 & 0.29 & 0.30 & 0.29 & 0.37 & 0.35 & 0.34 & 0.29 & 0.31 & 0.30 \\ 
   \hline
         \hline
\end{tabular}
\end{table}

\begin{table}[H]
\centering
\caption{The mean squared errors (MSE), bias, and variance of the mediation effect estimates under the probit structural equation model. The MSE, bias, and variance are computed as the relative to those under $n = 1000$. The results are averaged over 2000 simulations.}
\label{table:estimation:probit}
\begin{tabular}{cc ccc ccc ccc}
  \hline
    \hline
 % & null\_case & n & Relative\_MSE\_M.1 & Relative\_MSE\_M.2 & Relative\_MSE\_M.3 & Relative\_Bias\_M.1 & Relative\_Bias\_M.2 & Relative\_Bias\_M.3 & Relative\_Var\_M.1 & Relative\_Var\_M.2 & Relative\_Var\_M.3 \\ 
 & $n$ & \multicolumn{3}{c}{MSE} & \multicolumn{3}{c}{Bias} & \multicolumn{3}{c}{Variance} \\
\cline{3-5} \cline{6-8} \cline{9-11}
 & & $M_1$ & $M_2$ & $M_3$ & $M_1$ & $M_2$ & $M_3$ & $M_1$ & $M_2$ & $M_3$ \\
  \hline
   \multirow{4}{*}{$H_{01}$} & 1500 & 0.65 & 0.62 & 0.66 & -0.13 & 0.28 & 0.25 & 0.65 & 0.63 & 0.66 \\ 
   & 2000 & 0.48 & 0.47 & 0.50 & 1.80 & -0.23 & 0.41 & 0.48 & 0.48 & 0.50 \\ 
   & 2500 & 0.38 & 0.37 & 0.40 & -3.95 & -0.03 & -0.60 & 0.38 & 0.37 & 0.40 \\ 
   & 3000 & 0.31 & 0.30 & 0.31 & -1.46 & -0.12 & -0.04 & 0.31 & 0.30 & 0.31 \\ 
         \hline
   \multirow{4}{*}{$H_{02}$} & 1500 & 0.63 & 0.66 & 0.58 & -0.95 & -0.06 & -0.45 & 0.63 & 0.66 & 0.58 \\ 
   & 2000 & 0.48 & 0.48 & 0.44 & 0.95 & -1.13 & 2.69 & 0.48 & 0.48 & 0.44 \\ 
   & 2500 & 0.36 & 0.39 & 0.36 & 0.86 & -0.24 & -1.99 & 0.36 & 0.39 & 0.36 \\ 
   & 3000 & 0.33 & 0.31 & 0.28 & -0.91 & 0.65 & 0.66 & 0.33 & 0.31 & 0.28 \\ 
         \hline
   \multirow{4}{*}{$H_{03}$} & 1500 & 0.43 & 0.43 & 0.42 & 1.96 & -1.52 & 0.09 & 0.43 & 0.43 & 0.42 \\ 
   & 2000 & 0.24 & 0.23 & 0.23 & -0.70 & 2.01 & 0.19 & 0.24 & 0.23 & 0.23 \\ 
   & 2500 & 0.14 & 0.15 & 0.16 & 1.22 & -0.41 & -0.44 & 0.14 & 0.15 & 0.16 \\ 
   & 3000 & 0.10 & 0.10 & 0.09 & -0.78 & 0.33 & -0.15 & 0.10 & 0.10 & 0.09 \\ 
         \hline
   \multirow{4}{*}{$H_{1}$} & 1500 & 0.57 & 0.61 & 0.62 & 0.62 & 0.68 & 0.62 & 0.59 & 0.62 & 0.62 \\ 
   & 2000 & 0.41 & 0.44 & 0.46 & 0.51 & 0.64 & 0.65 & 0.43 & 0.45 & 0.46 \\ 
   & 2500 & 0.29 & 0.34 & 0.33 & 0.34 & 0.43 & 0.70 & 0.31 & 0.35 & 0.32 \\ 
   & 3000 & 0.24 & 0.28 & 0.29 & 0.29 & 0.35 & 0.40 & 0.26 & 0.28 & 0.29 \\ 
   \hline
         \hline
\end{tabular}
\end{table}

\begin{figure}[H]
	\centerline{%\renewcommand{\arraystretch}{0.8} %<- modify value to suit your needs
		\begin{tabular}{c}			
        \includegraphics[width=4.5in]{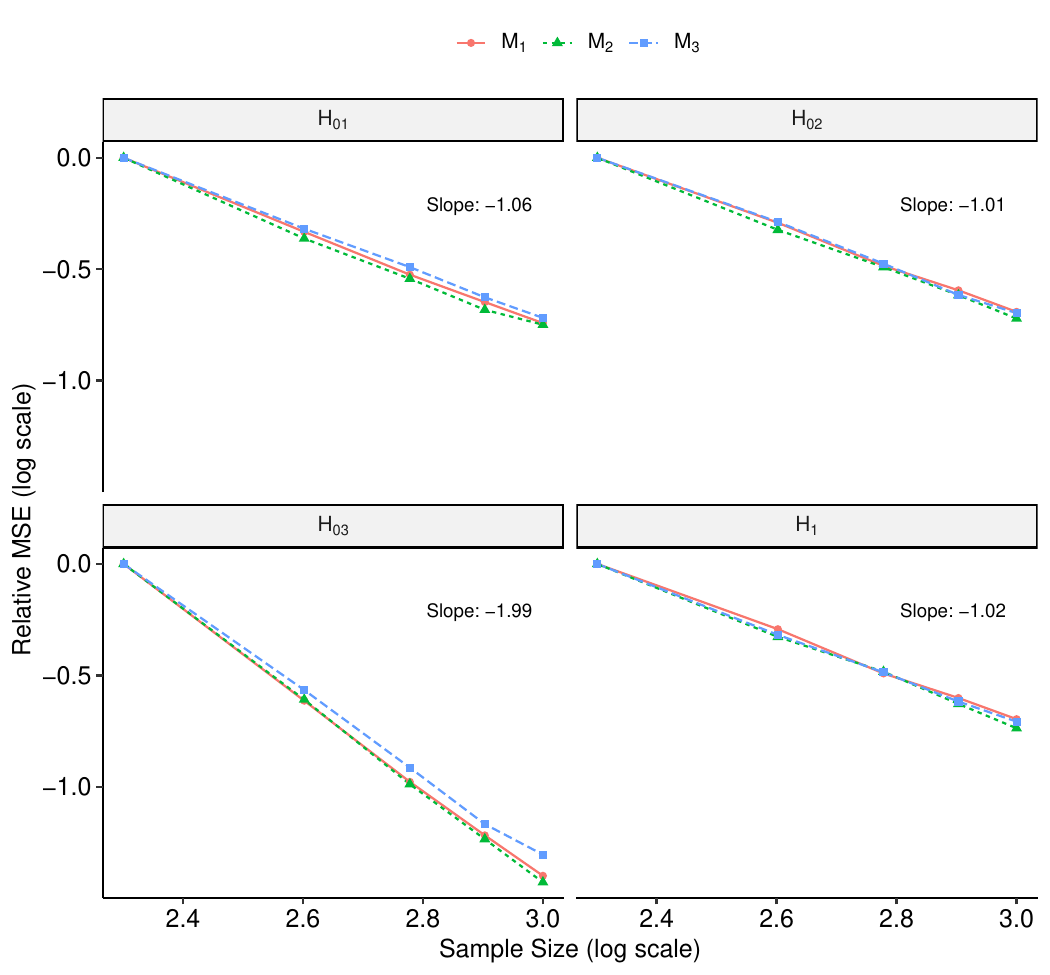}
		\end{tabular}
	}
	\captionsetup{font = footnotesize}
\caption{Mean squared errors (log scale) versus the sample size (log scale) under the linear structural equation model. }
	\label{fig:estimation_accuracy}
\end{figure}

\begin{figure}[H]
	\centerline{%\renewcommand{\arraystretch}{0.8} %<- modify value to suit your needs
		\begin{tabular}{c}			
        \includegraphics[width=4.5in]{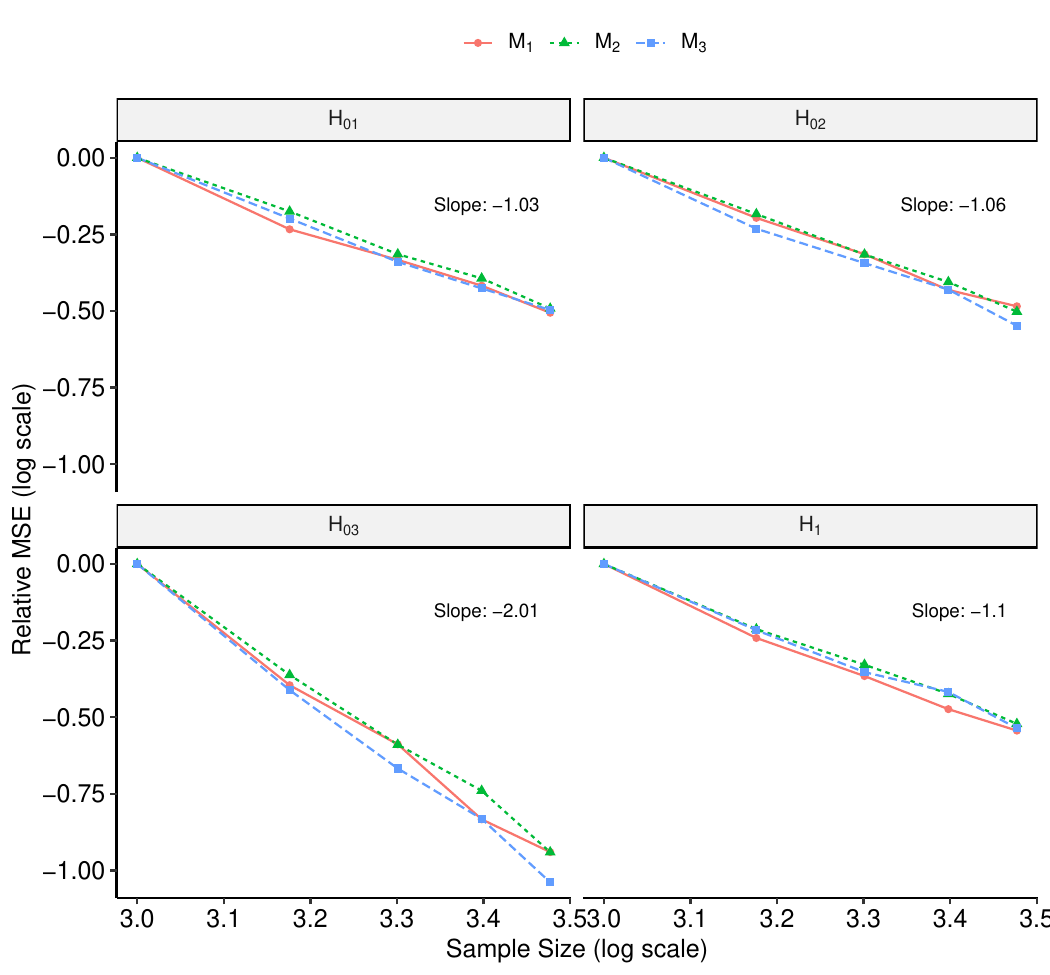}
		\end{tabular}
	}
	\captionsetup{font = footnotesize}
\caption{Mean squared errors (log scale) versus the sample size (log scale) under the logistic structural equation model.}
	\label{fig:estimation_accuracy_logit}
\end{figure}

\begin{figure}[H]
	\centerline{%\renewcommand{\arraystretch}{0.8} %<- modify value to suit your needs
		\begin{tabular}{c}			
        \includegraphics[width=4.5in]{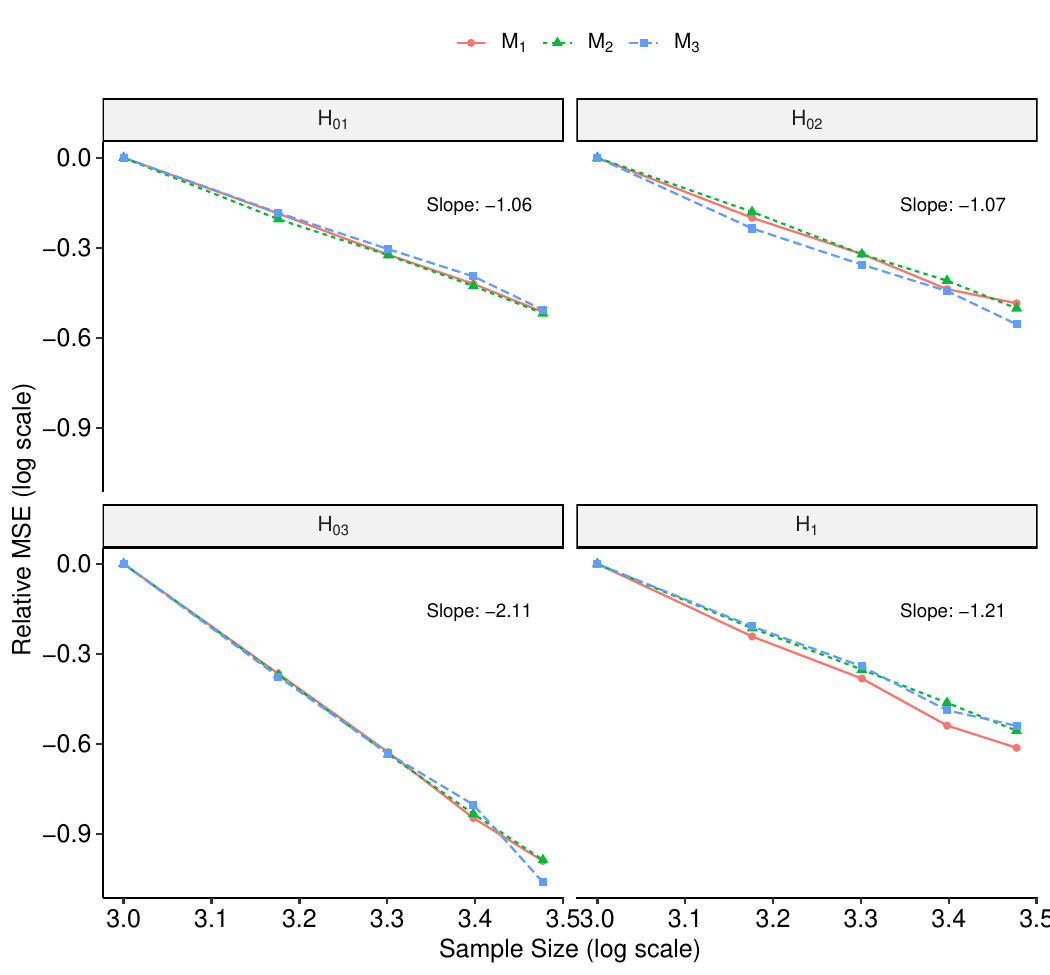}
		\end{tabular}
	}
	\captionsetup{font = footnotesize}
\caption{Mean squared errors (log scale) versus the sample size (log scale) under the probit structural equation model. }
	\label{fig:estimation_accuracy_probit}
\end{figure}
\newpage
\subsection{Additional Figures}\label{appendix:additional_figures}

\begin{figure}[H]
	\centerline{\renewcommand{\arraystretch}{0.8} %<- modify value to suit your needs
		\begin{tabular}{c}			
        \includegraphics[width=5.5in]{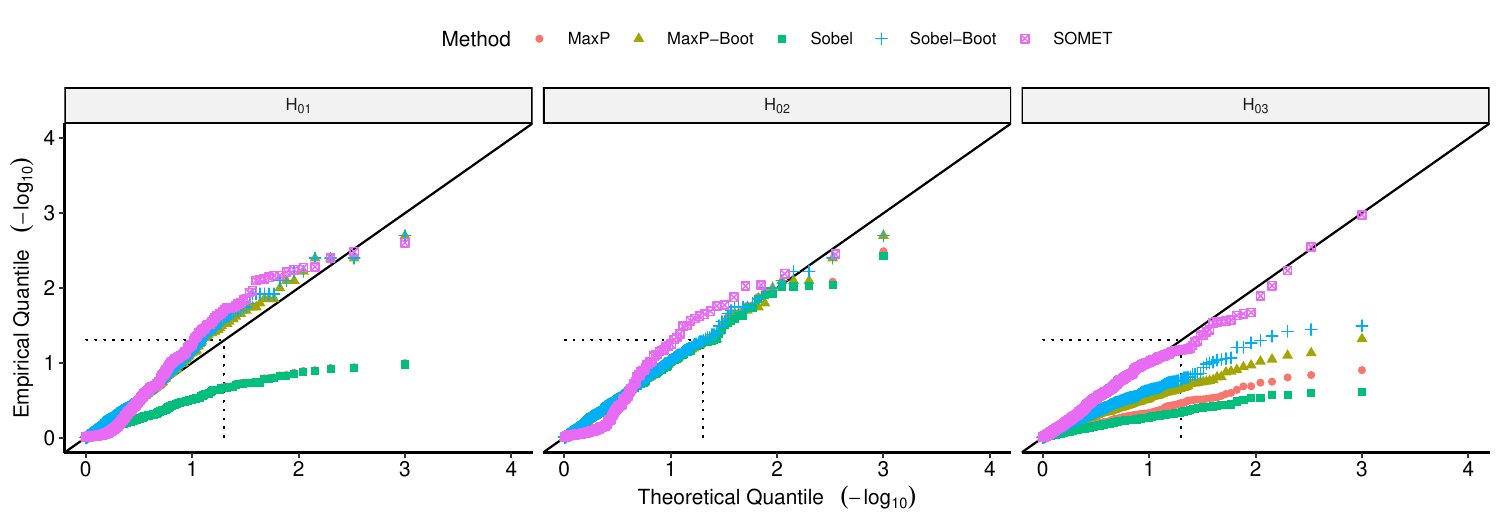}\\
        (a) $S\to M_1\leadsto Y$\\
        \includegraphics[width=5.5in]{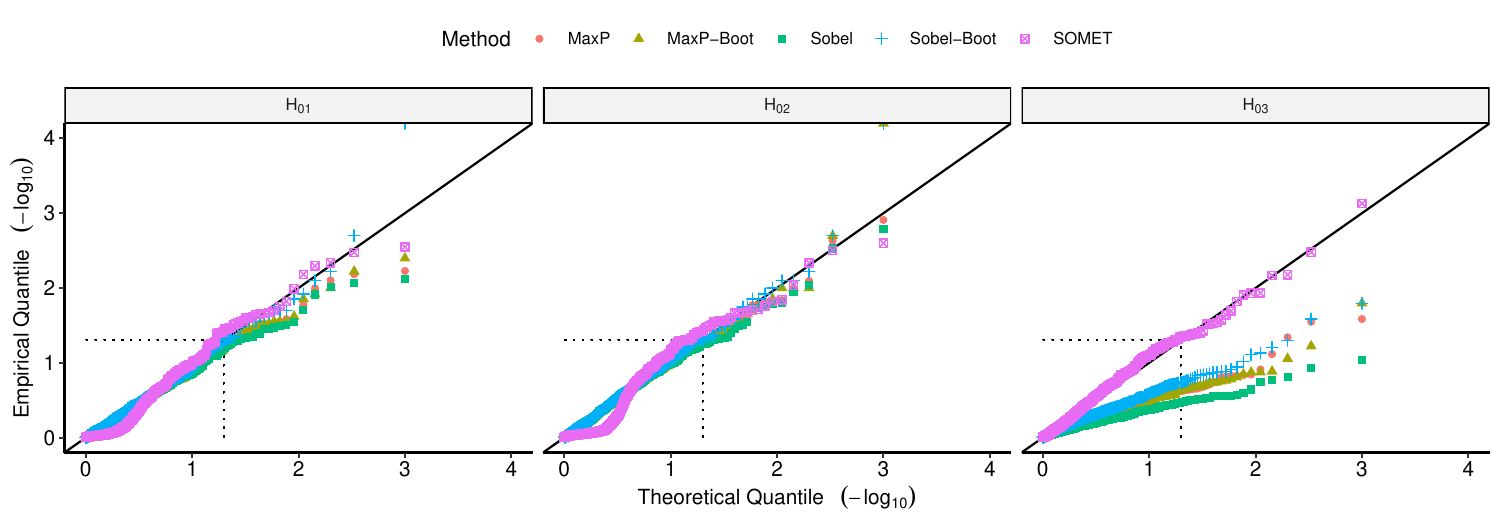}\\
       (b) $S\to M_2\leadsto Y$\\
        \includegraphics[width=5.5in]{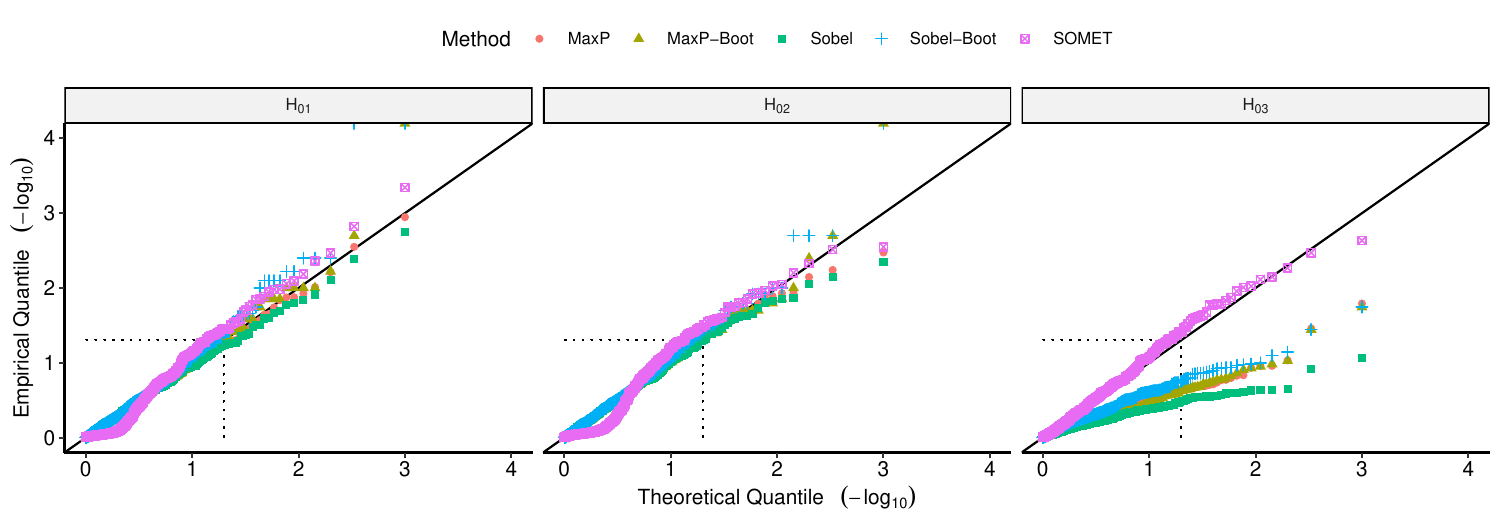}\\
(c)        $S \to M_3\leadsto Y$
		\end{tabular}
	}
	\captionsetup{font = footnotesize}
\caption{Quantile-quantile plots of $p$-values obtained by five tests over 500 replicates under the three nulls, where data are generated by Logistic SCOM(3) model and sample size $n = 2000$. The dotted line marks the $0.05$ nominal level.}
	\label{fig:qqplot_three_mediators_logit}
\end{figure}

\begin{figure}[H]
	\centerline{%\renewcommand{\arraystretch}{0.8} %<- modify value to suit your needs
		\begin{tabular}{c}			
        \includegraphics[width=5.5in]{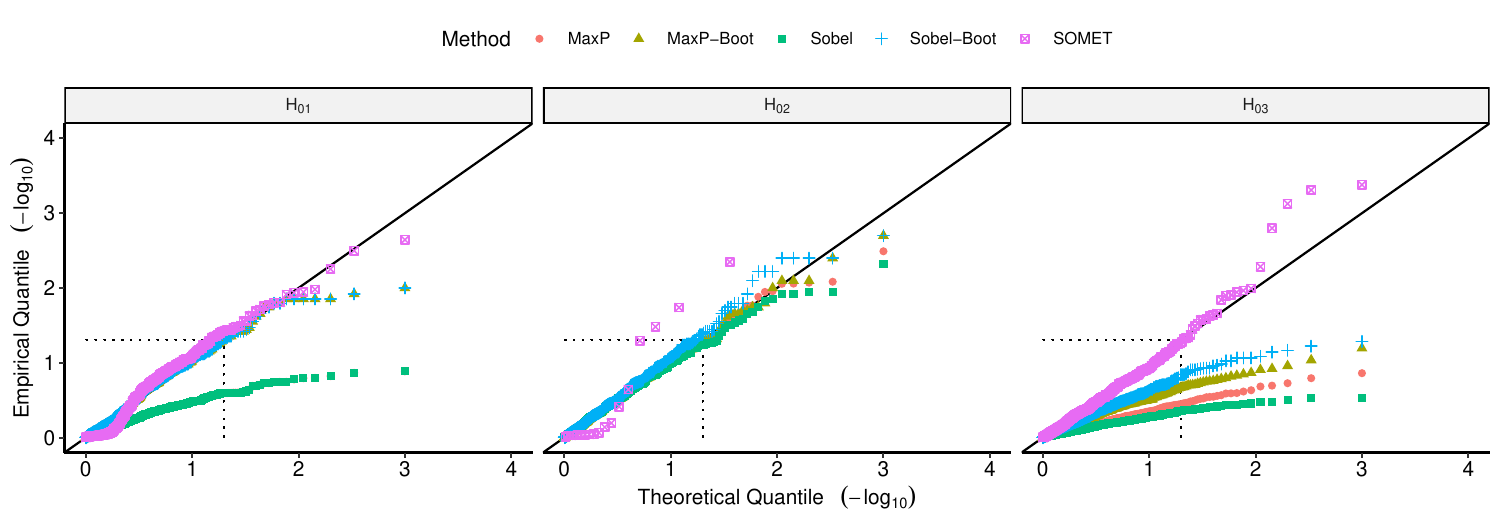}\\
(a)        $S\to M_1\leadsto Y$\\
        \includegraphics[width=5.5in]{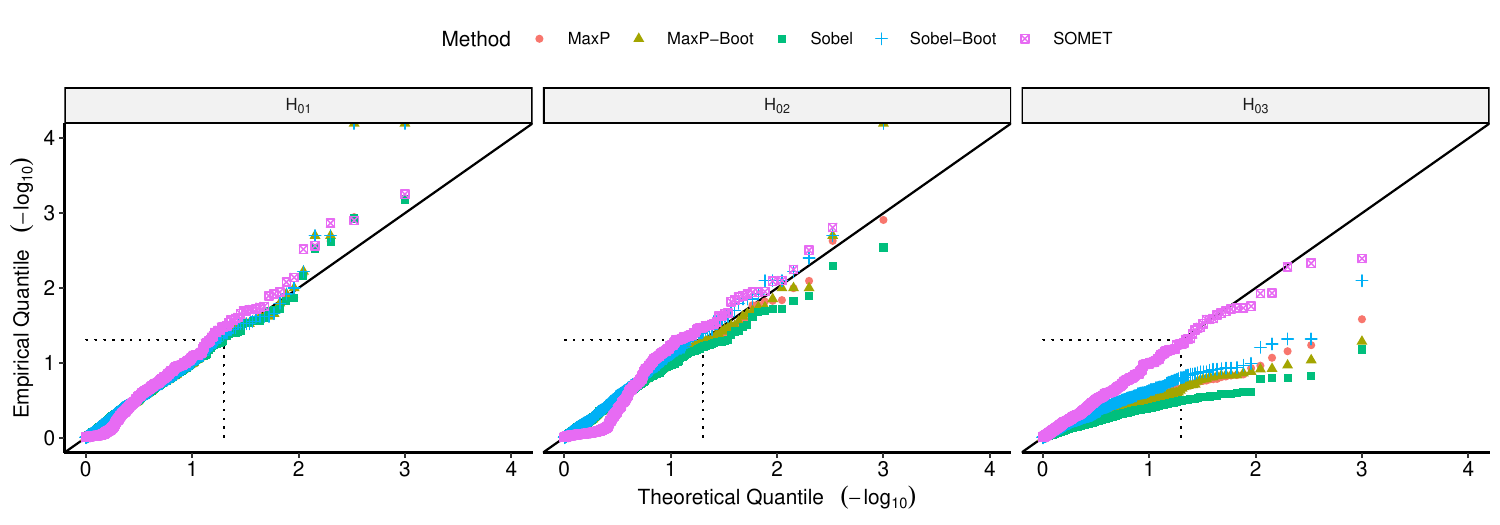}\\
   (b)     $S\to M_2\leadsto Y$\\
        \includegraphics[width=5.5in]{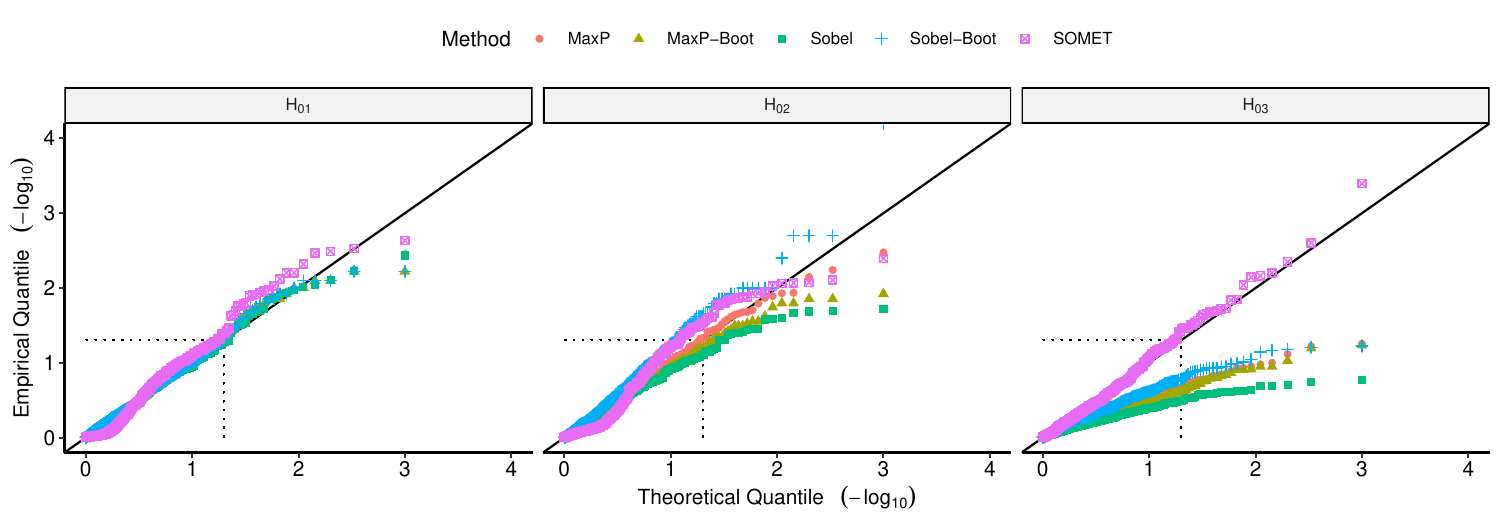}\\
     (c)   $S \to M_3\leadsto Y$
		\end{tabular}
	}
	\captionsetup{font = footnotesize}
\caption{Quantile-quantile plots of $p$-values obtained by five tests over 500 replicates under the three nulls, where data are generated by Probit SCOM(3) model and sample size $n = 2000$. The dotted line marks the $0.05$ nominal level.}
	\label{fig:qqplot_three_mediators_probit}
\end{figure}

\begin{figure}[H]
	\centerline{%\renewcommand{\arraystretch}{0.8} %<- modify value to suit your needs
		\begin{tabular}{c}			
        % \psfig{figure=sim_multiple_mediators_power_logit_qqplot_log_all_1_2000_a,width=5.5in,angle=0}\\
        % (a)  Signal strength  $\alpha_k = \beta_k$ \\
        % \psfig{figure=sim_multiple_mediators_power_logit_qqplot_log_all_1_2000_b,width=5.5in,angle=0}\\
        % (b) Ratio $\alpha_k/\beta_k$
        \includegraphics[width=5.5in]{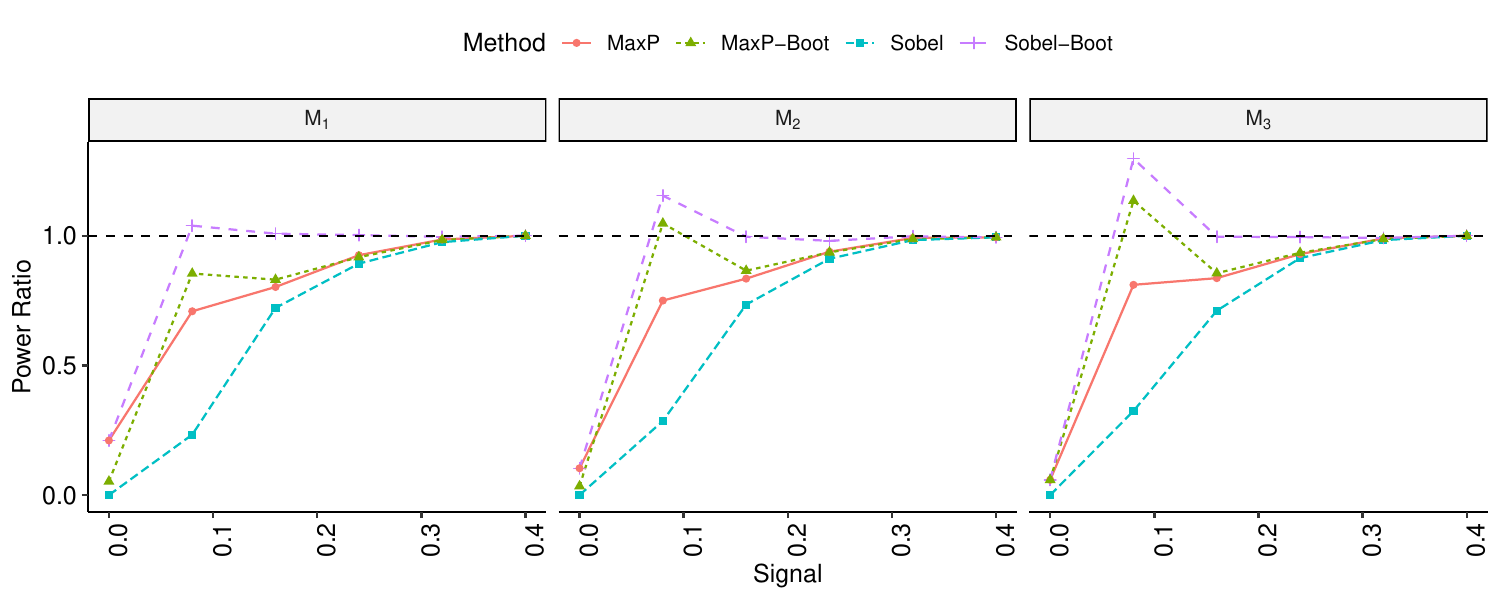}\\
        (a)  Signal strength  $\alpha_k = \beta_k$ \\
        \includegraphics[width=5.5in]{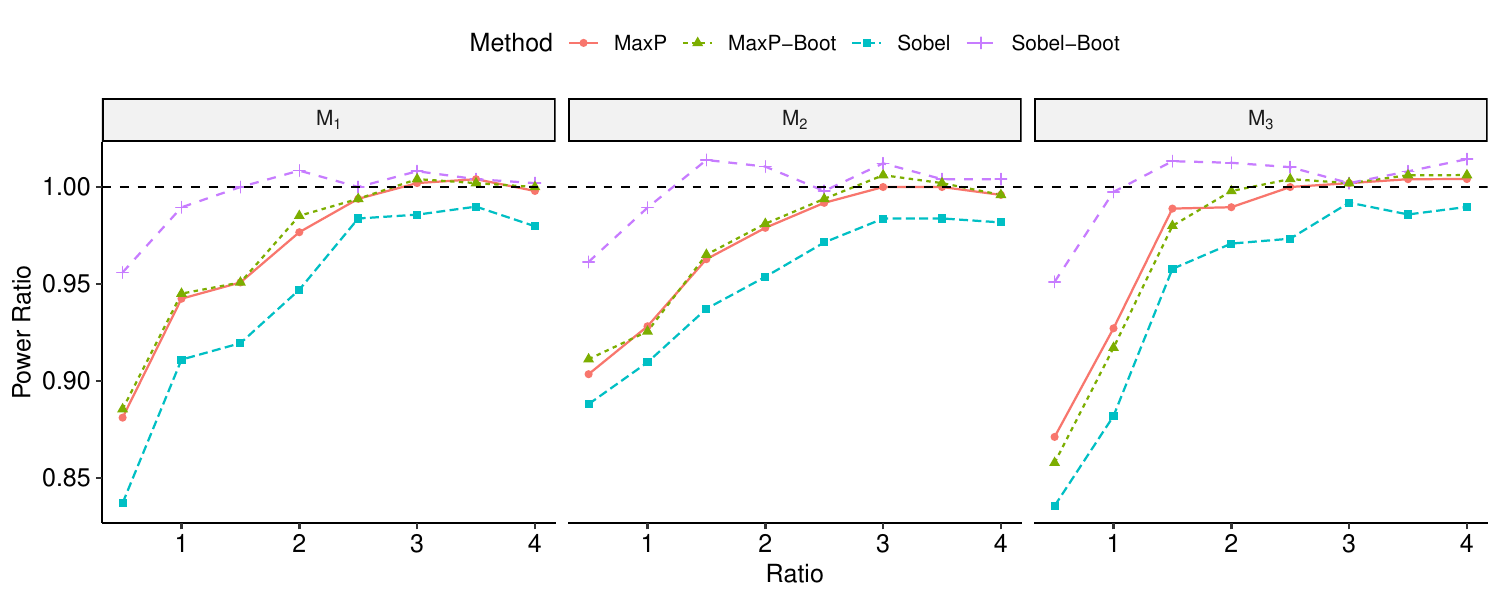}\\
        (b) Ratio $\alpha_k/\beta_k$
		\end{tabular}
	}
	\captionsetup{font = footnotesize}
	\caption{Empirical power ratio of the competing test versus the proposed test  over the signal strength  $\alpha_k = \beta_k$ in panel (a) and over the ratio $\alpha_k/\beta_k$ in panel (b). The composite mediation pathway $S\rightarrow M_k\leadsto Y$, $k=1,2,3$ is tested from the left ($k=1$) to right ($k=3$). The model is logistic SEM.}
	\label{fig:power_three_mediators_logit}
\end{figure}

\begin{figure}[H]
	\centerline{%\renewcommand{\arraystretch}{0.8} %<- modify value to suit your needs
		\begin{tabular}{c}			
        \includegraphics[width=5.5in]{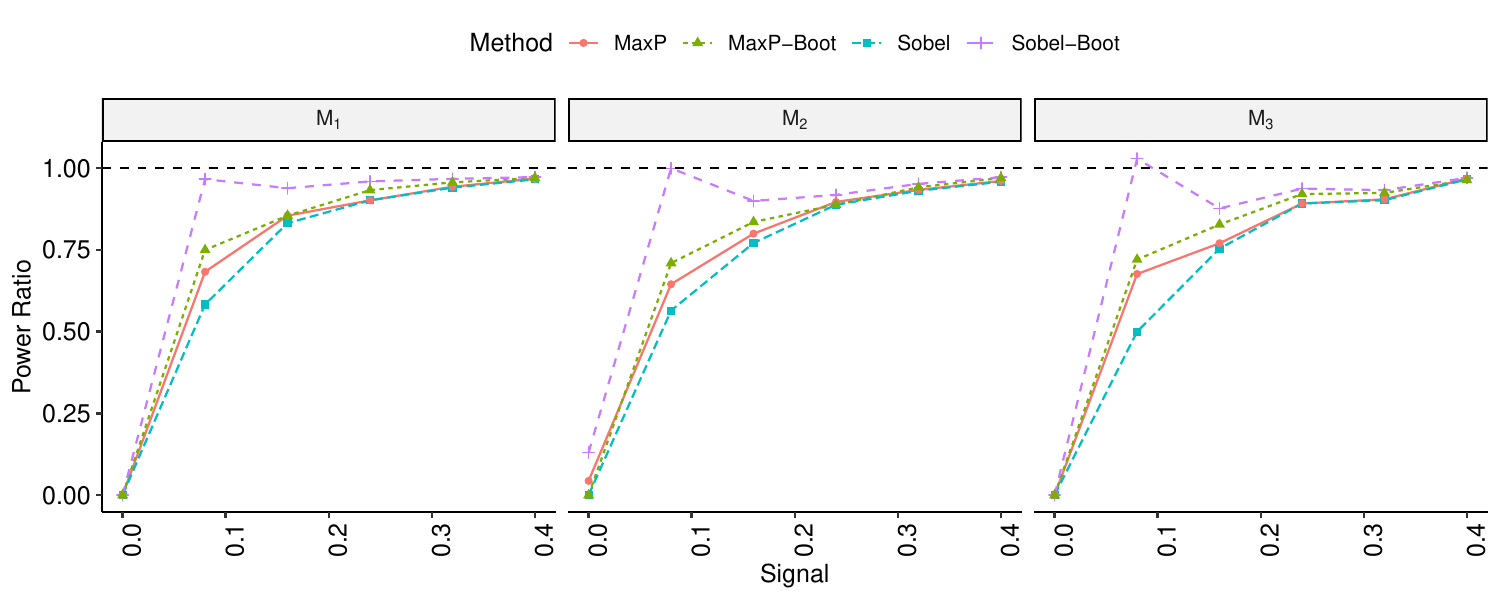}\\
        (a)  Signal strength  $\alpha_k = \beta_k$ \\
        \includegraphics[width=5.5in]{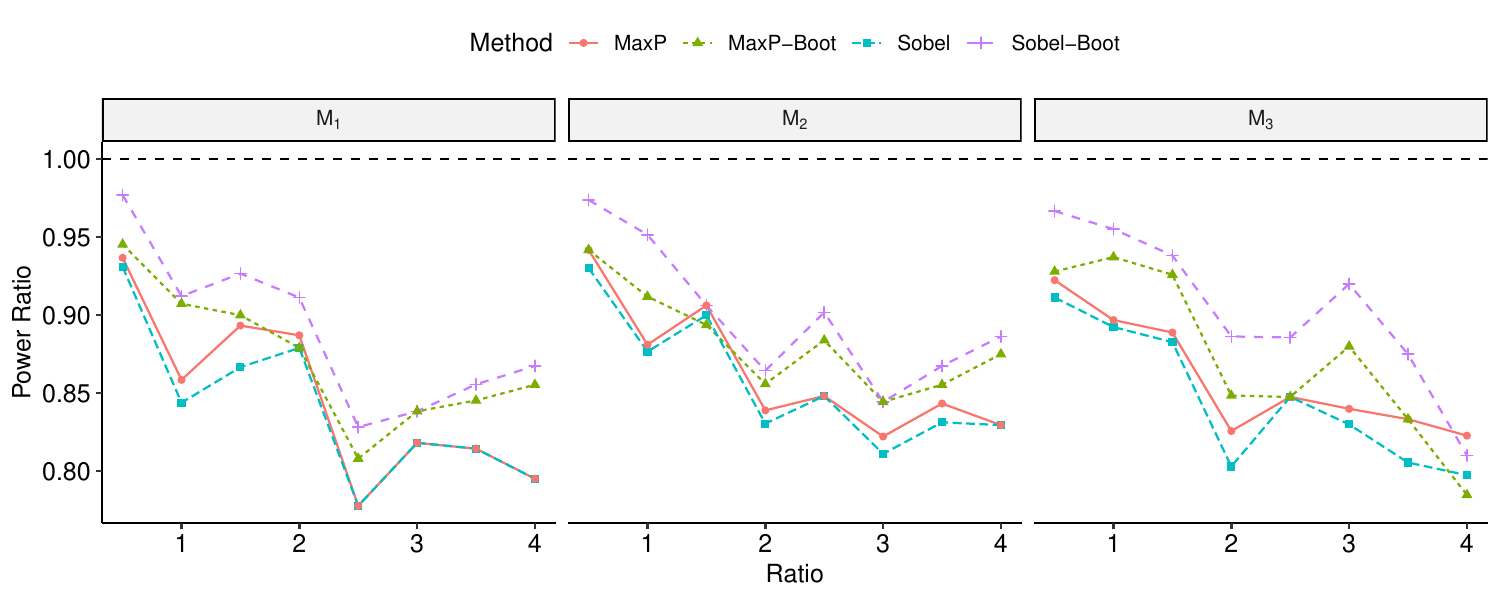}\\
        (b) Ratio $\alpha_k/\beta_k$
		\end{tabular}
	}
	\captionsetup{font = footnotesize}
	\caption{Empirical power ratio of the competing test versus the proposed test  over the signal strength  $\alpha_k = \beta_k$ in panel (a) and over the ratio $\alpha_k/\beta_k$ in panel (b). The composite mediation pathway $S\rightarrow M_k\leadsto Y$, $k=1,2,3$ is tested from the left ($k=1$) to right ($k=3$). The model is probit SEM.}
	\label{fig:power_three_mediators_probit}
\end{figure}
\begin{figure}[H]
    \centering
    \includegraphics[width=0.9\linewidth]{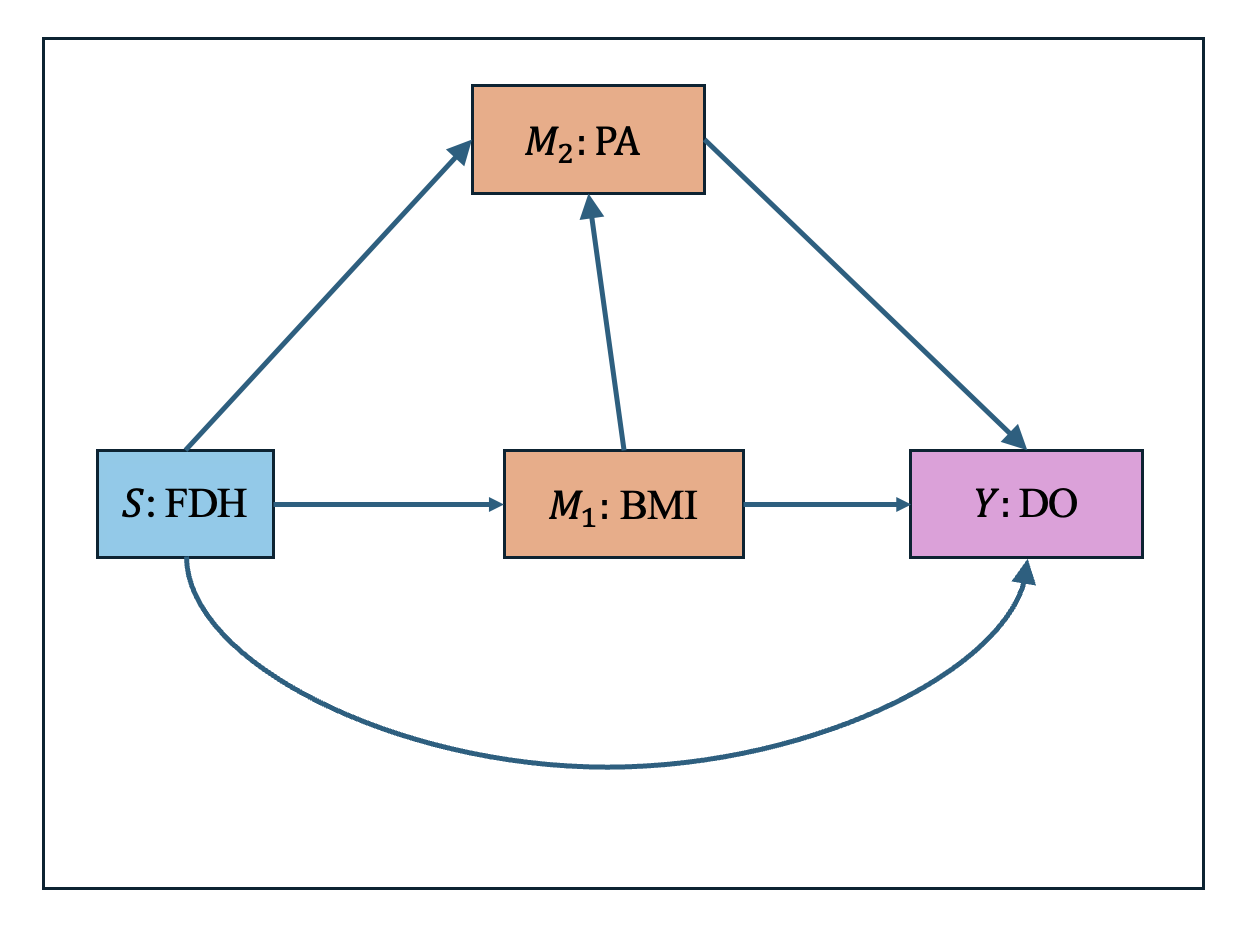}
    \caption{A directed acyclic graph with two sequentially ordered mediators, $M_1$ (BMI: Body Mass Index) and $M_2$ (PA: Physical Activity) through which prenatal exposure $S$ (FDH: Family Diabetes History) influences outcome $Y$ (DO: Diabetes Onset).}
    \label{fig:aou_dag}
\end{figure}

\clearpage
\section{Additional Tables}

\begin{table}[H]
\centering
\caption{Baseline Characteristics of the NIH All of US Fitbit Cohort} 
\begin{tabular}{lll}
  \toprule
{Variable} & {Level} & {Summary} \\ 
  \midrule
Sample Size (N) &  & 20233 \\ 
\hline
 Diabetes (\%) & No & 19898 (98.3) \\ 
   & Yes & 335 (1.7) \\
  \hline
  Family History of Diabetes (\%) & No & 13609 (67.3) \\ 
   & Yes & 6624 (32.7) \\ 
   \hline
  BMI (mean (SD)) &  & 30.59 (27.33) \\ 
  \hline
  Steps (mean (SD)) &  & 6095.17 (2726.47) \\ 
  \hline
  Age (mean (SD)) &  & 55.03 (16.22) \\ 
  \hline
  Sex (\%) & Females & 13995 (69.2) \\ 
   & Males & 6238 (30.8) \\ 
   \hline
  Smoking (\%) & No & 11865 (58.6) \\ 
   & Yes & 8368 (41.4) \\ 
   \hline
  Alcohol (\%) & No & 849 (4.2) \\ 
   & Yes & 19384 (95.8) \\ 
   \hline
  Married (\%) & No & 8970 (44.3) \\ 
   & Yes & 11263 (55.7) \\ 
   \hline
  Race (\%) & Non-NHW & 5970 (29.5) \\ 
   & NHW & 14263 (70.5) \\ 
   \bottomrule
\end{tabular}
\label{tab:baseline}
\end{table}

\begin{table}[ht]
    \centering
    \caption{$P$-values and estimated effects of each pathway for all lipids in the mediation analysis of $\text{maternal blood lead}\to\text{birth weight}\to \text{lipid}\to \text{childhood BMI}$.}
    \label{table:real_data_detailed}
    \resizebox{\textwidth}{!}{
    \begin{tabular}{l cccccccccccccc}
        \toprule
        & \multicolumn{5}{c}{$S\to M_1\leadsto Y$} & \multicolumn{5}{c}{$S\to M_2\to Y$} & \multicolumn{4}{c}{Effects}\\
        Lipid & SOMET & JS & JS-B & PoC & PoC-B & SOMET & JS & JS-B & PoC & PoC-B & NIE ($M_1$) & NIE ($M_2$) & NDE & NTE \\
        \midrule
        FA 18:0 OH-3   & 0.8516 & 0.37 & 0.27 & 0.4456 & 0.285 & 0.0007 & 0.3884 & 0.342 & 0.4139 & 0.221 & 0.0061 & -0.0021 & 0.2614 & 0.2654 \\
        FA 9:0 OH   & 0.7936 & 0.3616 & 0.264 & 0.4474 & 0.287 & 0.0107 & 0.9659 & 0.946 & 0.9661 & 0.946 & 0.006 & 0.0038 & 0.2378 & 0.2476 \\
        FA 14:0 OH   & 0.8543 & 0.3661 & 0.266 & 0.4417 & 0.281 & 0.0216 & 0.0033 & 0.415 & 0.03 & 0.001 & 0.0061 & -0.0071 & 0.2628 & 0.2618 \\
        FA 17:0 DiC   & 0.9281 & 0.3647 & 0.285 & 0.4731 & 0.331 & 0.0392 & 0.8544 & 0.781 & 0.8597 & 0.81 & 0.0058 & -0.0001 & 0.2443 & 0.25 \\
        FA 22:3   & 0.8709 & 0.3704 & 0.266 & 0.4441 & 0.282 & 0.0618 & 0.383 & 0.241 & 0.442 & 0.224 & 0.0058 & 0.0018 & 0.2486 & 0.2563 \\
        FA 13:0 DiC   & 0.8176 & 0.3657 & 0.264 & 0.4319 & 0.281 & 0.0713 & 0.1597 & 0.362 & 0.2575 & 0.082 & 0.0065 & -0.0087 & 0.2505 & 0.2484 \\
        FA 20:2   & 0.8544 & 0.3687 & 0.266 & 0.447 & 0.283 & 0.0871 & 0.7424 & 0.73 & 0.7443 & 0.721 & 0.0063 & 0.0083 & 0.2571 & 0.2717 \\
        FA 4:0 OH   & 0.865 & 0.3707 & 0.287 & 0.4741 & 0.329 & 0.1099 & 0.1107 & 0.05 & 0.1672 & 0.039 & 0.006 & 0.0066 & 0.2382 & 0.2508 \\
        FA 12:0 OH-2   & 0.8614 & 0.3674 & 0.274 & 0.445 & 0.289 & 0.1369 & 0.408 & 0.36 & 0.4229 & 0.277 & 0.0062 & 0.0072 & 0.247 & 0.2605 \\
        FA 8:0 NH2 II (2-aminooctanoate)   & 0.8963 & 0.3707 & 0.284 & 0.4791 & 0.327 & 0.1469 & 0.7845 & 0.69 & 0.7925 & 0.687 & 0.0064 & -0.0093 & 0.2704 & 0.2675 \\
        FA 7:0 OH-1   & 0.975 & 0.3848 & 0.297 & 0.4943 & 0.334 & 0.1484 & 0.1701 & 0.167 & 0.1988 & 0.12 & 0.0055 & 0.001 & 0.2608 & 0.2673 \\
        FA 18:0 OH-2   & 0.8965 & 0.3706 & 0.277 & 0.477 & 0.327 & 0.1489 & 0.9017 & 0.826 & 0.9022 & 0.812 & 0.0065 & -0.001 & 0.249 & 0.2545 \\
        FA 18:1-2   & 0.9181 & 0.3703 & 0.271 & 0.473 & 0.322 & 0.173 & 0.966 & 0.922 & 0.9677 & 0.935 & 0.0059 & 0.0037 & 0.253 & 0.2626 \\
        FA 7:0 OH-2   & 0.9487 & 0.3829 & 0.294 & 0.4868 & 0.326 & 0.1871 & 0.1396 & 0.147 & 0.18 & 0.118 & 0.0053 & 0.0044 & 0.2576 & 0.2674 \\
        FA 10:0 DiC (sebacic acid)   & 0.8602 & 0.3666 & 0.279 & 0.461 & 0.316 & 0.2086 & 0.5699 & 0.523 & 0.5932 & 0.517 & 0.006 & 0.0003 & 0.2641 & 0.2704 \\
        FA 26:0 OH   & 0.8857 & 0.3702 & 0.27 & 0.4714 & 0.314 & 0.2125 & 0.5548 & 0.457 & 0.6191 & 0.503 & 0.0063 & -0.0082 & 0.2637 & 0.2618 \\
        FA 22:2   & 0.8521 & 0.3657 & 0.267 & 0.4559 & 0.301 & 0.2157 & 0.492 & 0.439 & 0.5338 & 0.405 & 0.0065 & -0.0002 & 0.264 & 0.2703 \\
        FA 8:0 (Caprylic acid)   & 0.9163 & 0.3671 & 0.274 & 0.4589 & 0.315 & 0.228 & 0.9661 & 0.939 & 0.9664 & 0.933 & 0.0054 & 0.0005 & 0.2546 & 0.2605 \\
        FA 24:4   & 0.8557 & 0.3711 & 0.277 & 0.4455 & 0.29 & 0.2311 & 0.2574 & 0.236 & 0.3798 & 0.265 & 0.0064 & 0.0231 & 0.2301 & 0.2597 \\
        FA 16:0   & 0.8275 & 0.3685 & 0.266 & 0.4442 & 0.283 & 0.2457 & 0.3376 & 0.25 & 0.3711 & 0.203 & 0.0062 & 0.027 & 0.2362 & 0.2694 \\
        FA 26:1 DiC   & 0.8709 & 0.3682 & 0.272 & 0.4522 & 0.298 & 0.267 & 0.6796 & 0.668 & 0.6828 & 0.641 & 0.0056 & -0.0008 & 0.2633 & 0.2682 \\
        FA 16:0 DiC   & 0.9074 & 0.3679 & 0.27 & 0.4591 & 0.31 & 0.2791 & 0.8009 & 0.713 & 0.8315 & 0.696 & 0.0059 & -0.0062 & 0.259 & 0.2587 \\
        FA 16:2   & 0.8627 & 0.3659 & 0.28 & 0.4806 & 0.336 & 0.2936 & 0.4084 & 0.545 & 0.547 & 0.43 & 0.0057 & 0.0037 & 0.241 & 0.2504 \\
        FA 18:0 OH-4   & 0.8369 & 0.369 & 0.27 & 0.4431 & 0.289 & 0.2954 & 0.0794 & 0.044 & 0.1506 & 0.042 & 0.0064 & 0.0257 & 0.2357 & 0.2678 \\
        FA 16:0 OH-2   & 0.9021 & 0.3644 & 0.259 & 0.4418 & 0.274 & 0.2984 & 0.1732 & 0.507 & 0.3126 & 0.138 & 0.0058 & 0.0006 & 0.2572 & 0.2636 \\
        FA 18:1 3   & 0.7949 & 0.3686 & 0.266 & 0.4498 & 0.281 & 0.3194 & 0.96 & 0.913 & 0.9617 & 0.923 & 0.0064 & 0.0043 & 0.2472 & 0.258 \\
        FA 20:3   & 0.8684 & 0.3694 & 0.266 & 0.448 & 0.286 & 0.3405 & 0.7011 & 0.556 & 0.7626 & 0.605 & 0.0061 & 0.0012 & 0.2689 & 0.2762 \\
        FA 28:1 DiC   & 0.9204 & 0.3673 & 0.27 & 0.4539 & 0.302 & 0.3444 & 0.5715 & 0.434 & 0.6395 & 0.49 & 0.005 & -0.0079 & 0.2753 & 0.2725 \\
        FA 22:1   & 0.8466 & 0.363 & 0.27 & 0.4494 & 0.294 & 0.3596 & 0.4225 & 0.322 & 0.4567 & 0.285 & 0.0065 & -0.0015 & 0.2628 & 0.2678 \\
        FA 16:0 OH-1   & 0.8698 & 0.3658 & 0.263 & 0.4442 & 0.287 & 0.3742 & 0.0241 & 0.384 & 0.0833 & 0.022 & 0.0059 & -0.0213 & 0.2793 & 0.2639 \\
        FA 14:0 OH   & 0.8632 & 0.367 & 0.269 & 0.4439 & 0.286 & 0.378 & 0.0562 & 0.374 & 0.135 & 0.024 & 0.006 & -0.004 & 0.2616 & 0.2635 \\
        FA 18:0 DiC   & 0.8934 & 0.3674 & 0.284 & 0.4771 & 0.33 & 0.4257 & 0.325 & 0.495 & 0.4088 & 0.254 & 0.0064 & -0.0167 & 0.2616 & 0.2513 \\
        FA 22:2 DiC   & 0.8524 & 0.3622 & 0.276 & 0.4647 & 0.32 & 0.4319 & 0.2823 & 0.516 & 0.4021 & 0.241 & 0.0062 & 0.0127 & 0.2334 & 0.2523 \\
        FA 22:0   & 0.8701 & 0.3699 & 0.274 & 0.457 & 0.3 & 0.4546 & 0.9477 & 0.912 & 0.9479 & 0.907 & 0.0069 & 0.0009 & 0.2502 & 0.258 \\
        FA 12:0 OH-1   & 0.8452 & 0.363 & 0.265 & 0.4446 & 0.283 & 0.455 & 0.0052 & 0.415 & 0.0376 & 0.006 & 0.006 & -0.0184 & 0.273 & 0.2606 \\
        FA 18:0 OH-5   & 0.8214 & 0.3679 & 0.268 & 0.4562 & 0.299 & 0.4581 & 0.7572 & 0.746 & 0.7605 & 0.7 & 0.0058 & 0.012 & 0.2649 & 0.2827 \\
        FA 5:0 OH   & 0.9523 & 0.3798 & 0.282 & 0.4741 & 0.32 & 0.4707 & 0.3712 & 0.331 & 0.3793 & 0.3 & 0.0051 & 0.0106 & 0.257 & 0.2727 \\
        FA 26:0   & 0.9114 & 0.3686 & 0.285 & 0.4772 & 0.328 & 0.4847 & 0.908 & 0.841 & 0.9128 & 0.833 & 0.0066 & -0.0044 & 0.2587 & 0.2609 \\
        FA 27:0 OH   & 0.8549 & 0.3695 & 0.313 & 0.4891 & 0.359 & 0.4848 & 0.6911 & 0.575 & 0.7098 & 0.576 & 0.0061 & 0 & 0.2538 & 0.2598 \\
        FA 26:2 DiC   & 0.8658 & 0.3653 & 0.266 & 0.4599 & 0.308 & 0.4851 & 0.4527 & 0.344 & 0.4786 & 0.295 & 0.0059 & 0.0191 & 0.2274 & 0.2524 \\
        FA 18:0   & 0.8571 & 0.3706 & 0.269 & 0.4485 & 0.286 & 0.4868 & 0.5221 & 0.44 & 0.5455 & 0.398 & 0.0065 & 0.016 & 0.2472 & 0.2697 \\
        FA 24:5   & 0.862 & 0.3714 & 0.275 & 0.4517 & 0.292 & 0.5288 & 0.2587 & 0.194 & 0.3673 & 0.205 & 0.0064 & 0.0075 & 0.2528 & 0.2667 \\
        FA 20:0   & 0.8488 & 0.3713 & 0.265 & 0.4528 & 0.285 & 0.5447 & 0.3473 & 0.244 & 0.4282 & 0.239 & 0.0068 & 0.0121 & 0.2458 & 0.2648 \\
        FA 16:2   & 0.8739 & 0.368 & 0.263 & 0.4432 & 0.283 & 0.5678 & 0.4273 & 0.319 & 0.4623 & 0.279 & 0.0058 & 0.0066 & 0.2582 & 0.2707 \\
        FA 16:0 OH-3   & 0.8183 & 0.3639 & 0.264 & 0.4499 & 0.297 & 0.5938 & 0.6499 & 0.522 & 0.6982 & 0.561 & 0.0066 & -0.0023 & 0.2466 & 0.251 \\
        FA 26:0 OH   & 0.8885 & 0.375 & 0.279 & 0.4703 & 0.313 & 0.6423 & 0.6727 & 0.553 & 0.7085 & 0.551 & 0.0059 & 0.0049 & 0.2422 & 0.253 \\
        FA 20:1   & 0.8424 & 0.3681 & 0.267 & 0.4467 & 0.284 & 0.6444 & 0.4655 & 0.338 & 0.5401 & 0.358 & 0.0067 & 0.003 & 0.2533 & 0.2629 \\
        FA 18:1 DiC   & 0.9391 & 0.3729 & 0.28 & 0.4639 & 0.316 & 0.6964 & 0.2476 & 0.195 & 0.3583 & 0.216 & 0.0063 & 0.0064 & 0.242 & 0.2547 \\
        FA 22:3   & 0.7937 & 0.3629 & 0.259 & 0.4484 & 0.286 & 0.7088 & 0.499 & 0.56 & 0.5531 & 0.394 & 0.0062 & 0.0014 & 0.2498 & 0.2573 \\
        FA 22:5   & 0.8716 & 0.37 & 0.263 & 0.4422 & 0.281 & 0.7252 & 0.7501 & 0.708 & 0.7557 & 0.637 & 0.006 & 0.0055 & 0.2618 & 0.2733 \\
        FA 8:0 NH2 (2-aminooctanoate)   & 0.8284 & 0.3606 & 0.256 & 0.4358 & 0.277 & 0.7264 & 0.7986 & 0.694 & 0.8079 & 0.685 & 0.0064 & 0.0012 & 0.2457 & 0.2533 \\
        FA 12:0 NH-2   & 0.8217 & 0.3677 & 0.263 & 0.4494 & 0.286 & 0.734 & 0.4512 & 0.323 & 0.5243 & 0.327 & 0.0056 & -0.0181 & 0.2733 & 0.2608 \\
        FA 28:2 DiC   & 0.8697 & 0.3637 & 0.268 & 0.4576 & 0.308 & 0.7548 & 0.1317 & 0.411 & 0.2226 & 0.106 & 0.0051 & -0.0019 & 0.2451 & 0.2482 \\
        FA 22:4   & 0.8701 & 0.3662 & 0.264 & 0.4408 & 0.279 & 0.763 & 0.8562 & 0.786 & 0.8599 & 0.774 & 0.0064 & 0.0096 & 0.253 & 0.269 \\
        FA 4:0 OH   & 0.8551 & 0.3703 & 0.281 & 0.4705 & 0.325 & 0.8 & 0.8553 & 0.816 & 0.8559 & 0.818 & 0.0063 & -0.0028 & 0.2494 & 0.2529 \\
        FA 20:0 DiC   & 0.9186 & 0.3655 & 0.293 & 0.4786 & 0.335 & 0.8057 & 0.4792 & 0.614 & 0.5131 & 0.431 & 0.006 & -0.0032 & 0.2449 & 0.2476 \\
        FA 16:4   & 0.8851 & 0.3648 & 0.267 & 0.4576 & 0.3 & 0.8502 & 0.9347 & 0.889 & 0.9352 & 0.879 & 0.0058 & 0.051 & 0.189 & 0.2457 \\
        FA 20:0 DiC   & 0.9213 & 0.3614 & 0.285 & 0.4749 & 0.328 & 0.8885 & 0.7124 & 0.629 & 0.7296 & 0.577 & 0.0057 & -0.0134 & 0.2625 & 0.2548 \\
        FA 22:4   & 0.8859 & 0.3695 & 0.261 & 0.4413 & 0.279 & 0.9208 & 0.3759 & 0.249 & 0.446 & 0.243 & 0.006 & 0.0139 & 0.2472 & 0.2671 \\
        FA 27:0 OH   & 0.8394 & 0.3709 & 0.275 & 0.4665 & 0.311 & 0.9291 & 0.7611 & 0.666 & 0.8026 & 0.666 & 0.0063 & -0.0022 & 0.2461 & 0.2503 \\
        FA 18:1-1   & 0.8515 & 0.3681 & 0.27 & 0.4461 & 0.285 & 0.9343 & 0.9046 & 0.882 & 0.9048 & 0.884 & 0.0061 & 0.0203 & 0.2467 & 0.273 \\
        FA 26:1 DiC   & 0.8519 & 0.3631 & 0.27 & 0.4595 & 0.31 & 0.9428 & 0.1145 & 0.332 & 0.1807 & 0.099 & 0.0059 & -0.0092 & 0.2613 & 0.2579 \\
        FA 22:3   & 0.8597 & 0.3697 & 0.268 & 0.4453 & 0.284 & 0.9717 & 0.2745 & 0.226 & 0.4378 & 0.284 & 0.0064 & 0.0124 & 0.2475 & 0.2663 \\
        FA 21:0 DiC   & 0.9253 & 0.3666 & 0.287 & 0.4839 & 0.34 & 0.9825 & 0.7508 & 0.784 & 0.7583 & 0.702 & 0.0059 & -0.0008 & 0.24 & 0.245 \\
        FA 15:0 DiC   & 0.8879 & 0.3692 & 0.272 & 0.4544 & 0.302 & 0.9827 & 0.512 & 0.375 & 0.6197 & 0.426 & 0.0062 & 0.0005 & 0.2456 & 0.2523 \\
        \bottomrule
    \end{tabular}}
\end{table}

\end{document}